\documentclass[12pt]{elsarticle}

\expandafter\let\csname equation*\endcsname\relax
\expandafter\let\csname endequation*\endcsname\relax

\makeatletter
\def\ps@pprintTitle{%
   \let\@oddhead\@empty
   \let\@evenhead\@empty
   \def\@oddfoot{\reset@font\hfil\thepage\hfil}
   \let\@evenfoot\@oddfoot
}
\makeatother

\usepackage{amsmath}
\usepackage{ulem}
\usepackage{xspace}

\usepackage{xparse}
\usepackage{diagbox}
\usepackage{amsthm}
\usepackage{amssymb}
\usepackage{commath}              
\usepackage{graphicx}
\usepackage[headheight=15pt]{geometry}
\usepackage{setspace}
\usepackage{color}
\usepackage{xparse}
\usepackage{algorithm,algpseudocode}
\usepackage[caption=false]{subfig}
\usepackage{endnotes}
\usepackage[page,titletoc]{appendix}	
\usepackage[inline]{enumitem}
\usepackage{dsfont}
\usepackage{cases}
\usepackage{mathtools}
\usepackage{url}
\usepackage{tikz}
\usetikzlibrary{matrix,decorations.pathreplacing,shapes,arrows,positioning}
\usetikzlibrary{calc,trees,positioning,arrows,chains,shapes.geometric,%
    decorations.pathreplacing,decorations.pathmorphing,shapes,%
    matrix,shapes.symbols}

\def\signmatCn{\Gamma}
\def\evCn{u_{\Gamma}}

\newcommand\iInm{i \in [m]}

\newcommand\Riitilde{\tilde{R}_{ii}}

\newcommand\Rii{R_{ii}}
\newcommand\Rjj{R_{jj}}

\newcommand{\dotprod}[2]{\left\langle #1,#2 \right\rangle}

\definecolor{darkpastelgreen}{rgb}{0.01, 0.75, 0.24}

\DeclarePairedDelimiter\floor{\lfloor}{\rfloor}

\DeclareMathOperator*{\argmax}{arg\,max} 
\DeclareMathOperator*{\argmin}{arg\,min} 

\DeclareMathOperator{\Span}{span}
\DeclareMathOperator{\sign}{sign}

\newcommand*{\argmaxl}{\argmax\limits}

\let\conjugatet\overline

\newlist{inlinelist}{enumerate*}{1}
\setlist*[inlinelist,1]{%
	label=(\roman*),
}

\newtheorem{theorem}{Theorem}[section]
\newtheorem{lemma}[theorem]{Lemma}
\newtheorem{corollary}[theorem]{Corollary}

\newcommand\planeRi{\mathcal{T}_{R_{i}}}
\newcommand\planeRj{\mathcal{T}_{R_{j}}}

\newcommand\hatPi{\hat{P}_{R_i}}
\newcommand\hatPj{\hat{P}_{R_j}}
\newcommand\hatPk{\hat{P}_{R_k}}

\newcommand\hatPsi{\hat{P}_{g_{n}^s R_i}}

\newcommand\RjgsRj{R_j^T g_{n}^{s} R_j}
\newcommand\RigsRi{R_i^T g_{n}^{s} R_i}

\newcommand\RigsiRi{R_i^T g_{n}^{s_{i}} R_i}
\newcommand\RigsijRj{R_i^T g_{n}^{s_{ij}} R_j}

\newcommand\RigsRj{R_i^T g_{n}^{s} R_j}

\newcommand\Rij{R_{ij}}

\newcommand\RijsThetaHat{R_{ij}\left( \theta_k, s \right)}

\newcommand\planeRsi{\mathcal{T}_{g_{n}^s R_{i}}}

\newcommand\planeRsj{\mathcal{T}_{g_{n}^s R_{j}}}

\newcommand\vitrans{v_{i}^{T}}
\newcommand\vjtrans{v_{j}^{T}}

\newcommand\Riitrans{ R_{ii}^T}
\newcommand\Rjjtrans{ R_{jj}^T}
\newcommand\RigRi{R_i^T g_{n} R_i}

\newcommand\RjgRj{R_j^T g_{n} R_j}

\newcommand\RignRi{R_i^T g_{n}^{n-1} R_i}

\newcommand\RjgsjRj{R_j^T g_{n}^{s_{j}} R_j}

\newcommand\vivj{v_{i} v_{j}^{T}}
\newcommand\vjvk{v_{j} v_{k}^{T}}
\newcommand\vivk{v_{i} v_{k}^{T}}
\newcommand\vivi{v_{i} v_{i}^{T}}

\newcommand\vij{v_{ij}}
\newcommand\vjk{v_{jk}}
\newcommand\vik{v_{ik}}
\newcommand\vii{v_{ii}}

\newcommand\vjitrans{v_{ji}^{T}}
\newcommand\vijtrans{v_{ij}^{T}}

\newcommand\viij{v_{ii}^j}
\newcommand\vjji{v_{jj}^i}
\newcommand\viijast{v_{ii}^{j\ast}}

\newcommand\vkl{v_{kl}}

\newcommand\blueJ{\textcolor{blue}{\boldsymbol{J}}}

\newcommand{\angstrom}{\textup{\AA}}

\newcommand\ijk{i<j<k \in [m]}
\newcommand\ijkInm{i<j<k \in [m]}
\newcommand\ijInm{i<j \in [m]}

\tikzset{
>=stealth',
  punktchain/.style={
    rectangle,
    rounded corners,
    draw=black, very thick,
    text width=15em,
    minimum height=3em,
    text centered,
    on chain},
  line/.style={draw, thick, <-},
  element/.style={
    tape,
    top color=white,
    bottom color=blue!50!black!60!,
    minimum width=8em,
    draw=blue!40!black!90, very thick,
    text width=10em,
    minimum height=3.5em,
    text centered,
    on chain},
  every join/.style={->, thick,shorten >=1pt},
  decoration={brace},
  tuborg/.style={decorate},
  tubnode/.style={midway, right=2pt},
}

\begin{document}

\begin{frontmatter}
\title{A common lines approach for ab-initio modeling of cyclically-symmetric molecules}

\author{Gabi Pragier and Yoel Shkolnisky}

\address{Department of Applied Mathematics, School of Mathematical Sciences, Tel-Aviv University, Israel}
\ead{gabipragier@gmail.com, yoelsh@tauex.tau.ac.il}

\begin{abstract}
One of the challenges in single particle reconstruction in cryo-electron microscopy is to find a three-dimensional model of a molecule
using its two-dimensional noisy projection-images. In this paper, we propose a robust ``angular reconstitution'' algorithm for molecules with $n$-fold cyclic symmetry, that estimates the orientation parameters of the projections-images. Our suggested method utilizes self common lines which
induce identical lines within the Fourier transform of each projection-image. We show that the location of self common lines admits quite a few favorable geometrical constraints, thus allowing to detect them even in a noisy setting. In addition, for molecules with higher order rotational symmetry, our proposed method exploits the fact that there exist numerous common lines between any two Fourier transformed projection-images of such molecules, thus allowing to determine their relative orientation even under high levels of noise.  The efficacy of our proposed method is demonstrated using numerical experiments conducted on simulated and experimental data.

\end{abstract}

\end{frontmatter}

\section{Introduction}\label{sec:introduction}

Cryo-electron microscopy is a method for determining the three-dimensional structure of a molecule from its two-dimensional projection-images~\cite{Frank2006}. The method consists of generating projection-images of copies of the investigated molecule, where each copy assumes a random unknown orientation before being imaged. Formally, if we denote by $\psi \colon \mathbb{R}^3 \to \mathbb{R}$ the electrostatic potential of the molecule, and consider a rotation matrix
\begin{equation}\label{eq:R_i}
SO\left(3\right) \ni R_i =
\left(
\begin{array}{cccc}
\vrule  &  \vrule   & \vrule\\
R_{i}^{(1)} &   R_{i}^{(2)} &  R_{i}^{(3)} \\
\vrule  &  \vrule   & \vrule
\end{array}
\right),
\end{equation}
then the projection-image $P_{R_{i}}$ is given by the line integrals of $\psi$ along the beaming-direction $R_i^{(3)}$. That is,
\begin{equation}\label{eq:psi}
P_{R_i} \left( x,y \right) = \int_{-\infty}^{\infty} \psi \left(R_i r\right)\,\mathrm{d}z = \int_{-\infty}^{\infty} \psi (x R_{i}^{(1)} + y R_{i}^{(2)} + z R_{i}^{(3)})\,\mathrm{d}z, \quad r=\left (x,y,z \right )^{T}.
\end{equation}
In this work, we focus on molecules that have an $n$-fold, $n \geq 2$, rotational symmetry about some unknown axis. Such molecules are referred to as molecules with $C_{n}$ symmetry. {We assume that} the cyclic symmetry order $n$ of the underlying molecule is known from prior knowledge, or may be inferred from rotational invariants computed by spherical harmonics (see e.g.,~\cite{Kam}). Intuitively, molecules with $C_{n}$ symmetry ``look exactly the same'' when rotated by $2\pi s/n, \ s \in [n-1]$~(we subsequently denote by~$[m]$ the set $\left\{ 1,\ldots,m \right\}$), radians about their axis of symmetry. Mathematically, it means that the electrostatic potential function $\psi \colon \mathbb{R}^3 \to \mathbb{R}$ of any such molecule satisfies
\begin{equation}\label{eq:psi_symm}
\psi \left(r\right)
=
\psi \left(g_{n} r\right)
= \ldots =
\psi \left(g_{n}^{n-1} r\right), \quad \forall r \in \mathbb{R}^3,
\end{equation}
where $g_{n} \in SO\left(3\right)$ represents a rotation of $2\pi/n$ radians about the unknown axis of symmetry, and $n \in \mathbb{N}$ is the largest for which~\eqref{eq:psi_symm} holds. Since rotating the three-dimensional coordinate system has no effect on the three-dimensional structure of the molecule, it may be assumed without loss of generality that the axis of symmetry coincides with the $z$-axis. Thus, the matrix $g_{n}$ which satisfies~\eqref{eq:psi_symm} may be written as
\begin{equation}\label{def:g}
g_{n} =
\left(
\begin{array}{rrr}
\cos(2\pi/n) & -\sin(2\pi/n) & 0\\
\sin(2\pi/n) & \cos(2\pi/n) & 0\\
0 &   0 & 1\\
\end{array}
\right),
\end{equation}
and any three-dimensional coordinate system for the molecule must be set such that it keeps the axis of symmetry aligned with the $z$-axis (with possibly flipping its orientation). That is, the degrees of freedom in setting the three-dimensional coordinate system of the molecule consist of 
in-plane rotations about the $z$-axis, and an in-plane rotation of $\pi$ radians about either the $x$-axis or the $y$-axis.

Finding the three-dimensional structure of the molecule amounts to recovering the unknown electrostatic potential function $\psi$ given only the projection-images $P_{R_{i}}$. Typically, this is done by first finding the rotation matrices $R_{i}$, followed by a standard tomographic inversion algorithm e.g.,~\cite{Herman09,ReconMethods01}.

A fundamental limitation of cryo-electron microscopy is the handedness ambiguity~\cite{sync}, whereby the best one may expect is to recover either the set $\left\{R_{i}\right\}_{i=1}^{m}$ or the set $\left\{ R_{i} J \right\}_{i=1}^{m}$ where $J = \operatorname{diag}(-1,-1,1)$, not being able to distinguish between the two.
Indeed, given any molecule whose electrostatic potential function is~$\psi$, consider the molecule whose electrostatic potential function $\tilde{\psi}$ is given by $\tilde{\psi}(r)= \psi (-r), \ r \in \mathbb{R}^3$. For any $R_{i}$, the projection-image $\tilde{P}_{ R_i J}$ of the molecule $\tilde{\psi}$ is given according to~\eqref{eq:psi} by
\begin{equation}\label{eq:handedness_2}
\tilde{P}_{ R_i J} \left( x,y \right)
=
\int_{-\infty}^{\infty} \tilde{\psi} \left( R_i J r \right)\,\mathrm{d}z
=
\int_{-\infty}^{\infty} \psi \left( - R_i J r \right)\,\mathrm{d}z,
\end{equation}
where $r = \left( x, y, z \right)^{T} \in \mathbb{R}^3$. Noting that~$-J r = \left( x,y,-z \right)$, and by changing the variable $z$ to $z' = -z$ we have
\begin{equation}\label{eq:handedness_3}
\int_{-\infty}^{\infty} \psi \left( - R_i J r \right)\,\mathrm{d}z
=
\int_{-\infty}^{\infty} \psi ( R_i \left( x,y,z' \right)^{T} )\,\mathrm{d}z'
=
P_{R_{i}} \left( x,y \right).
\end{equation}
As such, both $\left\{ R_{i} \right\}_{i=1}^{m}$ and $\left\{ R_{i} J \right\}_{i=1}^{m}$ are consistent with the same set of projection-images $\left\{ P_{R_{i}} \right\}_{i=1}^{m}$, yet the reconstructed model using the latter set is biologically infeasible, and the true set may therefore only be determined by visual examination of the reconstructed model, possibly exploiting other known structural information.
Furthermore, an important property of molecules with $C_{n}$ symmetry is that any $n$ projection-images $P_{R_{i}}, P_{g_{n} R_{i}},\ldots,P_{g_{n}^{n-1} R_{i}}, \ \iInm$, are identical. Indeed, from~\eqref{eq:psi_symm}
\begin{equation}\label{eq:psi_g_equal}
\psi \left(R_{i} r\right)
=
\psi \left(g_{n} R_{i} r\right)
=
\ldots
=
\psi \left(g_{n}^{n-1} R_{i} r\right) \quad \forall r \in \mathbb{R}^3.
\end{equation}
Thus, by letting $r=\left (x,y,z \right )^{T}$, and integrating over $z$, it follows from~\eqref{eq:psi} that for any $s \in [n-1]$,
\begin{equation}\label{eq:identical_ims}
P_{g_{n}^s R_{i}}\left(x,y\right)
=
\int_{-\infty}^{\infty} \psi \left(g_{n}^s R_{i} r\right)\,\mathrm{d}z
=
\int_{-\infty}^{\infty} \psi \left(R_{i} r\right)\,\mathrm{d}z
=
P_{R_{i}}\left(x,y\right).
\end{equation}
Equations~\eqref{eq:handedness_2},~\eqref{eq:handedness_3}, and~\eqref{eq:identical_ims} show that for each $R_{i}$ in the forward imaging model~\eqref{eq:psi} there exists an equivalence class $\mathcal{R}_{i} =\left \{O g_{n}^{s_{i}} R_{i} J^{\delta}\right \}$, with $s_{i}\in[n]$, $O\in SO(3)$, and $\delta\in\{ 0,1 \}$, such that $P_{R_{i}}= P_{R_{i}^{*}}$ for any $R_{i}^{*}\in\mathcal{R}_{i}$. The matrices $g_{n}^{s_{i}}$ are due to the ambiguity stated in~\eqref{eq:psi_g_equal} and so $s_{i}$ may be chosen independently for each image $i=1,\ldots,m$, while $\delta$ and $O \in SO(3)$ are common to all images, with $O$ being some in-plane rotation matrix about the $z$-axis (the axis of symmetry). In other words, there are no ``true'' rotation matrics $R_{i}$ that need to be recovered, as may be implied from~\eqref{eq:psi}, but rather we need to recover for each $i=1,\ldots,m$ any member of the equivalence class $\mathcal{R}_{i}$.
In summary, in light of the above, given projection-images $\left\{ P_{R_{i}} \right\}_{i=1}^m$ of a molecule with $C_{n}$ symmetry, the goal is to recover $m$ rotation matrices $\left \{O g_{n}^{s_{i}} R_{i} J^{\delta}\right \}_{i=1}^{m}$, where $s_{i}\in[n]$, $O\in SO(3)$, and $\delta\in\{ 0,1 \}$, such that~\eqref{eq:psi} is satisfied for all $\iInm$ for some $\psi \colon \mathbb{R}^3 \to \mathbb{R}$.

The rest of this paper is organized as follows. In Sections~\ref{sec:common_lines} and~\ref{sec:self_common_lines} we review the projection slice theorem~\cite{Natr2001a} and further describe the properties of common lines and self common lines which are induced by this theorem. In Section~\ref{sec:previous_work} we review some previous work. In Section~\ref{sec:outline_cn} we present an outline of our proposed method. Next, in Section~\ref{sec:third_row_cn}, we describe how to estimate the third row of each rotation matrix $R_{i}$. A procedure that assures that all these estimates correspond to a single hand (see Section~\ref{sec:introduction} which points out the inherent handedness ambiguity in cryo-electron microscopy) is presented in Section~\ref{sec:handedness_sync_cn}. In Section~\ref{sec:in_plane_cn} we then describe a method to determine the remaining first two rows of each $R_{i}$. Then, in Section~\ref{sec:Numerical_experiments_cn} we report some numerical experiments that we conducted on simulated and experimental datasets that show the efficacy of our proposed method. Finally, in Section~\ref{sec:summary} we present some conclusions and possible extensions of this work.

As the cases of $C_{3}$ and $C_{4}$ symmetry exhibit special geometry which is advantageous from the computational point-of-view, we describe in Appendix~\ref{sec:third_row_c2_c3_c4} an algorithm for estimating the third row of each rotation matrix $R_{i}$ for molecules with such symmetries. This algorithm may replace the algorithm of Section~\ref{sec:third_row_cn} for such symmetries, as it is much faster and is empirically observed to be more robust to noise. Unfortunately, generalizations of the geometry derived in Appendix~\ref{sec:third_row_c2_c3_c4} for $C_n$ with $n>4$ are of little practical use.

\section{Common lines}\label{sec:common_lines}
The projection slice theorem~\cite{Natr2001a} is a key theorem underpinning many of the methods for recovering the rotation matrices of a given set of projection-images. It states that the two-dimensional Fourier transform of any projection-image $P_{R_{i}}$ is equal to the restriction of the three-dimensional Fourier transform of the electrostatic potential function $\psi$ to the central plane $\planeRi = \Span \{ R_{i}^{(1)}, R_{i}^{(2)} \} \subset \mathbb{R}^3$, whose normal coincides with $R_{i}^{(3)}$ (the third column of $R_{i}$ of~\eqref{eq:R_i}). Mathematically, if we denote by $\hat{\psi}$ the three-dimensional Fourier transform of $\psi$, and denote by $\hatPi$ the two-dimensional Fourier transform of the projection-image $P_{R_i}$, then the projection slice theorem states that
\begin{equation}\label{eq:PST}
\hatPi\left(\omega_{x},\omega_{y}\right) = \hat{\psi}\left(R_{i}\omega\right), \quad \forall \omega = \left(\omega_{x},\omega_{y},0\right)^T \in \mathbb{R}^3.
\end{equation}
As such, since any two central planes $\planeRi$ and $\planeRj$ which are not parallel intersect along a common axis, it follows that any two Fourier-transformed images $\hatPi$ and $\hatPj$ have a pair of lines (one line in each image) on which their values agree. Such pairs of lines are called ``common lines''. Hereafter we refer to $\hatPi$ and $\hatPj$ simply as ``images''. The following lemma shows that any two images of a molecule with $C_n$ symmetry have $n$ pairs of common lines (see Figure~\ref{fig:four_common_lines} for an illustration of the case $n=4$).

\begin{figure}
	\centering
	\begin{tikzpicture}
	\filldraw[color=red!60, fill=red!5, line width=1mm](6,2) circle [radius=1];
	\node at (4.5,2) {$\hat{P}_{R_j}$};
	\draw[color=violet,line width=1mm] (6,2) -- ++(270:1);
	\draw[color=black,line width=1mm,dash pattern=on 2pt off 3pt on 4pt off 4pt] (6,2) -- ++(70:1);
	\draw[color=blue,line width=1mm,loosely dashed] (6,2) -- ++(160:1);
	\draw[color=brown,line width=1mm,densely dotted] (6,2) -- ++(180:1);
	
	\filldraw[color=green!60, fill=green!5, line width=1mm](3,4)  circle [radius=1];
	\node at (1.5,4) {$\hat{P}_{R_i}$};
	\draw[color=violet,line width=1mm] (3,4) -- ++(140:1);
	
	\filldraw[color=green!60, fill=green!5, line width=1mm](9,4)  circle [radius=1];
	\node at (7.5,4) {$\hat{P}_{g_{4}R_i}$};
	\draw[color=black,line width=1mm,dash pattern=on 2pt off 3pt on 4pt off 4pt] (9,4) -- ++(0:1);
	
	\filldraw[color=green!60, fill=green!5, line width=1mm](3,0)  circle [radius=1];
	\node at (1.5,0) {$\hat{P}_{g_{4}^2R_i}$};
	\draw[color=blue,line width=1mm,loosely dashed] (3,0) -- ++(30:1);
	
	\filldraw[color=green!60, fill=green!5, line width=1mm](9,0)  circle [radius=1];
	\node at (7.5,0) {$\hat{P}_{g_{4}^3R_i}$};
	\draw[color=brown,line width=1mm,densely dotted] (9,0) -- ++(300:1);
	
	\draw[color=black,line width=1mm,dash pattern=on 2pt off 3pt on 4pt off 4pt] (3,4) -- ++(0:1);
	\draw[color=blue,line width=1mm,loosely dashed] (3,4) -- ++(30:1);
	\draw[color=brown,line width=1mm,densely dotted] (3,4) -- ++(300:1);
	\end{tikzpicture}
	\caption{An illustration of Lemma~\ref{lemma:n_common_lines} for the case of $C_{4}$. The red circle corresponds to some image $\hatPj$. The four green circles correspond to some other four identical images $\hat{P}_{R_i}$, $\hat{P}_{g_{4} R_i}$, $\hat{P}_{g_{4}^2 R_i}$, $\hat{P}_{g_{4}^3 R_i}$. The lines in $\hatPi$ and $\hatPj$ represent their four common lines, each of which is induced by a common line between $\hatPj$ and one of the images $\hat{P}_{R_i}$, $\hat{P}_{g_{4} R_i}$, $\hat{P}_{g_{4}^2 R_i}$, $\hat{P}_{g_{4}^3 R_i}$.
	}
	\label{fig:four_common_lines}
\end{figure}
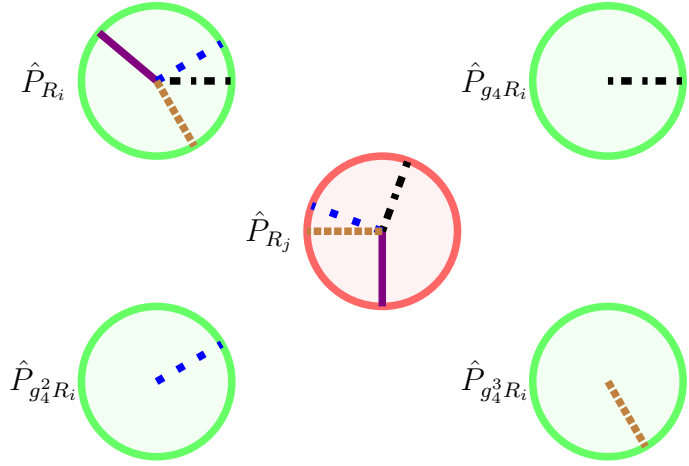

\begin{lemma}\label{lemma:n_common_lines}
	For any two images $\hatPi$ and $\hatPj$, such that $\abs{\dotprod{R_{i}^{(3)}}{R_{j}^{(3)}}} \ne 1$, there exist angles $\alpha_{ij}^{(s)}, \alpha_{ji}^{(s)} \in [0,2 \pi)$, $s = 0,\ldots,n-1$, such that for any $\xi \in \mathbb{R}$,
\begin{equation}\label{eq:clm_four_pairs_cls}
\hatPi\left(\xi\cos\alpha_{ij}^{(s)},\xi\sin\alpha_{ij}^{(s)}\right)
=
\hatPj\left(\xi\cos\alpha_{ji}^{(s)},\xi\sin\alpha_{ji}^{(s)}\right).
\end{equation}
\end{lemma}

\begin{proof}
Denote by $q_{ij}^{(s)} \in \mathbb{R}^3, \  s=0,\ldots,n-1$, the unit vector in the direction of the common axis between $\planeRi$ and $\planeRsj$, namely
\begin{equation}\label{eq:q_ij_s}
q_{ij}^{(s)} = \frac{R_{i}^{(3)} \times g_{n}^{s} R_{j}^{(3)}}{\left\|R_{i}^{(3)} \times g_{n}^{s} R_{j}^{(3)}\right\|}, \quad s=0,\ldots,n-1.
\end{equation}
Also denote for every $s=0,\ldots,n-1$
\begin{equation}\label{eq:alpha_ij_q_ij_s}
\begin{aligned}
&\left(\cos \alpha_{ij}^{(s)},\sin\alpha_{ij}^{(s)},0\right)^T
=
R_{i}^T q_{ij}^{(s)},\\
&\left(\cos\alpha_{ji}^{(s)},\sin\alpha_{ji}^{(s)},0\right)^T
=
\left( g_{n}^{s} R_{j} \right)^T q_{ij}^{(s)}.
\end{aligned}
\end{equation}
That is, the angles $\alpha_{ij}^{(s)}$ and $\alpha_{ji}^{(s)}$ express $q_{ij}^{(s)}$ in the frame of reference of $\planeRi$ and in the frame of reference $\planeRsj$, respectively. Now, by ~\eqref{eq:psi_g_equal},~\eqref{eq:PST} and~\eqref{eq:alpha_ij_q_ij_s}, for $\xi\in\mathbb{R}$ and $s=0,\ldots,n-1$,
\begin{align}
\hatPi\left(\xi \cos \alpha_{ij}^{(s)},\xi \sin\alpha_{ij}^{(s)} \right) &= \hat{\psi}\left( \xi q_{ij}^{(s)} \right) = \hat{\psi}\left( \xi g_{n}^{-s} q_{ij}^{(s)} \right) \label{eq:PST_use_symm}\\
&= \hat{\psi}\left( \xi  R_{j} R_{j}^{T} g_{n}^{-s} q_{ij}^{(s)} \right)
= \hat{\psi}\left( \xi  R_{j} (g_{n}^{s} R_{j})^{T} q_{ij}^{(s)} \right) \notag \\
&= \hatPj\left(\xi\cos\alpha_{ji}^{(s)},\xi\sin\alpha_{ji}^{(s)}\right). \notag
\end{align}
\end{proof}

The relative orientation between any two central planes $\planeRi$ and $\planeRsj$ is expressed algebraically by $\RigsRj$ (henceforth regarded as the relative orientation between $\planeRi$ and $\planeRsj$). In addition, denoting by $\gamma_{ij}^{(s)}, s \in [n]$, the acute angle between the central planes $\planeRi$ and $\planeRsj$, we get that $\left( \alpha_{ij}^{(s)}, \gamma_{ij}^{(s)}, -\alpha_{ji}^{(s)} \right)$ is the ``$x$-convention" Euler angles parameterization~\cite{Euler_param_Tuma_book} of~$\RigsRj$. Specifically,
\begin{equation}\label{eq:euler_angles}
\RigsRj
=
R_z(\alpha_{ij}^{(s)}) R_x(\gamma_{ij}^{(s)}) R_z(-\alpha_{ji}^{(s)}),
\end{equation}
where
\begin{equation}\label{def:R_x_R_z}
R_x(\theta)
=
\begin{pmatrix}
1 & 0 & 0 \\
0 & \cos \theta & -\sin \theta \\
0 & \sin \theta & \hphantom{+}\cos \theta
\end{pmatrix} \quad \text{and} \quad
R_z(\theta)
=
\begin{pmatrix}
\cos \theta & -\sin \theta & 0 \\
\sin \theta & \hphantom{+}\cos \theta & 0 \\
0 & 0 & 1
\end{pmatrix}
\end{equation}
denote the matrices that rotate vectors by $\theta \in \mathbb{R}$ radians about the $x$-axis and $z$-axis, respectively. Finally, we mention in passing that any $\alpha_{ij}^{(s)}$ and $\alpha_{ji}^{(s)}$ may be recovered from the entries of $\RigsRj$ using (see~\cite{Hadani})
\begin{equation}\label{eq:hadani}
\alpha_{ij}^{(s)} = \arctan \left(-\frac{\left( \RigsRj \right)_{1,3}}{\left( \RigsRj \right)_{2,3}}\right), \quad
\alpha_{ji}^{(s)} = \arctan \left(-\frac{\left( \RigsRj \right)_{3,1}}{\left( \RigsRj \right)_{3,2}}\right).
\end{equation}
\section{Self common lines}\label{sec:self_common_lines}
A ``self common line'' is a common line between any two images $\hatPi$ and $\hatPsi$ with $s \in [n-1]$. Thus, similarly to~\eqref{eq:clm_four_pairs_cls}, for any $\iInm$ there exist angles $\alpha_{ii}^{(s)}, \alpha_{gi}^{(s)} \in [0,2 \pi), \  s = 1,\ldots,n-1$ (the subscripts will be clarified shortly), such that for any $\xi \in \mathbb{R}$
\begin{equation*}
\hatPi\left(\xi\cos\alpha_{ii}^{(s)},\xi\sin\alpha_{ii}^{(s)}\right)
=
\hatPsi\left(\xi\cos\alpha_{gi}^{(s)},\xi\sin\alpha_{gi}^{(s)}\right).
\end{equation*}
Since by~\eqref{eq:identical_ims}, any two such images $\hatPi$ and $\hatPsi$ are identical, it follows that
\begin{equation}\label{eq:clm_four_pairs_scls}
\hatPi\left(\xi\cos\alpha_{ii}^{(s)},\xi\sin\alpha_{ii}^{(s)}\right)
=
\hatPi\left(\xi\cos\alpha_{gi}^{(s)},\xi\sin\alpha_{gi}^{(s)}\right), \quad \iInm, \ s \in [n-1].
\end{equation}
That is, any image $\hatPi$ has $n-1$ pairs of identical lines (regarded henceforth as self common lines as well).
Similarly to~\eqref{eq:alpha_ij_q_ij_s}, the angles $\alpha_{ii}^{(s)}$ and $\alpha_{gi}^{(s)}$ may be expressed in terms of the unit vector $q_{ii}^{(s)} \in \mathbb{R}^3$ of~\eqref{eq:q_ij_s} in the direction of the common axis between the central planes $\planeRi$ and $\planeRsi$ as
\begin{equation}\label{eq:alpha_ii_q_ii}
\begin{aligned}
&\left(\cos \alpha_{ii}^{(s)},\sin\alpha_{ii}^{(s)},0\right)^T
=
R_{i}^T q_{ii}^{(s)},\\
&\left(\cos\alpha_{gi}^{(s)},\sin\alpha_{gi}^{(s)},0\right)^T
=
\left( g_{n}^{s} R_{i} \right)^T q_{ii}^{(s)}.
\end{aligned}
\end{equation}
The following lemma shows that the direction of the $s$-th and the $(n-s)$-th pairs of self common lines in any image coincide.
\begin{lemma}\label{lemma:self_cls_coincide}
	For any $\iInm$, and for any $s \in [n-1]$,
	\begin{equation}\label{eq:self_clm_part1}
	\alpha_{gi}^{(s)} = \alpha_{ii}^{(n-s)} + \pi \mod{2 \pi}.
	\end{equation}
\end{lemma}
The proof of Lemma~\ref{lemma:self_cls_coincide} is given in Appendix~\ref{app:proof of lemma_self_cls_coincide}. Figure~\ref{fig:scl_c3} and Figure~\ref{fig:scl_c4} illustrate Lemma~\ref{lemma:self_cls_coincide} for $C_{3}$ and $C_{4}$ symmetries, respectively.
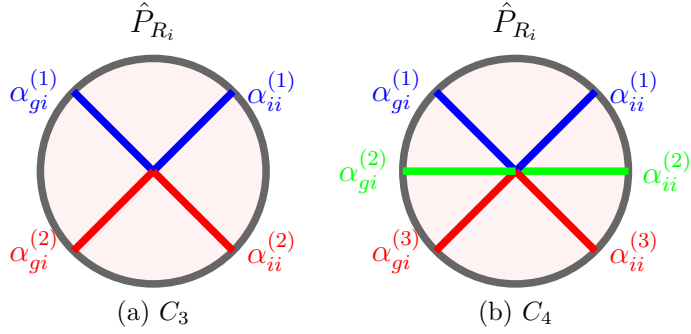
\begin{figure}
	\centering
	\subfloat[$C_3$]{
		\label{fig:scl_c3}
		\centering
		\begin{tikzpicture}
		\filldraw[color=black!60, fill=red!5, line width=1mm](0,0) circle [radius=1.5];
				\node at (0,2) {$\hat{P}_{R_i}$};
		
		\draw[color=blue,line width=1mm]  (0,0) -- ++(45:1.5) node (4) [black,right] {\textcolor{blue}{$\alpha_{ii}^{(1)}$}};
		\draw[color=blue,line width=1mm] (0,0) -- ++(135:1.5) node (2) [black,left] {\textcolor{blue}{$\alpha_{gi}^{(1)}$}};
		
		\draw[color=red,line width=1mm]  (0,0) -- ++(225:1.5) node (4) [black,left] {\textcolor{red}{$\alpha_{gi}^{(2)}$}};
		\draw[color=red,line width=1mm] (0,0) -- ++(315:1.5) node (2) [black,right] {\textcolor{red}{$\alpha_{ii}^{(2)}$}};
		\end{tikzpicture}
	}
	\subfloat[$C_4$]{
		\label{fig:scl_c4}
		\centering
		\begin{tikzpicture}
		\filldraw[color=black!60, fill=red!5, line width=1mm](0,0) circle [radius=1.5];
				\node at (0,2) {$\hat{P}_{R_i}$};
		\draw[color=green,line width=1mm]  (0,0) -- ++(0:1.5) node (1) [black,right] {\textcolor{green}{$\alpha_{ii}^{(2)}$}};
		
		\draw[color=blue,line width=1mm]  (0,0) -- ++(45:1.5) node (4) [black,right] {\textcolor{blue}{$\alpha_{ii}^{(1)}$}};
		\draw[color=blue,line width=1mm] (0,0) -- ++(135:1.5) node (2) [black,left] {\textcolor{blue}{$\alpha_{gi}^{(1)}$}};
		
		\draw[color=red,line width=1mm]  (0,0) -- ++(225:1.5) node (4) [black,left] {\textcolor{red}{$\alpha_{gi}^{(3)}$}};
		\draw[color=red,line width=1mm] (0,0) -- ++(315:1.5) node (2) [black,right] {\textcolor{red}{$\alpha_{ii}^{(3)}$}};
		
		\draw[color=green,line width=1mm]  (0,0) -- ++(180:1.5) node (1) [black,left] {\textcolor{green}{$\alpha_{gi}^{(2)}$}};
		\end{tikzpicture}
	}
	\caption{Circles representing an image $\hatPi$ of a molecule with $C_{3}$ symmetry (left) and of a molecule with $C_{4}$ symmetry (right), along with the pairs of self common lines in each image. By Lemma~\ref{lemma:self_cls_coincide} the lines are collinear.
	\label{fig:scl}}
\end{figure}
\begin{corollary}\label{cor:self_cls_half_of_pairs}
	By~\eqref{eq:clm_four_pairs_scls} and~\eqref{eq:self_clm_part1} it follows that for any $\iInm$, $s \in [n-1]$, and $\xi \in \mathbb{R}$,
	\begin{equation}\label{eq:scl_n-1_pairs}
	\begin{aligned}
	\hatPi\left(\xi\cos\alpha_{ii}^{(s)},\xi\sin\alpha_{ii}^{(s)}\right)
	&=
	\hatPi\left(\xi\cos\alpha_{gi}^{(s)},\xi\sin\alpha_{gi}^{(s)}\right)\\
	&=
	\hatPi \left( \xi \cos (\alpha_{ii}^{(n-s)}+\pi ), \xi \sin ( \alpha_{ii}^{(n-s)}+\pi ) \right).
	\end{aligned}
	\end{equation}
	In addition, as $\hatPi$ is conjugate-symmetric (since $P_{R_{i}}$ is real-valued~\eqref{eq:psi}),
	\begin{equation}\label{eq:bla2}
	\hatPi \left( \xi \cos (\alpha_{ii}^{(n-s)}+\pi ), \xi \sin ( \alpha_{ii}^{(n-s)}+\pi ) \right)
	=
	\conjugatet{
		\hatPi \left( \xi \cos \alpha_{ii}^{(n-s)}, \xi \sin \alpha_{ii}^{(n-s)} \right)},
	\end{equation}
	where $\conjugatet{\left( \cdot \right)}$ denotes complex conjugation. As a result,
	\begin{equation}\label{eq:self_clm_part2}
	\hatPi \left( \xi \cos \alpha_{ii}^{(s)}, \xi \sin \alpha_{ii}^{(s)} \right)
	=
	\conjugatet{
		\hatPi \left( \xi \cos \alpha_{ii}^{(n-s)}, \xi \sin \alpha_{ii}^{(n-s)} \right)}, \quad \forall \xi \in \mathbb{R}.
	\end{equation}
\end{corollary}

We therefore see from Lemma~\ref{lemma:self_cls_coincide} and Corollary~\ref{cor:self_cls_half_of_pairs} that the $n-1$ pairs of self common lines in any image consist in fact of $\floor*{\frac{n-1}{2}}$ different pairs. In addition, in case $n$ is even, the two self common lines that constitute the $\frac{n}{2}$-th pair of lines are collinear. In particular, for molecules with $C_{2}$ symmetry, the sole pair of self common lines in any image consists of two collinear lines. In contrast,
\begin{itemize}
	\item for molecules with $C_{3}$ symmetry (i.e., $n=3$), the two pairs of self common lines (i.e., four lines altogether) in any image coincide (i.e., there are merely two non-collinear lines), see Figure~\ref{fig:scl_c3},
	\item for molecules with $C_{4}$ symmetry (i.e., $n=4$), one of the pairs of self common lines in every image consists of collinear lines whereas the other two remaining pairs of self common lines coincide, see Figure~\ref{fig:scl_c4}, and
	\item for molecules with cyclic symmetry of higher order (i.e., $C_{n}, \ n>4$), each image has in general $\floor*{\frac{n-1}{2}} > 1$ pairs of non-collinear self common lines.
\end{itemize}

In what follows, we refer to any relative orientation $\RigsRi, \ \iInm, \ s \in [n-1]$, as the self relative orientation between $\planeRi$ and $\planeRsi$. By~\eqref{eq:euler_angles}, any such self relative orientation is parameterized by the ordered triplet $\left( \alpha_{ii}^{(s)}, \gamma_{ii}^{(s)}, -\alpha_{gi}^{(s)} \right)$ which by Lemma~\ref{lemma:self_cls_coincide} is equal to $\left( \alpha_{ii}^{(s)}, \gamma_{ii}^{(s)}, -\alpha_{ii}^{(n-s)} - \pi \right)$. As such,
\begin{equation}\label{eq:euler_angles_self}
\RigsRi
=
R_z(\alpha_{ii}^{(s)}) R_x(\gamma_{ii}^{(s)}) R_z(-\alpha_{ii}^{(n-s)}-\pi),
\end{equation}
where, by~\eqref{eq:hadani},
\begin{equation}\label{eq:hadani_self}
\alpha_{ii}^{(s)} = \arctan \left(-\frac{\left( \RigsRi \right)_{1,3}}{\left( \RigsRi \right)_{2,3}}\right), \quad s=1,\ldots,n-1.
\end{equation}
\section{Previous work}\label{sec:previous_work}
The method of angular reconstitution~\cite{Goncharov1988,VanHeel87} recovers the orientations corresponding to a given set of projection-images in a sequential manner. Specifically, it first determines the common lines between arbitrary three images. This establishes a coordinate system, and the orientations of all remaining projection-images are then recovered one after the other by using the common lines between any given image and those three images.
However, due to the sequential nature of angular reconstitution, the method critically depends on correctly detecting the common lines between the first three images. As a result, when the input projection-images are noisy, this positioning might be wrong, which would lead to errors when inferring the relative orientations of the remaining central planes.

In~\cite{3n}, a non-sequential common-lines-based method was suggested, which estimates simultaneously the orientations of all projection-images. As such, contrary to the method of angular reconstitution, it does not suffer from accumulation of errors and is therefore robust even when the input images are noisy. However, this method may not be applied to projection-images of a molecule with $C_n$ symmetry, since there is no way to guarantee a consistent choice of common lines, and so for each pair of images $\hatPi$ and $\hatPj$ we can only estimate $\RigsijRj$ for an unknown $s_{ij}$, whereas to apply the method in~\cite{3n} all $s_{ij}$ must be the same. Unfortunately, such a consistency in all $s_{ij}$ cannot be guaranteed using the method of~\cite{3n}.

In~\cite{Cheng_thesis}, a method which successfully handles molecules with $C_{2}$ symmetry was proposed. The method is based on determining for every two images $\hatPi$ and $\hatPj$ both of their two common lines. Unfortunately, generalizations of this method to $C_{n}$ for $n>2$ are of little practical use, as it requires to detect all $n$ common lines~(Lemma~\ref{lemma:n_common_lines}) of every two images, which is impractical when the input images are noisy. Nevertheless, the method presented in the current paper is in some sense an extension of~\cite{Cheng_thesis}, but with two key differences. First, detecting all common lines between each pair of images is replaced by a maximum-likelihood-type procedure that directly estimates their relative orientation, from which their common lines can be easily computed. Second, our method takes advantage of self common lines, which do not exist in the case of $C_{2}$ symmetry.

\section{Outline of our method}\label{sec:outline_cn}
In this section, we provide an outline of our method for recovering the orientations of a given set of projection-images of a molecule with $C_{n}$ symmetry with $n>2$. Figure~\ref{fig:flowchart_cn} depicts a flowchart of the method.

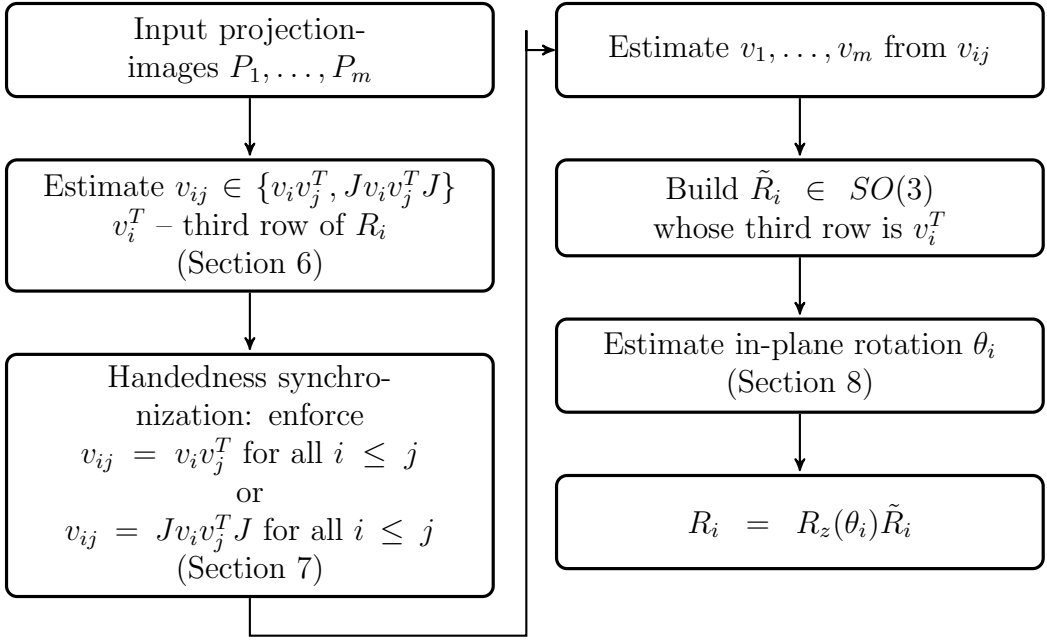
\begin{figure}
	\centering
\begin{tikzpicture}
  [node distance=.8cm,
  start chain=going below,]
     \node[punktchain, join] (input) {Input projection-images $P_{1},\ldots,P_{m}$};
     \node[punktchain, join] (estimatevivj) {Estimate $\vij \in \{ \vivj, J\vivj J \}$ \\
     $v_{i}^T$ -- third row of $R_{i}$ \\
     (Section~\ref{sec:third_row_cn})};
     \node[punktchain, join] (handedness) {Handedness synchronization: enforce \\ $\vij = \vivj$ for all $i \leq j$ \\ or \\ $\vij = J \vivj J$ for all $i \leq j$ \\ (Section~\ref{sec:handedness_sync_cn})};
     \node[punktchain, right = of input] (rank1) {Estimate $v_{1},\ldots,v_{m}$ from $v_{ij}$};
     \node[punktchain, join] (Rtilde) {Build $\tilde{R}_{i} \in SO(3)$ whose third row is $v_{i}^T$};
     \node[punktchain, join] (theta) {Estimate in-plane rotation $\theta_{i}$\\ (Section~\ref{sec:in_plane_cn})};
     \node[punktchain, join] (R) {$R_{i} = R_{z}(\theta_{i}) \tilde{R}_{i}$};

     \draw [->,thick] (handedness.south) |- ++(0,-0.5) |- ++(8.9em,0) |- ++(0,19.5em) |- (rank1.west);
  \end{tikzpicture}
	\caption{A flowchart illustrating the steps for estimating the orientations of projection-images of molecules with $C_{n}$ symmetry.}
	\label{fig:flowchart_cn}
\end{figure}

In what follows, we denote by $\vitrans$ the third row of $R_{i}, \ i=1,\ldots,m$. Given projection-images $P_{1},\ldots,P_{m}$ with corresponding unknown rotations $R_{1},\ldots,R_{m}$ (to be estimated), the first step of our method, denoted as ``relative viewing directions estimation" (Section~\ref{sec:third_row_cn}), consists of estimating for each pair of projection-images $P_{i}$ and $P_{j}$, $i \leq j$, one of the two $3\times 3$ matrices $\vivj$ or $J \vivj J$, where $J = \operatorname{diag}(-1,-1,1)$, using common lines and self common lines of $P_{i}$ and $P_{j}$. The two matrices $\vivj$ and $J \vivj J$ are indistinguishable using common lines due to the handedness ambiguity  discussed in Section~\ref{sec:introduction}. We thus denote the estimated matrix by $\vij$. The second step, which is denoted as ``handedness synchronization"~(Section~\ref{sec:handedness_sync_cn}), enforces that either all estimates $\vij$ have a spurious $J$, or none have at all.  Once these estimates all correspond to a single hand, the next step, denoted as ``viewing directions estimation", consists of forming a $3m \times 3m$ symmetric matrix $V$ whose $(i,j)$-th block of size $3 \times 3$ is given by $\vij$. That is,
\begin{equation}\label{def:V}
V_{ij} =
\begin{cases}
\vij, & \text{if } i \leq j, \\
\vjitrans, &  \text{if } i > j.
\end{cases}
\end{equation}
Depending on the output of the handedness synchronization step described above, the factorization
\begin{equation}\label{eq:V_factorization}
V = v v^{T}, \quad v \in \mathbb{R}^{3m},
\end{equation}
yields at once either the estimates of the third rows $\vitrans$ of all rotation matrices $R_{i}$, or the estimates of all $J$-multiplied third rows $\vitrans J$. 
The next step is to form $m$ rotation matrices $\tilde{R}_{1},\ldots,\tilde{R}_{m}$ where the third row of each $\tilde{R}_{i}$ is set to be equal to the estimate for either $\vitrans$ or $\vitrans J$. The first two rows of each $\tilde{R}_{i}$ are set arbitrarily (so that $\tilde{R}_{i} \in SO(3)$). In Lemma~\ref{lemma:in_plane_single_param} below, we prove that for any $\iInm$ there exist $\theta_{i} \in [0,2 \pi/n)$ and $s_{i} \in [n]$ such that
\begin{equation}\label{eq:theta_i_in_2pi_over_n}
g_{n}^{s_{i}} R_{i} = R_z ( \theta_{i} ) \tilde{R}_{i},
\end{equation}
where $R_z ( \theta_{i} )$ is a rotation matrix that rotates vectors by an angle of $\theta_{i}$ about the $z$-axis (see~\eqref{def:R_x_R_z}). As a result, since each $R_{i}$ may be replaced by either $g_{n}^{s_{i}} R_{i}$ or $g_{n}^{s_{i}} R_{i} J$ for some $s_{i} \in [n-1]$, it suffices to recover the in-plane rotation angles $\theta_{i} \in [0,2 \pi /n)$. The step which consists of recovering all these angles $\theta_{i}$ is
denoted by ``in-plane rotation angles estimation" (Section~\ref{sec:in_plane_cn}). Finally, the last step is ``orientations estimation" in which all rotation matrices $g_{n}^{s_{i}} R_{i}$ (or $g_{n}^{s_{i}} R_{i} J$) are formed using~\eqref{eq:theta_i_in_2pi_over_n}.

We next present a method for estimating the set of all relative viewing directions $\{ \vivj \mid i \le j, \ i,j=1,\ldots,m \}$ which may applied to molecules with $C_{n}$ symmetry with $n>2$ (a full treatment of $C_{2}$ symmetry may be found in~\cite{Cheng_thesis}). In addition, we describe in Appendix~\ref{sec:third_row_c2_c3_c4} an alternative method that is applicable only to molecules with either $C_{3}$ or $C_{4}$ symmetry. While the method of Section~\ref{sec:third_row_cn} may be applied to molecules with $C_{3}$ or $C_{4}$ symmetry, the alternative method of Appendix~\ref{sec:third_row_c2_c3_c4} utilizes the underlying geometry that is induced by such molecules. As a result, it produces better results in practice and is much faster.

\section{Relative viewing directions estimation for $C_{n}$ symmetry with $n>2$}\label{sec:third_row_cn}

Loosely speaking, estimating $\{ \vivj \mid i,j=1,\ldots,m, \ i\leq j \}$ is based on inspecting for each pair of images all possible pairs of rotation matrices (discretized in some proper manner, as explained later on in Section~\ref{sec:Numerical_experiments_cn}), and finding the pair that induces common lines and self common lines which are most correlated. Such an approach is advantageous for the following two reasons:
\begin{enumerate}
	\item \label{itm:first_reason} By Lemma~\ref{lemma:n_common_lines}, any two images have $n$ pairs of common lines, and by Corollary~\ref{cor:self_cls_half_of_pairs}, any image has $\floor*{\frac{n-1}{2}}$ pairs of (non collinear) self common lines. Therefore, for any two images, the degree by which a given pair of candidate rotation matrices induces the correct relative orientation of the two images, may be ascertained with greater confidence as $n$ increases. Specifically, if most of the common lines and self common lines that are induced by a given pair of candidate rotation matrices are highly correlated, then it is likely that these candidates correspond to the true rotation matrices.
	\item \label{itm:second_reason} Recall that any $R_{i}$ may be replaced by $g_{n}^{s_{i}} R_{i}$ where $s_{i} \in [n-1]$ is arbitrary. In light of that, for each $\iInm$, if we write the third column $R_{i}^{(3)}$ of $R_{i}$ in its spherical coordinates representation, i.e., $R_{i}^{(3)} = \left( \sin \theta_{i} \cos \phi_{i}, \sin \theta_{i} \sin \phi_{i},\cos \theta_{i} \right)^{T}$ for some $\phi_{i} \in [0,2\pi)$ and $\theta_{i} \in [0,\pi)$, then by a direct calculation, we get that for any $s_{i} \in [n]$,
	\begin{equation}\label{eq:pre_def_SO_n_3}
	g_{n}^{s_{i}} R_{i}^{(3)}
	=
	\left( \sin \theta_{i} \cos \left( \phi_{i} + \frac{2\pi s_{i}}{n} \right),\sin \theta_{i} \sin \left( \phi_{i} + \frac{2\pi s_{i}}{n} \right),\cos \theta_{i} \right)^{T}.
	\end{equation}
	Thus, since for any $\phi_{i}$ there exists $s_{i} \in [n]$ such that $\phi_{i} + \frac{2\pi s_{i}}{n} \mod 2 \pi \in [0, 2\pi/n)$, it follows that instead of considering the set of ``all possible" candidate rotations for each image, it suffices to only consider the set $SO_{n}(3)$ given by
	\begin{equation}\label{def:SO_n_3}
	SO_{n}(3) =
	\left\{
	R \in SO(3) \,\middle|\, \Phi ( R^{(3)} ) \in [0,2\pi/n)
	\right\},
	\end{equation}
	where $\Phi \colon \mathbb{R}^3 \to [0,2\pi)$ is the mapping $(\sin\theta\cos\phi,\sin\theta\sin\phi,\cos\theta)^T \mapsto \phi$ of any vector in $\mathbb{R}^3$ to its azimuthal angle in its spherical coordinates representation. For large $n$ this significantly restricts the search space of the candidate rotation matrices.
\end{enumerate}

In light of reason~\eqref{itm:second_reason} above, throughout the remaining of this section, we assume without loss of generality that $R_{1},\ldots,R_{m} \in SO_{n}(3)$ (see~\eqref{def:SO_n_3}).
The following lemma, whose proof is given in Appendix~\ref{app:Proof of Lemma g_n_ks}, and the corollary that follows will be used in the sequel.
\begin{lemma}\label{lemma:g_n_ks}
	For any $n \in \mathbb{N}$, $n > 1$, and $l \in \mathbb{Z}$ such that $(l \bmod n) \neq 0$,
	\begin{equation}\label{eq:sum_g}
	\frac{1}{n} \sum_{s=0}^{n-1} g_{n}^{ls}
	=
	\operatorname{diag}\left( 0,0,1 \right).
	\end{equation}
\end{lemma}
\begin{corollary}
	By applying~\eqref{eq:sum_g} to any $n \in \mathbb{N}$, \mbox{$n>1$}, with \mbox{$l=1$}, we get that for any $i,j \in [m]$,
	\begin{equation}\label{eq:linear_comb_n}
	\frac{1}{n}  \sum_{s=0}^{n-1}  \RigsRj
	=
	R_i^T \left( \frac{1}{n} \sum_{s=0}^{n-1} g_{n}^{s} \right) R_j
	=
	R_i^T \operatorname{diag}\left( 0,0,1 \right) R_j
	=
	\vivj.
	\end{equation}	
\end{corollary}
At first glance, it would appear that the rotation matrices $R_{1},\ldots,R_{m} \in SO_{n}(3)$ may be estimated by searching for each pair of images for the pair of rotation matrices that induce pairs of lines whose values are most correlated. That is, since the values along common lines and self common lines have perfect correlation over all other pairs of lines, it follows from~\eqref{eq:clm_four_pairs_cls} and~\eqref{eq:self_clm_part2} that the pair of rotation matrices $\left( R_{i}, R_{j} \right)$ attains the maximal value of $S_{\hatPi, \hatPj} \colon SO_{n}(3) \times SO_{n}(3) \to \mathbb{R}$, given by
\begin{equation}\label{eq:S_tau}
\begin{aligned}
S_{\hatPi, \hatPj}
\left( \tilde{R}_{i}, \tilde{R}_{j} \right)
=
\prod_{s=0}^{n-1}
\operatorname{Re}
\int_{\xi}
\hatPi
\left(
\xi \cos \tilde{\alpha}_{ij}^{(s)},
\xi \sin \tilde{\alpha}_{ij}^{(s)}
\right)
\conjugatet{
\hatPj
\left(
\xi \cos \tilde{\alpha}_{ji}^{(s)},
\xi \sin \tilde{\alpha}_{ji}^{(s)}
\right)}\,\mathrm{d}\xi& \\
\prod_{k \in \{i,j\}}\prod_{s=1}^{\floor*{\frac{n-1}{2}}}
\operatorname{Re}
\int_{\xi}
\hatPk
\left(
\xi \cos \tilde{\alpha}_{kk}^{(s)},
\xi \sin \tilde{\alpha}_{kk}^{(s)}
\right)
	\hatPk
	\left(
	\xi \cos \tilde{\alpha}_{kk}^{(n-s)},
	\xi \sin \tilde{\alpha}_{kk}^{(n-s)}
	\right)
\,\mathrm{d}\xi&,
\end{aligned}
\end{equation}
with $\operatorname{Re}(z)$ denoting the real part of $z \in \mathbb{C}$, where each ray in each of the images in~\eqref{eq:S_tau} is normalized to have norm equal to one, and where (in accordance with~\eqref{eq:hadani} and~\eqref{eq:hadani_self}),
\begin{equation}\label{eq:S_hat_hat}
\begin{aligned}
&\tilde{\alpha}_{ij}^{(s)} = \arctan \left(-\frac{\left( \tilde{R}_{i}^{T} g_{n}^{s} \tilde{R}_{j} \right)_{1,3}}{\left( \tilde{R}_{i}^{T} g_{n}^{s} \tilde{R}_{j} \right)_{2,3}}\right), \quad
\tilde{\alpha}_{ji}^{(s)} = \arctan \left(-\frac{\left( \tilde{R}_{i}^{T} g_{n}^{s} \tilde{R}_{j} \right)_{3,1}}{\left( \tilde{R}_{i}^{T} g_{n}^{s} \tilde{R}_{j} \right)_{3,2}}\right), \quad s=0,\ldots,n-1,\\
&\tilde{\alpha}_{kk}^{(s)} = \arctan \left(-\frac{\left( \tilde{R}_{k} g_{n}^{s} \tilde{R}_{k} \right)_{1,3}}{\left( \tilde{R}_{k} g_{n}^{s} \tilde{R}_{k} \right)_{2,3}}\right), \quad k \in \{i,j\}, \quad s=1,\ldots,\floor*{\frac{n-1}{2}}.
\end{aligned}
\end{equation}
However, for any $\ijInm$, the maximum of~\eqref{eq:S_tau} over $SO_{n}(3) \times SO_{n}(3)$ is not necessarily unique, and---depending on $R_{i}$ and $R_{j}$---might be attained by other pairs of rotation matrices besides $( R_{i}, R_{j})$. For example, by~\eqref{eq:S_hat_hat}, any pair of rotation matrices~$(R^{\ast}_{i}, R^{\ast}_{j})$ such that
\begin{equation}\label{eq:same_self_relative}
\left\{
R_{i}^{\ast T} g_{n}^{s} R^{\ast}_{i}
\right\}_{s=1}^{n-1}
=	
\Big\{
\RigsRi
\Big\}_{s=1}^{n-1}, \quad
\left\{
R_{j}^{\ast T} g_{n}^{s} R^{\ast}_{j}
\right\}_{s=1}^{n-1}
=	
\Big\{
\RjgsRj
\Big\}_{s=1}^{n-1},
\end{equation}
and
\begin{equation}\label{eq:same_relative}
\left\{
R_{i}^{\ast T} g_{n}^{s} R^{\ast}_{j}
\right\}_{s=0}^{n-1}
=	
\Big\{
\RigsRj
\Big\}_{s=0}^{n-1},
\end{equation}
would also maximize~\eqref{eq:S_tau}.
%
Nevertheless, while not every maximizer~$\left( R^{\ast}_{i}, R^{\ast}_{j} \right)$ of~\eqref{eq:S_tau} is necessarily equal to~$\left( R_{i}, R_{j} \right)$, we observed empirically (using extensive simulations) that all such maximizers~$\left( R^{\ast}_{i}, R^{\ast}_{j} \right)$ satisfy~\eqref{eq:same_self_relative} and~\eqref{eq:same_relative}. As a result, it follows from~\eqref{eq:linear_comb_n} that
\begin{equation}\label{eq:vij_invar_cn}
\frac{1}{n}\sum_{s=0}^{n-1}
R_{i}^{\ast T} g_{n}^{s} R^{\ast}_{j}
=
\frac{1}{n}\sum_{s=0}^{n-1}
\RigsRj
=
\vivj,
\end{equation}
where $\vitrans$ and $\vjtrans$ are the third rows of $R_{i}$ and $R_{j}$, respectively. That is, any outer product $\vivj$ is invariant to permutations of the $n$ common lines in $\hatPi$ and $\hatPj$. Thus, for any $\ijInm$, we choose an arbitrary pair $( R^{\ast}_{i}, R^{\ast}_{j} )$ which maximizes~\eqref{eq:S_tau}, and obtain an estimate $\vij$ for $\vivj$ which, due to the inherent handedness ambiguity satisfies $\vij \in \{ \vivj, J \vivj J \}$.
In a similar vein, it follows from~\eqref{eq:same_self_relative} and~\eqref{eq:linear_comb_n} that, for any $\iInm$, any maximizer $( R^{\ast}_{i}, R^{\ast}_{j} )$ of~\eqref{eq:S_tau} yields
\begin{equation}\label{eq:vii_invar_cn}
\frac{1}{n}\sum_{s=0}^{n-1}
R_{i}^{\ast T} g_{n}^{s} R^{\ast}_{i}
=
\frac{1}{n}\sum_{s=0}^{n-1}
\RigsRi
=
\vivi.
\end{equation}
Thus, for any $\iInm$, any of the $m-1$ pairs $( R^{\ast}_{i}, R^{\ast}_{j} )$ where $j \neq i$, induces an estimate $\viij$ for $\vivi$ using~\eqref{eq:vii_invar_cn}. In practice, however, due to self common lines misidentification, any such estimate $\viij$ of $\vivi$ may contain some error. Thus, choosing any one of the estimates $\viij$ to be the single estimate $\vii$ for $\vivi$ is sub-optimal. Furthermore, averaging over all $m-1$ estimates $\viij$ doesn't make any sense, since due to the handedness ambiguity, $\viij \in \{ \vivi, J \vivi J \}$ independently of other estimates. Instead, since any $\vivi$ is a matrix of rank-$1$, we set $\vii$ to be equal to the estimate $\viij$ that is closest to a rank-$1$ matrix. Specifically, for every $\iInm$ we compute the estimates $\viij, \ j\neq i$, and for every such estimate we find (using SVD) its three singular-values $s_{i,1}^{(j)}, s_{i,2}^{(j)}, s_{i,3}^{(j)} \in \mathbb{R}$, and set $\vii = \viijast$ where
\begin{equation}\label{eq:vii_j_ast}
j^\ast \gets \argmin_{\substack{j=1\ldots,m \\ j \neq i}} \left\| \left( s_{i,1}^{(j)}, s_{i,2}^{(j)}, s_{i,3}^{(j)} \right)^{T} - \Big( 1,0,0 \Big)^T \right\|_{2}.
\end{equation}

The procedure for finding all estimates $\vij$ and $\vii$ which, due to the inherent handedness ambiguity, satisfy $\vij \in \{ \vivj, J \vivj J \}$ and $\vii \in \{ \vivi, J \vivi J \}$ is summarized in Algorithm~\ref{alg:vijEstCn}.

\begin{algorithm*}
	\caption{Estimate $\vij, \ i \leq j \in [m]$, for molecules with $C_{n}$ symmetry, with $n>2$}\label{alg:vijEstCn}
	\begin{algorithmic}[1]
		\State {\bfseries Input:}
		\begin{inlinelist}
			\item Images $\hatPi, \ \iInm$.
			\item Cyclic symmetry order $n>2$.
		\end{inlinelist}
		\For{$\ijInm$}
		\State{$ \left(R^{\ast}_{i}, R^{\ast}_{j} \right)
			\gets
			\mathop{\rm argmax} \limits_{\tilde{R}_{i}, \tilde{R}_{j} \in SO_{n}(3) }
			S_{\hatPi, \hatPj} \left( \tilde{R}_{i},\tilde{R}_{j} \right)$}
		\Comment{~\eqref{eq:S_tau}}
		\State{$\vij \gets \frac{1}{n}\sum_{s=0}^{n-1}
			R_{i}^{\ast T} g_{n}^{s} R^{\ast}_{j}$}
		\Comment{~\eqref{eq:vij_invar_cn}}
		\State{$\viij \gets
			\frac{1}{n}\sum_{s=0}^{n-1}
			R_{i}^{\ast T} g_{n}^{s} R^{\ast}_{i}$}
		\Comment{~\eqref{eq:vii_invar_cn}}
		\State{$\vjji \gets
			\frac{1}{n}\sum_{s=0}^{n-1}
			R_{j}^{\ast T} g_{n}^{s} R^{\ast}_{j}$}
		\Comment{~\eqref{eq:vii_invar_cn}}
		\EndFor
		\For{$\iInm$}
		\State{$s_{i,1}^{(j)}, s_{i,2}^{(j)}, s_{i,3}^{(j)} \gets \text{singular-values}\left(\viij\right), \quad j=1,\ldots,m, \quad j \neq i$}
		\Comment{using SVD}
		\State{$j^\ast \gets \argmin_{\substack{j=1,\ldots,m \\ j \neq i}} \left\| \left( s_{i,1}^{(j)}, s_{i,2}^{(j)}, s_{i,3}^{(j)} \right)^{T} - \Big( 1,0,0 \Big)^T \right\|_{2}$}
		\Comment{~\eqref{eq:vii_j_ast}}
		\State{$\vii \gets \viijast$}
		\EndFor
		\State {\bfseries Output:} $\vij, \ i\leq j \in [m]$.
		\Comment{$\vij \in \{\vivj, J \vivj J\}$}
	\end{algorithmic}
\end{algorithm*}

\section{Handedness synchronization}\label{sec:handedness_sync_cn}
At this stage, we have determined the estimates $\vij, \ i \leq j \in [m]$, of all relative viewing directions, where for each estimate $\vij$ either $\vij=\vivj$ or $\vij=J \vivj J$ independently of other estimates. In this section, we describe the ``handedness synchronization" step, where the task is to manipulate these estimates $\vij$ so that either $\vij=\vivj$ for all $i\le j \in [m]$, or $\vij=J \vivj J$ for all $i\le j \in [m]$. Once this procedure is completed, we can form the matrix $V$ from~\eqref{def:V} and infer the third rows $\vitrans$ of all rotation matrices $R_{i}$ (or all $J$-multiplied third rows $\vitrans J$) using~\eqref{eq:V_factorization}. Handedness synchronization is done in two steps. The first step synchronizes the estimates $\vij, \ \ijInm$. The second step synchronizes each of the remaining estimates $\vii, \ \iInm$, with the synchronized estimates from the first step. The reason for separating the synchronization of $\vij$ from that of $\vii$ will be clarified below.

\subsection{Step $1$: Synchronizing the estimates $\vij, \ \ijInm$}\label{sec:handedness_cn_step_1}
We employ a procedure similar to the one described in~\cite{3n}. Specifically, the task of synchronizing the set of estimates $\left\{\vij \mid \ijInm \right\}$ may be reduced to the task of partitioning this set into the following two disjoint sets
\begin{equation}\label{eq:q_ij split}
S_1 =  \left \{\vij \mid \vij = \vivj  \right \},\quad S_2= \left \{\vij \mid \vij = J \vivj J  \right \}.
\end{equation}
Indeed, once all estimates $\vij$ are partitioned into $S_1$ and $S_2$, we can choose either one of the sets (does not matter which one), and replace each estimate $\vij$ in it by $J \vij J$. As a result, since $J^2 = I$, we get that either $\vij = \vivj$ for all $\ijInm$, or $\vij = J \vivj J$ for all $\ijInm$, as needed. We now describe how we obtain such a partition~\eqref{eq:q_ij split}.

Let us define the ``handedness graph" $\signmatCn = (V,E)$ to be the undirected graph whose set of nodes $V$ consists of all estimates $\vij, \ \ijInm$, that is
\begin{equation}\label{eq:definition_V_signmatCn}
V = \left\{ \vij \mid i < j, \ i,j = 1,\ldots,m \right\},
\end{equation}
and whose set of edges $E$ consists of the undirected edges between all triplets of estimates $\vij$, $\vjk$, and $\vik$ (hence each triplet forms a ``triangle"), that is,
\begin{equation}\label{eq:definition_E_signmatCn}
E = \bigcup_{\ijk} \left \{ \left (\vij,\vjk \right ),\ \left (\vjk,\vik \right ), \left (\vik,\vij \right) \right \}.
\end{equation}
The weight of each edge is set to either $+1$ or $-1$ as explained below. For each $i < j < k \in [m]$, we consider the three estimates $\vij$, $\vjk$ and $\vik$, along with the ``triangle" they form in the graph. The goal is to set the weight $+1$ to all edges of the triangle whose incident nodes correspond to estimates that belong to the same set in~\eqref{eq:q_ij split}, and to set the weight of all other edges to $-1$. A crucial observation is that for any $\ijkInm$,
\begin{equation}\label{eq:J_sync_crucial_obs}
\vivj \vjvk = \vivk.
\end{equation}
As such, we check which of the following expressions
\begin{equation}\label{enum:four_poss_cn}
\begin{array}{ll}
1. \; \vij \vjk - \vik & \quad 3. \; \vij \blueJ \vjk \blueJ - \vik \\
2. \; \blueJ \vij \blueJ \vjk - \vik & \quad 4. \; \vij \vjk - \blueJ \vik \blueJ \;
\end{array}
\end{equation}
equals the $3 \times 3$ zero matrix, and assign the weight of each of the three edges in the corresponding triangle as illustrated in Figure~\ref{fig:triangles_cn}. The expressions listed in~\eqref{enum:four_poss_cn} allow to decide for each triplet $\vij$, $\vjk$, $\vik$ which nodes belong to the same set in~\eqref{eq:q_ij split} and which belong to different sets. For three nodes the only two possibilities are that either all nodes belong to the same set (case~1 in~\eqref{enum:four_poss_cn}), or that one node belongs to one of the sets and the remaining two nodes belong to the other set (cases 2, 3, 4 in~\eqref{enum:four_poss_cn}). Thus, by~\eqref{enum:four_poss_cn} we determine the partition~\eqref{eq:q_ij split} ``locally'' for each triplet of nodes. For example, if all three estimates belong to the same set in~\eqref{eq:q_ij split} (i.e., either all belong to $S_1$, or all belong to $S_2$), then the first expression in~\eqref{enum:four_poss_cn} is equal to the zero matrix, meaning that the weights of all three edges in the corresponding triangle in the graph are set to~$+1$~(as per Figure~\ref{fig:conf1_cn}). Indeed, in case $\vij$, $\vjk$, and $\vik$ belong to $S_1$, then since $v_{j}^{T} v_{j} = 1$, we get
$$
\vij \vjk - \vik = \vivj \vjvk - \vivk = 0,
$$
and in case all these three estimates belong to $S_2$, then since $J^2 = I$, it also follows that
$$
\vij \vjk - \vik = J \vivj J J \vjvk J - J \vivk J = 0.
$$

\begin{figure}
	\centering
	\subfloat[$\vij \vjk - \vik = 0$]{
		\label{fig:conf1_cn}
		\centering
		\begin{tikzpicture}[scale=0.8]	
		\draw [blue,line width=1.5pt, fill=gray!2] (10,0) -- (11.5,2.5981) -- (13,0) -- cycle;
		
		\coordinate[label={left:$\vij$}]  (A) at (10,0);
		\coordinate[label={right:$\vjk$}] (B) at (13,0);
		\coordinate[label={above:$\vik$}] (C) at (11.5,2.5981);
		
		\coordinate[label={[red]below:$1$}](c) at ($ (A)!.5!(B) $);
		\coordinate[label={[red]left:$1$}] (b) at ($ (A)!.5!(C) $);
		\coordinate[label={[red]right:$1$}](a) at ($ (B)!.5!(C) $);
		
		\end{tikzpicture}
	}
	\subfloat[$\blueJ \vij \blueJ \vjk - \vik = 0$]{
		\label{fig:conf2_cn}
		\centering
		\begin{tikzpicture}[scale=0.8]	
		\draw [blue,line width=1.5pt, fill=gray!2] (10,0) -- (11.5,2.5981) -- (13,0) -- cycle;
		
		\coordinate[label={left:$\vij$}]  (A) at (10,0);
		\coordinate[label={right:$\vjk$}] (B) at (13,0);
		\coordinate[label={above:$\vik$}] (C) at (11.5,2.5981);
		
		\coordinate[label={[red]below:$-1$}](c) at ($ (A)!.5!(B) $);
		\coordinate[label={[red]left:$-1$}] (b) at ($ (A)!.5!(C) $);
		\coordinate[label={[red]right:$1$}](a) at ($ (B)!.5!(C) $);
		
		\end{tikzpicture}
	}
	\\
	\subfloat[$\vij \blueJ \vjk \blueJ - \vik = 0$]{
		\label{fig:conf3_cn}
		\centering
		\begin{tikzpicture}[scale=0.8]	
		\draw [blue,line width=1.5pt, fill=gray!2] (10,0) -- (11.5,2.5981) -- (13,0) -- cycle;
		
		\coordinate[label={left:$\vij$}]  (A) at (10,0);
		\coordinate[label={right:$\vjk$}] (B) at (13,0);
		\coordinate[label={above:$\vik$}] (C) at (11.5,2.5981);
		
		\coordinate[label={[red]below:$-1$}](c) at ($ (A)!.5!(B) $);
		\coordinate[label={[red]left:$1$}] (b) at ($ (A)!.5!(C) $);
		\coordinate[label={[red]right:$-1$}](a) at ($ (B)!.5!(C) $);
		
		\end{tikzpicture}
	}
	\subfloat[$\vij \vjk - \blueJ \vik \blueJ = 0$]{
		\label{fig:conf4_cn}
		\centering
		\begin{tikzpicture}[scale=0.8]	
		\draw [blue,line width=1.5pt, fill=gray!2] (10,0) -- (11.5,2.5981) -- (13,0) -- cycle;
		
		\coordinate[label={left:$\vij$}]  (A) at (10,0);
		\coordinate[label={right:$\vjk$}] (B) at (13,0);
		\coordinate[label={above:$\vik$}] (C) at (11.5,2.5981);
		
		\coordinate[label={[red]below:$1$}](c) at ($ (A)!.5!(B) $);
		\coordinate[label={[red]left:$-1$}] (b) at ($ (A)!.5!(C) $);
		\coordinate[label={[red]right:$-1$}](a) at ($ (B)!.5!(C) $);
		
		\end{tikzpicture}
	}
	\caption{Subgraphs (triangles) of the handedness graph $\signmatCn$ corresponding to the four configurations in~\eqref{enum:four_poss_cn}. The edges have weight equal to $+1$ if the incident nodes belong to the same set in~\eqref{eq:q_ij split}, and equal to $-1$ otherwise.}
	\label{fig:triangles_cn}
\end{figure}
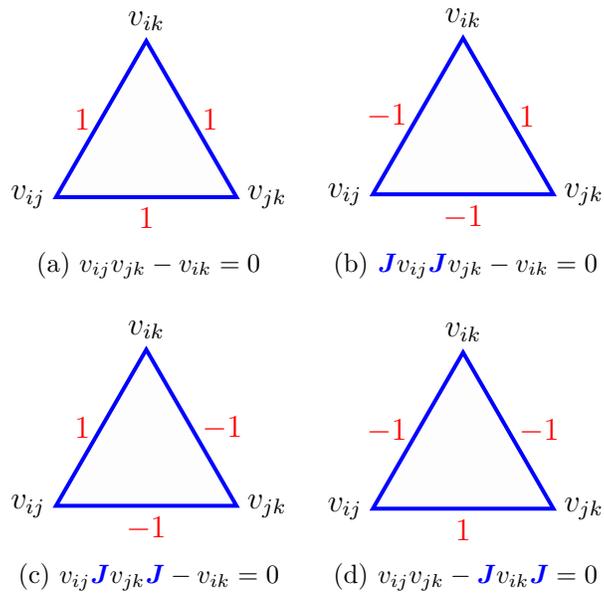

Once we have determined the local assignment of each triplet of nodes (via the weights on the edges between the vertices corresponding to the nodes), we obtain a global partition into the sets in~\eqref{eq:q_ij split} similarly to~\cite{3n}. Specifically, we define the weighted adjacency matrix of $\signmatCn$, also denoted by $\signmatCn$, as the $\binom{m}{2} \times \binom{m}{2}$ matrix whose entries are given by
\begin{equation}\label{def:signmat_cn}
\signmatCn_{(i,j),(k,l)} =
\begin{cases}
\hphantom{-}1,     &  \begin{array}{@{}l@{}}\text{if}\ \left\vert \left \{i,j \right \} \cap  \left \{k,l \right \}\right\vert = 1, \text{ and } \vij,\vkl \text{ are in the same set of~\eqref{eq:q_ij split},}\end{array}\\
-1,     &  \begin{array}{@{}l@{}}\text{if}\ \left\vert \left \{i,j \right \} \cap  \left \{k,l \right \}\right\vert = 1, \text{ and } \vij,\vkl \text{ are in different sets of~\eqref{eq:q_ij split},}\end{array}\\
\hphantom{-}0,     &  \text{if $\left\vert \left \{i,j \right \} \cap  \left \{k,l \right \}\right\vert = 0.$}
\end{cases}
\end{equation}
We then calculate the eigenvector $\evCn$ that corresponds to the leading eigenvalue of the matrix $\signmatCn$. As was shown in~\cite{3n}, this eigenvalue has multiplicity one and its corresponding eigenvector $\evCn \in  \left \{-1,1 \right \}^{\binom{m}{2}}$ encodes the set membership of the estimates. Specifically, if $\evCn(i,j)=1$ then $\vij$ belongs to one of the sets of~\eqref{eq:q_ij split}, and if $\evCn(i,j)=-1$ then $\vij$ belongs to the other set of~\eqref{eq:q_ij split}. As such, by $J$-conjugating all estimates in either one of the sets we are guaranteed that either $\vij = \vivj$ for all $\ijInm$, or $\vij = J \vivj J$ for all $\ijInm$.

Notice that, in practice, the estimates $\vij$ are computed from noisy projection-images, and thus for many triplets of estimates none of the four expressions listed in~\eqref{enum:four_poss_cn} might be equal exactly to the zero matrix. Thus, instead, we search for the expression that is as close as possible to the zero matrix. Specifically, we minimize
\begin{equation}\label{eq:Cijk_cn}
P_{ijk}(\mu_{ij},\mu_{jk})
=
||J^{\mu_{ij}} \vij J^{\mu_{ij}} \cdot J^{\mu_{jk}} \vjk J^{\mu_{jk}} - J^{\mu_{ik}} \vik J^{\mu_{ik}}||_F
\end{equation}
over $\mu_{ij},\mu_{jk},\mu_{ik}\in \{0,1\}$, subject to the constraint that $\mu_{ij} + \mu_{jk} + \mu_{ik} \le 1$, where each possible triplet $\left ( \mu_{ij}, \mu_{jk}, \mu_{ik}\right ) \in \{0,1\}^{3}$ corresponds to one of the four expressions in~\eqref{enum:four_poss_cn}, and $\norm{\cdot}_{F}$ denotes the Frobenius norm.

\subsection{Step $2$: Synchronizing the estimates $\vii, \ \iInm$}\label{sec:handedness_cn_step_2}
The second step of handedness synchronization consists of synchronizing each of the estimates $\vii$ so that if the previous step (described in Section~\ref{sec:handedness_cn_step_1}) resulted in every $\vij$ satisfying $\vij = \vivj$, then the goal is to enforce that $\vii = \vivi$ for every $\iInm$. Otherwise, if the output of the previous step is such that every $\vij$ satisfies $\vij = J \vivj J$, then the goal is to enforce that $\vii = J \vivi J$ for every $\iInm$. Recall, however, that it is unknown which of the above two possible outputs was obtained. Nevertheless, since $\vitrans v_{i} = 1$ for any $\iInm$, and since $J^{2} = I$, it follows that for any $j \in [m]$,
\begin{gather}
(\vivi) (\vivj) = \vivj, \label{eq:vivi_sync_case_1}\\
(J\vivi J) (J \vivj J) = J \vivj J. \label{eq:vivi_sync_case_2}
\end{gather}
As such, we can in principle synchronize every $\vii$ by choosing an arbitrary $\vij$ such that $j \neq i$, and reset $\vii$ as follows:
\begin{subnumcases}{\vii \gets }
\hphantom{J}\vii, & if $\vii \vij = \vij$, \label{eq:vii_sync_case_1}\\
J \vii J, & if $J \vii J \vij = \vij$. \label{eq:vii_sync_case_2}
\end{subnumcases}
Indeed, if $\vij = \vivj$ (which we cannot tell), then by~\eqref{eq:vivi_sync_case_1}, case~\eqref{eq:vii_sync_case_1} occurs if $\vii = \vivi$ and therefore $\vii$ should not be altered. Similarly, if $\vij = \vivj$ then by~\eqref{eq:vivi_sync_case_1}, case~\eqref{eq:vii_sync_case_2} occurs if $\vii = J \vivi J$ (since $J^2 = I$) so that by assigning $\vii \gets J \vii J$ we indeed end up having $\vii = \vivi$, as needed. The case of $\vij = J \vivj J$ is analogous in light of~\eqref{eq:vivi_sync_case_2}.

In practice, however, since each of the above estimates is computed from noisy images, it might be that neither~\eqref{eq:vii_sync_case_1} nor~\eqref{eq:vii_sync_case_2} hold. In addition, it is desirable to synchronize each $\vii$ based on all estimates $\vij$ such that $j > i$ rather than only using a single such estimate. Thus, instead, each estimate $\vii$ is set according to the majority-vote over all $\vij$. Specifically, let us denote by $\sign(x)$ the sign function~(which equals $1$ if $x\ge 0$ and equals $-1$ otherwise). Then, for every $\iInm$ we reset $\vii$ to be $J \vii J$ in case that
\begin{equation}\label{eq:vote_vii}
\sum_{\substack{j \in [m] \\ j > i}}\sign \left ( || J \vii J \vij - \vij ||_F - || \vii \vij - \vij ||_F \right ) < 0.
\end{equation}
Once this second step is completed, we are guaranteed that all estimates are synchronized, i.e., either $\vij=\vivj$ for all $i\le j \in [m]$, or $\vij=J \vivj J$ for all $i\le j \in [m]$. As such, we then construct the matrix $V$ of~\eqref{def:V}, factorize it as in~\eqref{eq:V_factorization}, and obtain all third rows $v_{1}^{T},\ldots,v_{m}^{T}$ (or $v_{1}^{T} J,\ldots,v_{m}^{T} J$).
The procedure for handedness synchronization is summarized in Algorithm~\ref{alg:Jsync_cn}.

Note that synchronizing the estimates $\vij, \ \ijInm,$ is done separately from  synchronizing the estimates $\vii$. Synchronizing $\vij$ based on~\eqref{eq:J_sync_crucial_obs} involves triplets of indices $i<j<k$, and so is the construction of the graph $\signmatCn$ in~\eqref{def:signmat_cn}. On the other hand, the synchronization of $\vii$ is based on pairs of indices (see~\eqref{eq:vii_sync_case_1} and~\eqref{eq:vii_sync_case_2}). So while it may be better to synchronize $\vij$ and $\vii$ simultaneously, it is currently unclear how to combine the pairwise information required for synchronizing $\vii$ into the triplets stricture required for constructing $\signmatCn$.

\begin{algorithm*}
	\caption{Handedness synchronization of the relative viewing directions estimates}\label{alg:Jsync_cn}
	\begin{algorithmic}[1]
		\State {\bfseries Input:} Relative viewing directions estimates $\vij, \ i \le j \in [m]$ (computed by Algorithm~\ref{alg:vijEstCn}).
		\State {\bfseries Initialization:} Matrix $\signmatCn$ of size $\binom{m}{2} \times \binom{m}{2}$ with all entries set to zero.
		\For{$\ijkInm$}
		\Comment{Section~\ref{sec:handedness_cn_step_1}}
		\State{$\left (\mu_{ij}^\ast,\mu_{jk}^\ast,\mu_{ik}^\ast \right ) \gets \mathop{\rm argmin}\limits_{\substack{\mu_{ij},\mu_{jk},\mu_{ik} \in \{0,1\} \\ \mu_{ij}+\mu_{jk}+\mu_{ik} \le 1}} P_{ijk}(\mu_{ij},\mu_{jk},\mu_{ik})$} \Comment{~\eqref{eq:Cijk_cn}}
		\State{$\signmatCn_{(j,k),(i,j)}, \  \signmatCn_{(i,j),(j,k)} \gets (-1)^{(\mu_{ij}^\ast - \mu_{jk}^\ast)}$} \Comment{$\signmatCn$ is an undirected graph (see~\eqref{def:signmat_cn})}

		\State{$\signmatCn_{(k,i),(j,k)}, \  \signmatCn_{(j,k),(k,i)} \gets (-1)^{(\mu_{jk}^\ast - \mu_{ik}^\ast)}$}

		\State{$\signmatCn_{(i,j),(k,i)}, \  \signmatCn_{(k,i),(i,j)} \gets (-1)^{(\mu_{ik}^\ast - \mu_{ij}^\ast)}$} 		
		\EndFor
		
		\State{$\evCn \gets \mathop{\rm argmax} \limits_{\|v\|=1} v^T \signmatCn v$}
		\Comment{$\evCn$ is the eigenvector of the leading eigenvalue of $\signmatCn$}
		\For{$\ijInm$}
		\If{$\evCn(i,j)<0$}
		\State{$\vij \gets J \vij J$}
		\EndIf
		\EndFor
		\For{$\iInm$}
		\Comment{Section~\ref{sec:handedness_cn_step_2}}
		\If{$\sum\limits_{\substack{j \in [m] \\ j > i}} \sign \left (  || J \vii J \vij - \vij ||_F - || \vii  \vij - \vij ||_F\right ) < 0$ }
		\State{$\vii \gets J \vii J$}
		\EndIf
		
		\EndFor
		\State {\bfseries Output:} $\vij, \ i\leq j \in [m]$.
	\end{algorithmic}
\end{algorithm*}

\section{In-plane rotation angles estimation}\label{sec:in_plane_cn}
At this stage, all third rows $v_{1}^{T},\ldots,v_{m}^{T}$ (or $v_{1}^{T} J,\ldots,v_{m}^{T} J$) of the rotation matrices $R_{1},\ldots,R_{m}$ (or $R_{1} J,\ldots,R_{m}J$) have been obtained. In this section, we describe a procedure to determine the remaining first two rows in each of these rotation matrices. The following lemma shows that any two such rows are determined by a single parameter, namely an in-plane rotation angle about the $z$-axis (the axis of symmetry).

\begin{lemma}\label{lemma:in_plane_single_param}
	Let $R$ and $\tilde{R}$ be any two rotation matrices with identical third rows. Then, for any $n \in \mathbb{N}$, there exist a unique $\theta \in [0,2 \pi/n)$ and a unique $s \in [n]$ such that
	\begin{equation}\label{eq:inplane}
		g_{n}^{s} R = R_z ( \theta ) \tilde{R},
	\end{equation}
	where $R_z(\theta)$ (given by~\eqref{def:R_x_R_z}) is the matrix that rotates vectors by an angle $\theta$ about the $z$-axis.
\end{lemma}
The proof of Lemma~\ref{lemma:in_plane_single_param} is given in Appendix~\ref{app:Proof of lemma in_plane_single_param}. In light of Lemma~\ref{lemma:in_plane_single_param}, we next form $m$ rotation matrices $\tilde{R}_{1},\ldots,\tilde{R}_{m}$ by setting the third row of each $\tilde{R}_{i}$ to be equal to the estimate of the third row of $R_{i}$ (see Section~\ref{sec:handedness_sync_cn}), and by arbitrarily setting the first two rows of each $\tilde{R}_{i}$ (while ensuring that $\tilde{R}_{i} \in SO(3)$). As a result, due to Lemma~\ref{lemma:in_plane_single_param}, for any $\iInm$, there only remains to recover the in-plane rotation angle $\theta_{i} \in [0,2 \pi /n)$, and form the matrix $R_z ( \theta_{i} ) \tilde{R_{i}}$.
In principle, all such angles $\theta_{i}$ may be determined in a sequential manner. However, we instead employ a more robust approach in which all $\theta_{i}$ are determined in a single step. Specifically, as we next show, the task of determining all angles $\theta_{i} \in [0, 2 \pi/n)$ may be reduced to finding for every $(i,j), \ i \leq j \in [m]$, the relative in-plane rotation angles
\begin{equation}\label{def:theta_ij_ast}
\theta_{ij} \coloneqq -\theta_{i}+\theta_{j} \ \bmod \ 2\pi/n.
\end{equation}
Indeed, by the definition of $\theta_{ij}$, there exist unique $k_{ij} \in \{0,1\}$, $i \leq j \in [m]$, such that
\begin{equation}\label{eq:pre_G}
\theta_{ij} = -\theta_{i}+\theta_{j}+\frac{2\pi k_{ij}}{n} \in [0,2 \pi/n).
\end{equation}
Thus, by letting $Q$ be the $m \times m$ matrix whose $(i,j)$-th entry $Q_{ij}$ is equal to $e^{\imath n \theta_{ij}}$, it follows that
\begin{equation}\label{def:Q}
Q_{ij}
=
e^{\imath n\theta_{ij}}
=
e^{\imath (n(-\theta_{i}+\theta_{j}) + 2\pi k_{ij})}
=
e^{\imath n(-\theta_{i}+\theta_{j})}, \quad i \leq j = 1,\ldots,m.
\end{equation}
As such, $Q$ is a hermitian rank-$1$ matrix whose factorization is given by
\begin{equation}\label{eq:Q_decomposition}
Q = q q^{H}, \quad q = \left(e^{\imath n\theta_{1}},\ldots,e^{\imath n\theta_{m}}\right)^{H},
\end{equation}
from which all angles $\theta_{i} \in [0, 2 \pi/n)$ may be retrieved using
\begin{equation}\label{eq:theta_i_retrieve_from_q}
\theta_{i} = \frac{1}{n} \arccos \left( \frac{q_i + \conjugatet{q_i}}{2} \right).
\end{equation}
Since $Q$ is a rank-$1$ matrix, its factorization~\eqref{eq:Q_decomposition} yields the eigenvector $\alpha q$ where $\alpha \in \mathbb{C}$ with $\abs{\alpha}=1$ is arbitrary and unknown. As a result,~\eqref{eq:theta_i_retrieve_from_q} actually yields the angles $\theta + \theta_{i}, \ \iInm$, for some arbitrary, unknown $\theta \in [0,2\pi)$. This poses no problem, since for any $\theta \in \mathbb{R}$, recovering $R_z ( \theta + \theta_{i} ) \tilde{R_{i}} = R_z ( \theta ) R_{z}(\theta_{i} ) \tilde{R_{i}}$ is just as good, as we have the degree of freedom of applying any single in-plane rotation about the $z$-axis to all rotation matrices (see the paragraph following~\eqref{def:g}).
In light of~\eqref{eq:pre_G}--\eqref{eq:theta_i_retrieve_from_q}, we describe below how to determine all $\theta_{ij}$ of~\eqref{def:theta_ij_ast}. We then form the matrix $Q$ above, obtain its factorization~\eqref{eq:Q_decomposition} using SVD, and recover all in-plane rotation angles $\theta_{i}$ using~\eqref{eq:theta_i_retrieve_from_q}.

By applying Lemma~\ref{lemma:in_plane_single_param} to any two rotation matrices $R_{i}$ and $R_{j}$, we get that there exist $s_i, s_j \in [n]$ and unique angles $\theta_{i},\theta_{j} \in [0,2\pi/n)$, such that for any $s \in [n]$,
\begin{equation}\label{eq:RiRj_decomp}
\RigsRj =
\left( g_{n}^{-s_i}R_z\left(\theta_{i}\right) \tilde{R}_{i} \right)^{T} g_{n}^{s} \left( g_{n}^{-s_j}R_z\left(\theta_{j}\right) \tilde{R}_{j} \right)
=
\tilde{R}_{i}^{T} R_z\left(-\theta_{i} + \theta_{j} + \dfrac{2\pi s_{ij}}{n}\right) \tilde{R}_{j},
\end{equation}
where the last equality follows by denoting $s_{ij} = s_i+s-s_j$ and using the fact that $g_{n}^{s_{ij}} = R_z\left(\frac{2\pi s_{ij}}{n}\right)$. As such,
\begin{equation}\label{eq:theta_ij_special}
\begin{aligned}
\Bigg\{ \RigsRj \Bigg\}_{s \in [n]}
&=
\left\{
\tilde{R}_{i}^{T} R_z\left(-\theta_{i} + \theta_{j} + \dfrac{2\pi s}{n}\right) \tilde{R}_{j}
\right\}_{s \in [n]} \\
&=
\left\{
\tilde{R}_{i}^{T} R_z\left( \theta_{ij} + \dfrac{2\pi s}{n}\right) \tilde{R}_{j}
\right\}_{s \in [n]},
\end{aligned}
\end{equation}
where the second equality used~\eqref{def:theta_ij_ast} and the fact that $R_z(\phi) = R_z(\phi \ \text{mod} \ 2\pi)$ for any $\phi \in \mathbb{R}$. Given $\theta_{i}$ and $\theta_j$, exactly one of the angles among $-\theta_{i} + \theta_{j} + \dfrac{2\pi s}{n}$, $s \in [n]$, lies in $[0,2\pi/n)$. Thus, by~\eqref{def:theta_ij_ast} and~\eqref{eq:theta_ij_special}, $\theta_{ij}$ is the only angle in $[0,2\pi/n)$ which satisfies
\begin{equation}\label{eq:eq84}
\left\{ \RigsRj \right\}_{s \in [n]}
=
\Big\{
\tilde{R}_{i}^{T} R_z\left( \theta_{ij} + \dfrac{2\pi s}{n}\right) \tilde{R}_{j}
\Big\}_{s \in [n]}.
\end{equation}
Equation~\eqref{eq:eq84} is the set of all possible relative orientations between $\hatPi$ and $\hatPj$. Thus, if we know $\theta_{ij}$ in~\eqref{eq:eq84}, then the common lines between $\hatPi$ and $\hatPj$ induced by these relative orientations will perfectly correlate (see~\eqref{eq:clm_four_pairs_cls}). Therefore, for any $\ijInm$, we set the angle $\theta_{ij}$ to be the maximizer over all $\theta \in [0,2\pi/n)$ of
\begin{equation}\label{eq:theta_ij_opt}
\Upsilon\left( \theta \right) =
\prod_{s=0}^{n-1}
\operatorname{Re}
\int_{\xi}
\hatPi\left(\xi \cos \alpha_{ij}^{(\theta,s)},\xi \sin \alpha_{ij}^{(\theta,s)} \right)
\conjugatet{
\hatPj\left(\xi \cos \alpha_{ji}^{(\theta,s)},
\xi \sin \alpha_{ji}^{(\theta,s)} \right)}\,\mathrm{d}\xi,
\end{equation}
where each ray in each image in~\eqref{eq:theta_ij_opt} is normalized to have norm equal one, and where (in accordance with~\eqref{eq:hadani}), for any $s=0,\ldots,n-1$,
\begin{equation}\label{eq:alpha_ij_theta_s}
\begin{gathered}
\alpha_{ij}^{(\theta,s)} = \arctan \left(-\frac{\left( \tilde{R}_{i}^{T} R_z\left( \theta + \frac{2\pi s}{n}\right) \tilde{R}_{j}\right)_{1,3}}{\left( \tilde{R}_{i}^{T} R_z\left( \theta + \frac{2\pi s}{n}\right) \tilde{R}_{j}\right)_{2,3}}\right), \\
\alpha_{ji}^{(\theta,s)} = \arctan \left(-\frac{\left( \tilde{R}_{i}^{T} R_z\left( \theta + \frac{2\pi s}{n}\right) \tilde{R}_{j}\right)_{3,1}}{\left( \tilde{R}_{i}^{T} R_z\left( \theta + \frac{2\pi s}{n}\right) \tilde{R}_{j}\right)_{3,2}}\right).
\end{gathered}
\end{equation}
As a result, any $\theta_{ij}$ of~\eqref{def:theta_ij_ast} may be obtained by optimizing $\Upsilon$ of~\eqref{eq:theta_ij_opt} over all $\theta \in [0,2\pi/n)$. The procedure for determining all in-plane rotation angles is summarized in Algorithm~\ref{alg:inplane_est}.

\begin{algorithm*}
	\caption{Recover the in-plane rotation angles}\label{alg:inplane_est}
	\begin{algorithmic}[1]
		\State {\bfseries Input:}
		\begin{inlinelist}
			\item Images $\hatPi, \ i=1,\ldots m$.
			\item $\tilde{R}_{i} \in SO(3), \ \iInm$, where the third row of $\tilde{R}_{i}$ is equal to the estimate of the third row of $R_{i}$.
			\item Discretization parameter $K \in \mathbb{N}$
			\item Cyclic symmetry order~$n$.
		\end{inlinelist}
		\State {\bfseries Initialize:}
		\begin{inlinelist}
		\item Array $\Upsilon$ of length $K$.
		\item Matrix $Q$ of size $m \times m$ with all entries set to zero.
		\end{inlinelist}
		\For{$\ijInm$}
		\For{$k \in [K]$}
		\State{$\theta_k \gets \dfrac{2\pi}{nK}(k-1)$}
		\State{$\RijsThetaHat \gets \tilde{R}_{i}^{T} R_z\left( \theta_k + \dfrac{2\pi s}{n} \right) \tilde{R}_{j}, \quad s=0,\ldots,n-1$}
		\State{$\alpha_{ij}^{(\theta_k,s)} \gets \arctan \left( -\dfrac{\left( \RijsThetaHat \right)_{1,3}}{\left( \RijsThetaHat \right)_{2,3}}\right), \quad s=0,\ldots,n-1$}
		\Comment{~\eqref{eq:alpha_ij_theta_s}}
		\State{$\alpha_{ji}^{(\theta_k,s)} \gets \arctan \left( -\dfrac{\left( \RijsThetaHat \right)_{3,1}}{\left( \RijsThetaHat \right)_{3,2}}\right), \quad s=0,\ldots,n-1$}
		\Comment{~\eqref{eq:alpha_ij_theta_s}}
		
		\State{$\Upsilon\left( \theta_k \right) \gets
			\prod_{s=0}^{n-1}
			\operatorname{Re}
			\int_{\xi}
			\hatPi\left(\xi \cos \alpha_{ij}^{(\theta_k,s)},\xi \sin \alpha_{ij}^{(\theta_k,s)} \right)
			\conjugatet{
			\hatPj\left(\xi \cos \alpha_{ji}^{(\theta_k,s)},
			\xi \sin \alpha_{ji}^{(\theta_k,s)} \right)}\,\mathrm{d}\xi$}
			\Comment{$\hatPi$ and $\hatPj$ should be normalized (see~\eqref{eq:theta_ij_opt})}
		\EndFor
		\State{$k^\ast \gets \argmax_{k} \Upsilon\left( \theta_k \right)$}
		
		\State{$\theta_{ij} \gets \dfrac{2\pi}{nK}(k^\ast-1)$ }
		\Comment{$\theta_{ij} = -\theta_{i}+\theta_{j} \bmod 2\pi/n$}
		\State{$Q_{ij} \gets e^{\imath n\theta_{ij}}$,\quad $Q_{ji} \gets e^{-\imath n\theta_{ij}}$}
		\Comment{~\eqref{def:Q}}
		\EndFor
		\State{$Q \gets Q + I$}
		\Comment{set $Q_{ii}=1$  in accordance with~\eqref{def:Q} and~\eqref{eq:Q_decomposition}}
		\State{$q \gets \mathop{\rm argmax} \limits_{\|u\|=\sqrt{m}} {u^{H} Q u}$} \Comment{$q$ is the eigenvector of the leading eigenvalue of $Q$~\eqref{eq:Q_decomposition}}
		\For{$\iInm$}
		\State{$\theta_{i} \gets \frac{1}{n} \arccos \left( \frac{q_i + \conjugatet{q_i}}{2} \right)$}
		\Comment{$q_i = e^{-\imath n\theta_{i}}$~\eqref{eq:theta_i_retrieve_from_q}}
		\EndFor
		\State {\bfseries Output:} $\theta_{i}, \ \iInm$.
		\Comment{$R_z\left( \theta_{i} \right) \tilde{R} = g_{n}^{s_{i}} R_{i}$ for some $s_{i} \in [n]$}
	\end{algorithmic}
\end{algorithm*}

Based on Algorithms~\ref{alg:vijEstCn},~\ref{alg:Jsync_cn},~\ref{alg:inplane_est} and on Section~\ref{sec:outline_cn}, the end-to-end algorithm for recovering all rotation matrices $R_{1},\ldots,R_{m}$ is summarized in Algorithm~\ref{alg:end_2_end}. Note that the fact that we assumed in~\eqref{def:g} without loss of generality that the axis of symmetry coincides with the $z$-axis, means that any reconstructed volume that is based on the output rotation matrices of Algorithm~\ref{alg:end_2_end} will have its axis of symmetry aligned with the $z$-axis as well.

\begin{algorithm*}
	\caption{Orientations estimation for projection-images of molecules with $C_{n}$ symmetry}\label{alg:end_2_end}
	\begin{algorithmic}[1]
		\State {\bfseries Input:}
		\begin{inlinelist}
			\item Images $\hatPi, \ i=1,\ldots m$.
			\item Cyclic symmetry order $n \in \mathbb{N}, \ n > 2$.
		\end{inlinelist}
		\State {\bfseries Initialize:} A matrix $V$ of size $3m \times 3m$.
		\State{Find estimates $\left\{ \vij \right\}_{i<j=1}^m$ and $\left\{ \vii \right\}_{i=1}^{m}$ using Algorithm~\ref{alg:vijEstCn}}
		\State{Synchronize the handedness of $\left\{ \vij \right\}_{i<j=1}^m$ and $\left\{ \vii \right\}_{i=1}^{m}$ using Algorithm~\ref{alg:Jsync_cn}}
		
		\For{$\ijInm$}
		\State{$V_{ij} \gets \vij,\quad V_{ji} \gets \vijtrans,\quad V_{ii} \gets \vii$}
		\Comment{~\eqref{def:V}}
		\EndFor
		
		\State{$\tilde{v} \gets \mathop{\rm argmax} \limits_{\|u\|=1} u^{T} V u$} \Comment{$\mathbb{R}^{3m} \ni \tilde{v}=\left(\tilde{v}_{1}^{T},\ldots,\tilde{v}_{m}^{T}\right)^{T}$}
		\For{$\iInm$}
		\State{Construct an arbitrary $\tilde{R}_{i} \in SO(3)$ whose third row is $\frac{\tilde{v}_{i}^T}{\|\tilde{v}_{i}\|}$}
		\EndFor
		\State{Estimate the in-plane rotation angles $\left\{ \theta_{i} \right\}_{i=1}^{m}$ using Algorithm~\ref{alg:inplane_est}}
		\For{$\iInm$}
		\State{$R_{i}^{\ast} \gets R_{z}\left( \theta_{i} \right) \tilde{R}_{i}$}
		\Comment{~\eqref{eq:theta_i_in_2pi_over_n}}
		\EndFor
		\State {\bfseries Output:} $R_{i}^{\ast}, \ \iInm$.
				\Comment{$\hat{P}_{R_{i}^{\ast}} = \hatPi$}
	\end{algorithmic}
\end{algorithm*}

\section{Numerical experiments}\label{sec:Numerical_experiments_cn}
We implemented Algorithm~\ref{alg:end_2_end} in Matlab and tested it on both simulated and experimental projection-images. Section~\ref{sec:implementation_details_cn} provides some of the implementation details of Algorithm~\ref{alg:end_2_end}, and Section~\ref{sec:complexity_analysis} analyzes its time and space complexity. Section~\ref{sec:experiments_simulated} describes the experiments conducted using noisy simulated projection-images of the three-dimensional density map EMD-$6458$~\cite{EMD_6458_paper}, which has $C_{11}$ symmetry. Section~\ref{sec:experiments_c3_10004} presents results for the Trimeric HIV-$1$ envelope glycoprotein~\cite{EMPIAR-10004} which has $C_3$ symmetry. Section~\ref{sec:experiments_c4_10081} focuses on the Human HCN1 hyperpolarization-activated cyclic nucleotide-gated ion channel~\cite{EMPIAR-10081} which possesses $C_4$ symmetry. Finally, Section~\ref{sec:experiments_c7_groel} reports results for the GroEL protein~\cite{GroEL_Ludtke}. Technically, GroEL has both a $7$-fold cyclic symmetry about some axis, as well as a $2$-fold cyclic symmetry (about a different axis which is perpendicular to the above-mentioned axis). Thus, strictly speaking, GroEL has $D_7$ (dihedral) symmetry. We nevertheless applied Algorithm~\ref{alg:end_2_end} to it while taking into account only its $7$-fold cyclic symmetry (i.e., we treated the molecule as $C_7$). The code of all algorithms presented in this paper is available as part of the ASPIRE software package~\cite{aspire}.

\subsection{Implementation details}\label{sec:implementation_details_cn}
All tests were executed on a dual Intel Xeon X5560 CPU (12 cores in total), with 96GB of RAM running Linux and an nVidia GTX TITAN GPU. Whenever possible, all $12$ cores were used simultaneously, either explicitly using Matlab's \texttt{parfor}, or implicitly, by employing Matlab's implementation of BLAS, which takes advantage of multi-core computing.
Some loop-intensive parts of the algorithm were implemented in \texttt{C} as Matlab \texttt{mex} files.
We next describe the discretization of the set $SO_{n}(3)$ of~\eqref{def:SO_n_3}. As implied by Lemma~\ref{lemma:in_plane_single_param}, any rotation $R \in SO_{n}(3)$ may be written as $R = R_{z}(\theta) \tilde{R}$ where
\begin{inlinelist}
	\item $\tilde{R}$ is a rotation matrix whose third row is given by $(\sin\tilde{\theta}\cos\tilde{\phi},\sin\tilde{\theta}\sin\tilde{\phi},\cos\tilde{\theta})^{T}$ for some angles $\tilde{\phi} \in [0,2\pi)$ and $\tilde{\theta} \in [0,\pi)$, and whose first two rows are set arbitrarily, and
	\item $R_{z}(\theta)$ is the matrix~\eqref{def:R_x_R_z} that rotates a vector by an angle $\theta \in [0,2\pi/n)$ about the $z$-axis.
\end{inlinelist}
As such, we sampled $360,000/n$ evenly-spaced points~$(\tilde{\phi},\tilde{\theta},\theta)^{T} \in [0,2\pi) \times [0,\pi) \times [0,2\pi/n)$, and formed a matrix $R \in SO_{n}(3)$ from each such point (we found experimentally that the resulting number of rotations is adequate).

\subsection{Complexity analysis}\label{sec:complexity_analysis}
Strictly speaking, the computational complexity of Algorithm~\ref{alg:end_2_end} is cubic in the number of images due to Algorithm~\ref{alg:Jsync_cn}. However, in practice, the running time of Algorithm~\ref{alg:end_2_end} is governed by Algorithm~\ref{alg:vijEstCn}, which is quadratic in both the number of images as well as in the size of $SO_{n}(3)$ of~\eqref{def:SO_n_3}. As such, the computational complexity of Algorithm~\ref{alg:end_2_end} is $\mathcal{O}(m^3 + m^2 k^2)$ where $k$ is the number of rotations in the discretization of the set $SO_{n}(3)$. The space (storage) complexity of Algorithm~\ref{alg:end_2_end} is $\mathcal{O}(k)$.

As of timing, using a simulated set of $m=500$ projection-images of a molecule with $C_3$ symmetry, it took $233$ seconds to compute all relative viewing directions $\vij$, $3.1$~seconds to resolve the handedness synchronization, $41$~seconds to estimate the in-plane rotation angles, and $601$~seconds in order to reconstruct the density map.

\subsection{Simulated noisy data}\label{sec:experiments_simulated}
We first tested Algorithm~\ref{alg:end_2_end} on several datasets of simulated noisy images. Specifically, we generated four sets of $m$ projection-images, with $m=500,\ 1000,\ 1500,\ 2000$. The projection-images in each set were generated from the three-dimensional density map EMD-$6458$~\cite{EMD_6458_paper} (which possesses $C_{11}$ symmetry) available in the Electron Microscopy Data Bank (EMDB)~\cite{EMDB}. The orientation $R_{i}$ of each projection-image was drawn uniformly at random from the uniform distribution on $SO(3)$, and the size of each projection-image was $65 \times 65$ pixels. 
We generated four different copies of each of these sets of projection-images by corrupting the projection-images in each set by an additive Gaussian white noise with $\text{SNR} \in \{1/2, 1/4, 1/8, 1/16\}$, where SNR (signal-to-noise ratio) is defined as the ratio between the energy (variance) of the signal and the energy of the noise.

We then applied Algorithm~\ref{alg:end_2_end} to each of the resulting sixteen sets of projection-images. For each such set we obtained an estimated set of rotations, which we denote by $\{ \tilde{R}_{i}\}_{i=1}^{m}$, and we reconstructed the volume using the noisy projection-images and the estimated set of rotations. We next describe how we assessed the degree by which each of these sets of estimated rotations $\{ \tilde{R}_{i}\}_{i=1}^{m}$ differs from the true rotations $\left \{ R_{i}\right \}_{i=1}^{m}$. To this end, we sampled the Fourier-transform $\hat{P}_{R_{i}}$ of each projection-image $P_{R_{i}}$ along the direction vectors $c_{i}^{(l)}=\left ( \cos\ 2\pi l/L, \sin\ 2\pi l/L \right )$, $l=0,\ldots,L-1$ where $L=360$, and then lifted these direction vectors to $\mathbb{R}^3$, i.e., we defined $c_{i}^{(l)}=\left ( \cos\ 2\pi l/L, \sin\ 2\pi l/L, 0 \right )$, $l=0,\ldots,L-1$. We then computed for each pair $(i,l)$ the angle $\epsilon_{il}$ between $R_{i} c_{i}^{(l)}$ and $\tilde{R}_{i} c_{i}^{(l)}$, that is
\begin{equation}\label{eq:epserr}
\begin{gathered}
\epsilon_{il} = \arccos \left \langle R_{i} c_{i}^{(l)},\tilde{R}_{i} c_{i}^{(l)} \right \rangle, \quad
i=1,\ldots,m,\quad l=0,\ldots,L-1,
\end{gathered}
\end{equation}
which measures the error in the estimation of the three-dimensional position of the Fourier ray $c_{i}^{(l)}$. In addition, the quality of each of the reconstructions was assessed by the resolution obtained using the $0.5$~criterion of the Fourier shell correlation (FSC) curve~\cite{vanHeel_Schatz} with respect to the reference density map EMD-$6458$ available in~\cite{EMDB}.
Table~\ref{tbl:simulation_c11} lists the median angular error of $\epsilon_{il}$ (over all $i$ and $l$) for each of the sets along with the obtained resolutions. The table illustrates the robustness of Algorithm~\ref{alg:end_2_end} to noise as the number~$m$ of images increases.

\begin{table}[ht]
	\centering
	\begin{tabular}{|c|c|c|c|c|c|}
		\hline
		\diagbox{m}{SNR}& 1/2 & 1/4 & 1/8 & 1/16 \\
		\hline
		 500& $2.7^o$ ($22.1\angstrom$) & $5.0^o$ ($30.0\angstrom$)&  $9.2^o$ ($30.7\angstrom$)&   $13.3^o$ ($39.1\angstrom$)\\
		1000& $3.1^o$ ($21.9\angstrom$) & $6.0^o$ ($30.1\angstrom$)&  $5.4^o$ ($30.2\angstrom$)&  $14.5^o$ ($36.4\angstrom$)\\
		1500& $2.9^o$ ($21.7\angstrom$) & $4.5^o$ ($26.7\angstrom$)&  $5.7^o$ ($26.4\angstrom$)&   $\hphantom{1}9.5^o$ ($30.0\angstrom$)\\
		2000& $2.9^o$ ($21.5\angstrom$) & $4.4^o$ ($26.3\angstrom$)&  $6.3^o$ ($26.8\angstrom$)&  $\hphantom{1}8.3^o$ ($30.3\angstrom$)\\
		\hline
	\end{tabular}
	\caption{Median angular estimation errors $\epsilon_{il}$ of~\eqref{eq:epserr} and obtained resolutions for various levels of noise applied to each of the sets of simulated projection-images.}
	\label{tbl:simulation_c11}
\end{table}

\subsection{Trimeric HIV-1 envelope glycoprotein ($C_{3})$}\label{sec:experiments_c3_10004}
We next tested Algorithm~\ref{alg:end_2_end} on the Trimeric HIV-$1$ envelope glycoprotein dataset. The dataset consists of $14199$ raw particle images provided in the EMPIAR-$10004$ dataset~\cite{EMPIAR-10004} from the EMPIAR archive~\cite{empiar}. The raw particle images are of size $127 \times 127$ pixels, with pixel size of $2.16\angstrom$. We processed the raw particle images using the ASPIRE~\cite{aspire} software package as follows. First, all images were phase-flipped (in order to remove the phase-reversals in the CTF), down-sampled to size of $89 \times 89$ pixels (hence with pixel size of $3.08\angstrom$), and normalized so that the noise in each image has zero mean and unit variance. We next used the class-averaging procedure in ASPIRE~\cite{aspire} to generate class averages from the raw particle images, where each image was averaged with its $K=100$ most similar images (after proper rotational and translational alignment). Next, we sorted the class averages according to their contrast (i.e., according to the standard deviation of the pixel values of each average). The input to Algorithm~\ref{alg:end_2_end} was the $m=2000$ class averages with the highest contrast. A sample of these class averages is displayed in Figure~\ref{fig:10004_classavg_nn100}.

\begin{figure}
	\centering
	\includegraphics[bb=53 100 379 222, clip=true,width=3.5in]{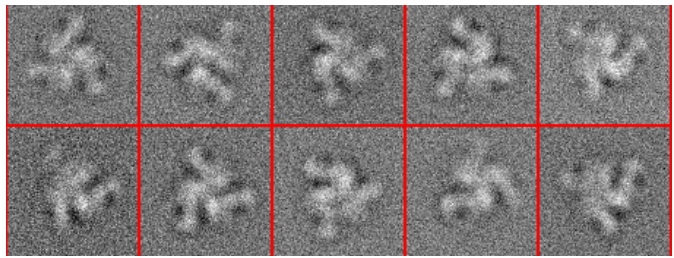}
	\caption{A sample of $89 \times 89$ class averages of the EMPIAR-$10004$ dataset~\cite{EMPIAR-10004} using $K=100$ raw projection-images per class.}
	\label{fig:10004_classavg_nn100}
\end{figure}

Next, we applied Algorithm~\ref{alg:end_2_end} to estimate the rotation matrices $\{ \tilde{R_i} \}_{i=1}^m$ that correspond to the $m$ class averages $\{ P_i \}_{i=1}^m$. Then, instead of reconstructing the three-dimensional density map using merely the $m$ pairs $\{(\tilde{R_i}, P_i)\}_{i=1}^m$, we made full use of the fact that the underlying molecule is $C_3$ symmetric by applying the reconstruction to the $3m$ pairs $\{(\tilde{R_i}, P_i), (g_{3}\tilde{R_i}, P_i), (g_{3}^2\tilde{R_i}, P_i)\}_{i=1}^m$ (see~\eqref{eq:identical_ims}).
Figure~\ref{fig:10004_mine} displays the reconstructed density map, and Figure~\ref{fig:10004_emdb2484} displays the reference density map EMD-$2484$ available in~\cite{EMD-2484_paper} in~\cite{EMDB}. The renderings of all volumes in this section were generated using USCF Chimera~\cite{chimera}. The quality of the reconstruction was assessed using the Fourier shell correlation (FSC) curve~\cite{vanHeel_Schatz}, implying that the resolution of the model estimated by our algorithm is equal to $21.7\angstrom$ according to the $0.5$~criterion (Figure~\ref{fig:10004_fsc}).

\begin{figure}
		\centering
	\subfloat[\label{fig:10004_mine}]%
	{\includegraphics[width=0.35\textwidth]{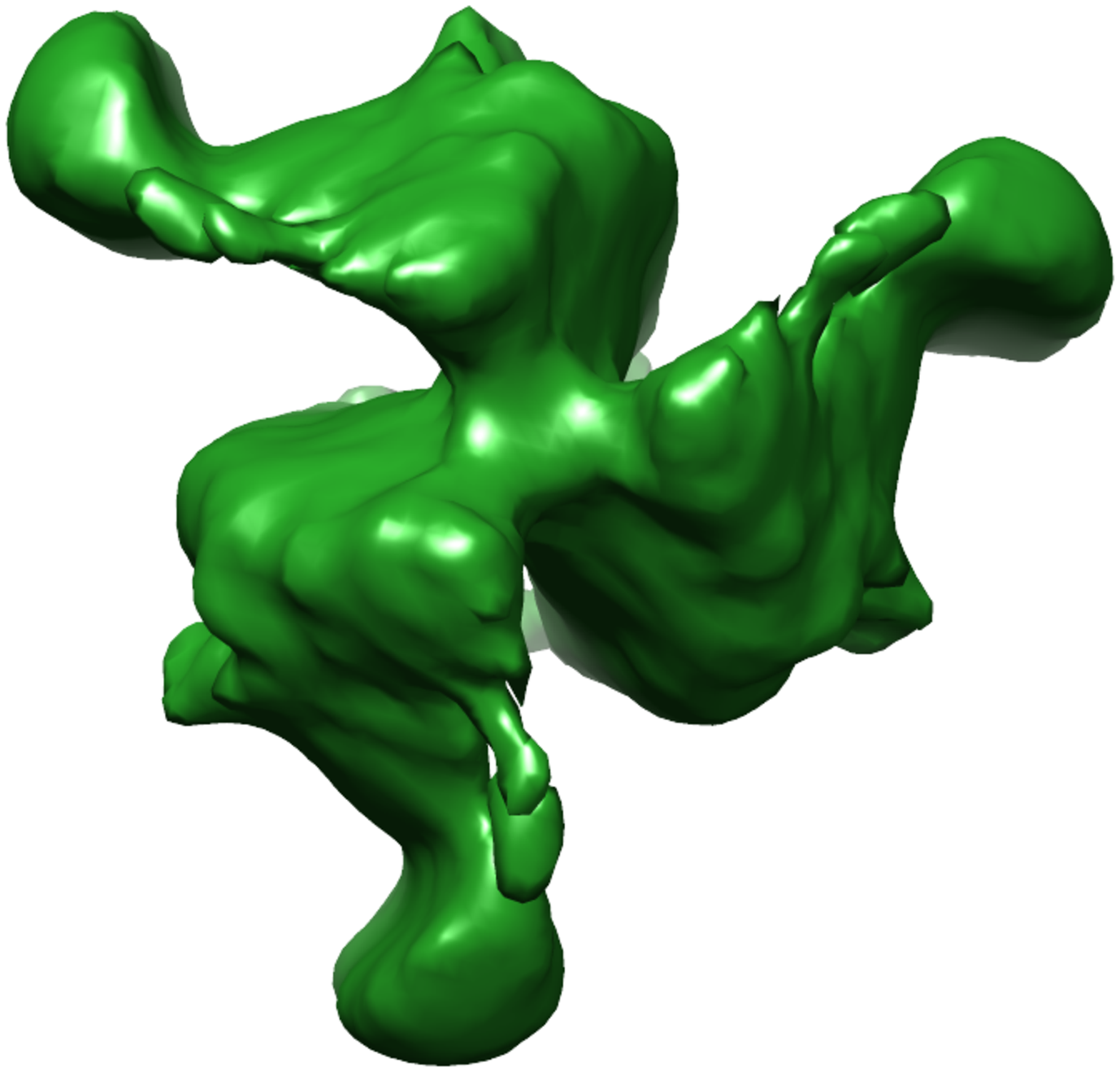}
			\vphantom{\includegraphics[width=0.35\textwidth]{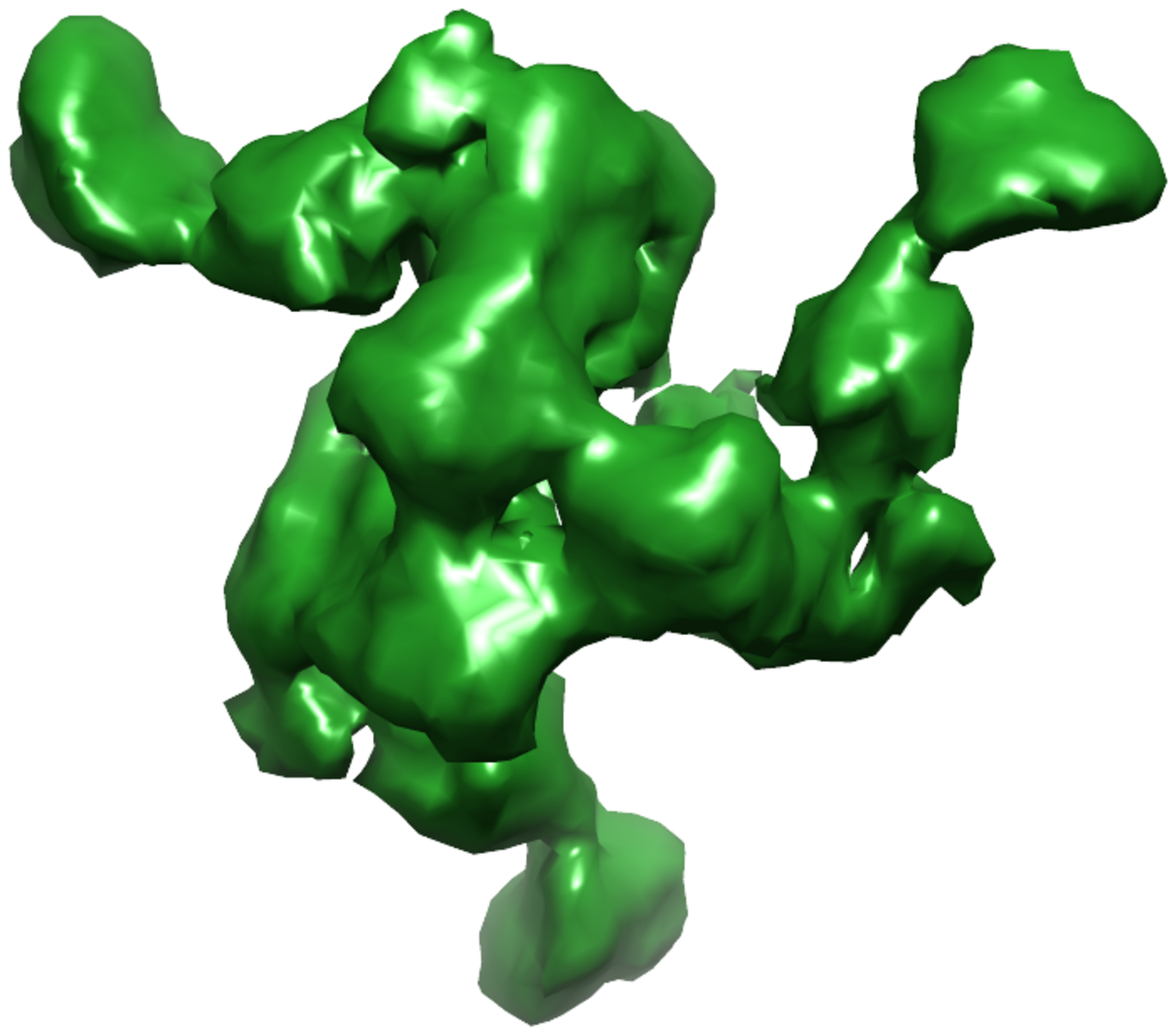}}
		}
		\subfloat[\label{fig:10004_emdb2484}]%
        {\includegraphics[width=0.35\textwidth]{./figures/10004_emdb2484}}

	   \subfloat[\label{fig:10004_fsc}]%
	   {\includegraphics[width=0.45\textwidth]{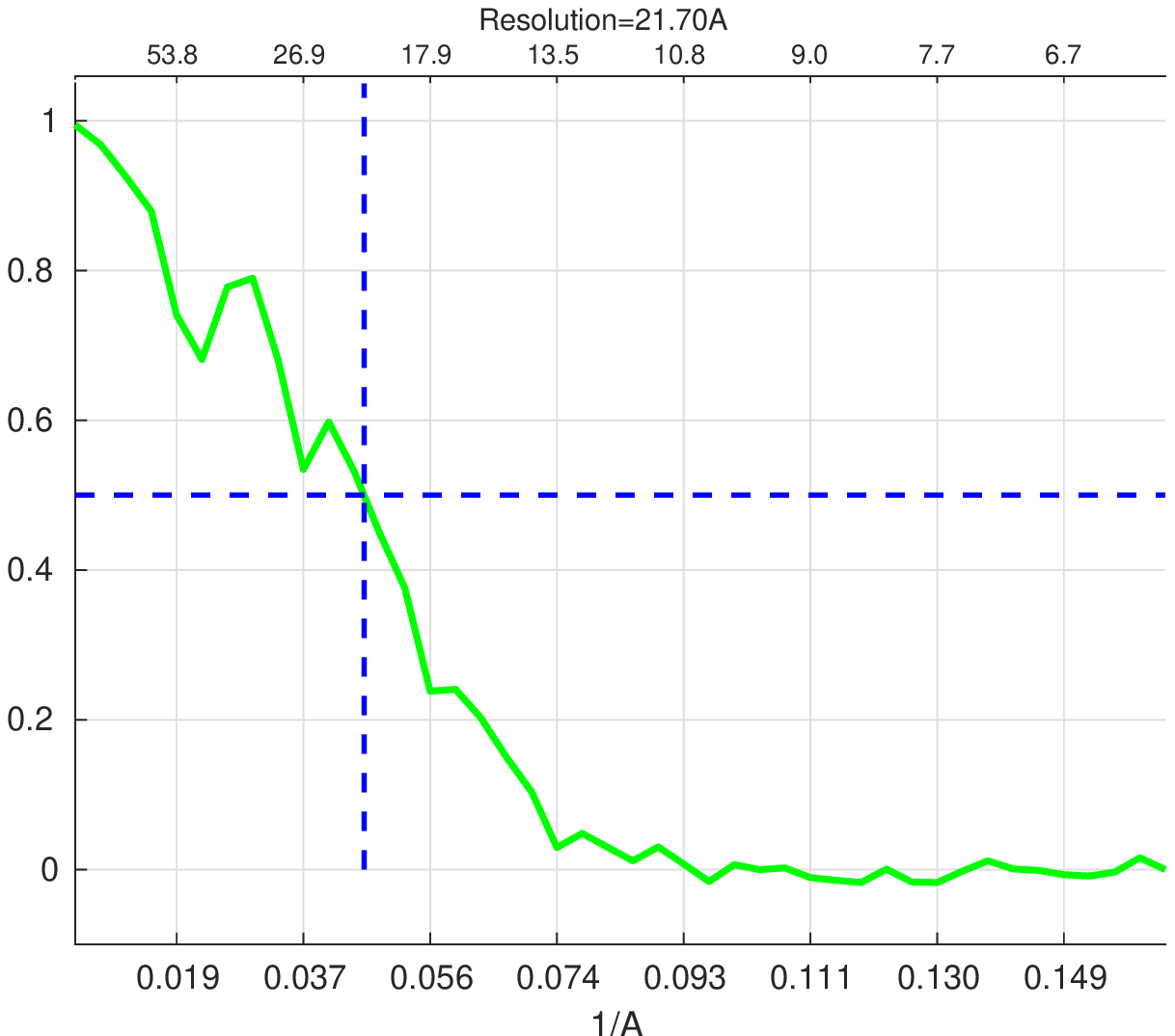}}	
	\caption{Reconstructed density maps of EMPIAR-$10004$~\cite{EMPIAR-10004}. \protect\subref{fig:10004_mine} Using Algorithm~\ref{alg:end_2_end}. \protect\subref{fig:10004_emdb2484} Reference density map~\cite{EMD-2484_paper}. \protect\subref{fig:10004_fsc} FSC between the reconstructed density maps.}
	\label{fig:c3_10004}
\end{figure}


\subsection{Human HCN1 hyperpolarization-activated cyclic nucleotide-gated ion channel~($C_{4}$)}\label{sec:experiments_c4_10081}
Next, we applied Algorithm~\ref{alg:end_2_end} to class averages of the Human HCN1 hyperpolarization activated channel which possesses $C_{4}$ symmetry. The class averages were generated from the particle images provided in the EMPIAR-$10081$ dataset~\cite{EMPIAR-10081}. This dataset comprises of raw particle images of size $256 \times 256$ pixels, with pixel size of $1.3\angstrom$. First, the raw particle images were phase-flipped, down-sampled to size of $129 \times 129$ pixels, and normalized so that the noise in each image would have zero mean and unit variance. To examine the consistency of Algorithm~\ref{alg:end_2_end}, the raw projection-images were randomly split into two groups of $27,935$ images each, and the class-averaging procedure in ASPIRE~\cite{aspire} was used to generate class averages from each of the two groups independently. The class averages were generated by averaging each raw image with its $K=50$ most similar images (using $K=100$ resulted later on in inferior results). The input to subsequent steps were the top (highest contrast) $m=5000$ class averages from each set. A sample of class averages is displayed in Figure~\ref{fig:10081_classavg_nn50}.

\begin{figure}
	\centering
	\includegraphics[bb=55 97 374 224, clip=true,width=3in]{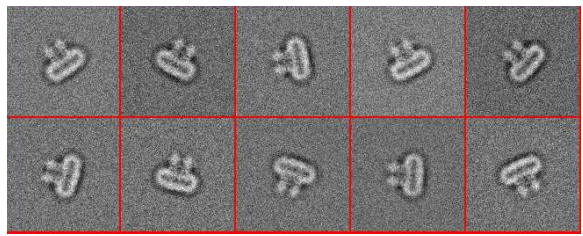}
	\caption{A sample of $129 \times 129$ class averages of the EMPIAR-$10081$ dataset~\cite{EMPIAR-10081} using $K=50$ raw projection-images per class.}
	\label{fig:10081_classavg_nn50}
\end{figure}

Next, we applied Algorithm~\ref{alg:end_2_end} to each of these two groups of class averages and estimated the rotation matrices corresponding to each group. We then reconstructed the two density maps using the class averages and the corresponding estimated rotation matrices, while considering the symmetry in the reconstruction process as was described in Section~\ref{sec:experiments_c3_10004}. The consistency of the reconstructions from the two groups of the data was first assessed using the $0.143$ criterion of the FSC curve~\cite{vanHeel_Schatz}, and was found to be equal to $11.62\angstrom$ (Figure~\ref{fig:10081_fsc_two_groups}). In addition, we compared (using the $0.5$ criterion) the reconstructions against the reference density map which was reconstructed from the same dataset as described in~\cite{EMD_8511_paper}, and found the resolution to be equal to $14.1\angstrom$ (Figure~\ref{fig:10081_fsc_against_emdb}). Figure~\ref{fig:10081_mine} displays a two-dimensional rendering of a density map generated by Algorithm~\ref{alg:end_2_end} (only the reconstruction of the first group is shown), and Figure~\ref{fig:10081_emdb8511} displays a two-dimensional rendering of the reference density map~\cite{EMD_8511_paper}.

\begin{figure}
	\centering
	\subfloat[][]{
		\label{fig:10081_fsc_two_groups}
		\centering
		\includegraphics[width=2.5in]{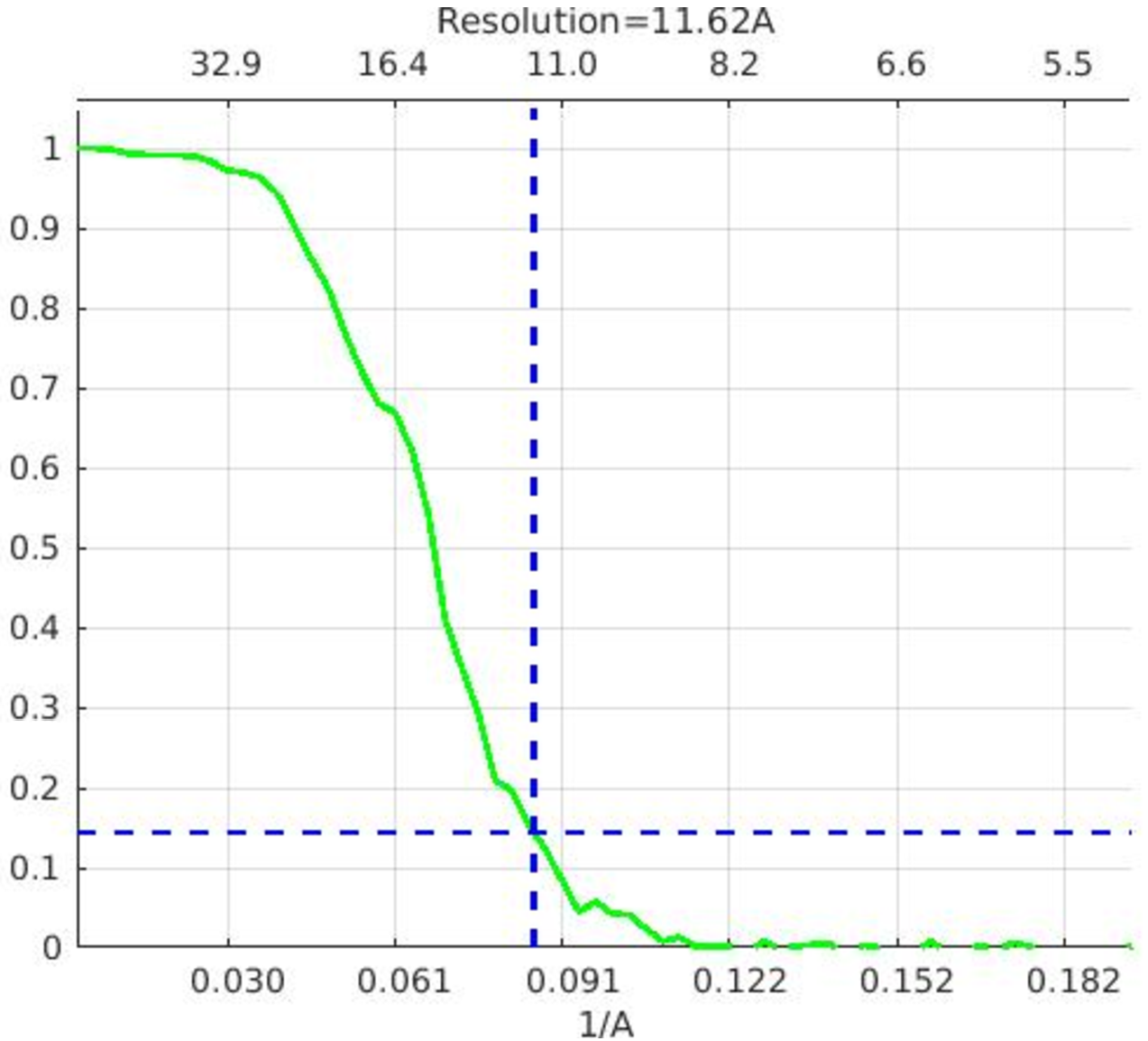}
	}
	\subfloat[][]{
		\label{fig:10081_fsc_against_emdb}
		\centering
		\includegraphics[width=2.5in]{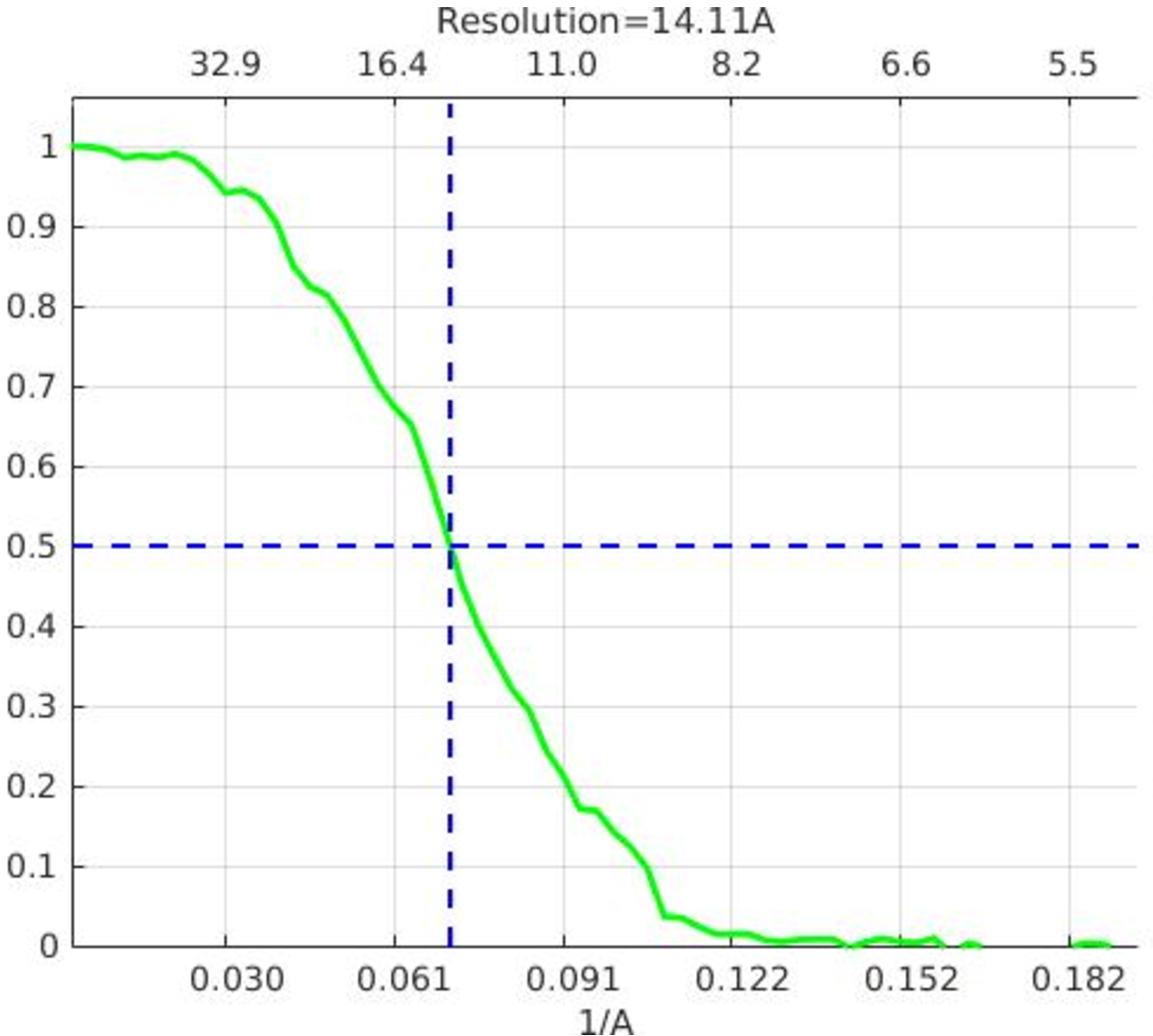}
	}
	\caption{\protect\subref{fig:10081_fsc_two_groups} FSC between the density maps reconstructed from the two halves of the data of EMPIAR-$10081$~\cite{EMPIAR-10081}. \protect\subref{fig:10081_fsc_against_emdb} FSC between the density map reconstructed from the first half of the data of EMPIAR-$10081$~\cite{EMPIAR-10081} and the reference density map~\cite{EMD_8511_paper}.}
\end{figure}

\begin{figure}
	\centering
	\subfloat[][]{
		\label{fig:10081_mine}
		\centering
		\includegraphics[width=1.75in]{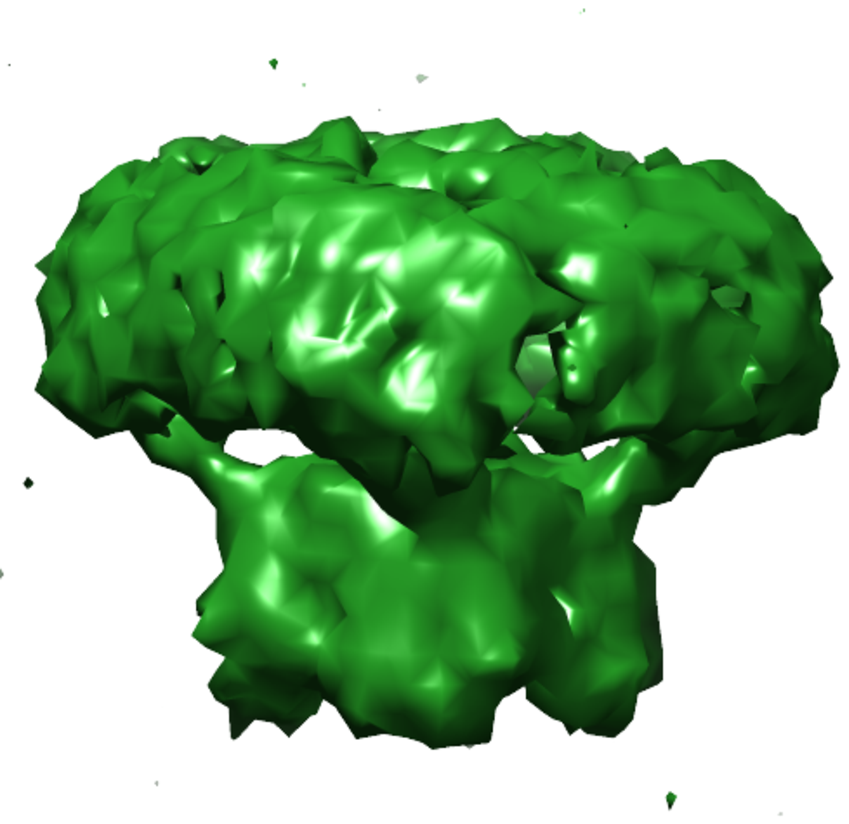}
	}
    \hspace{1cm}
	\subfloat[][]{
		\label{fig:10081_emdb8511}
		\centering
		\includegraphics[width=1.75in]{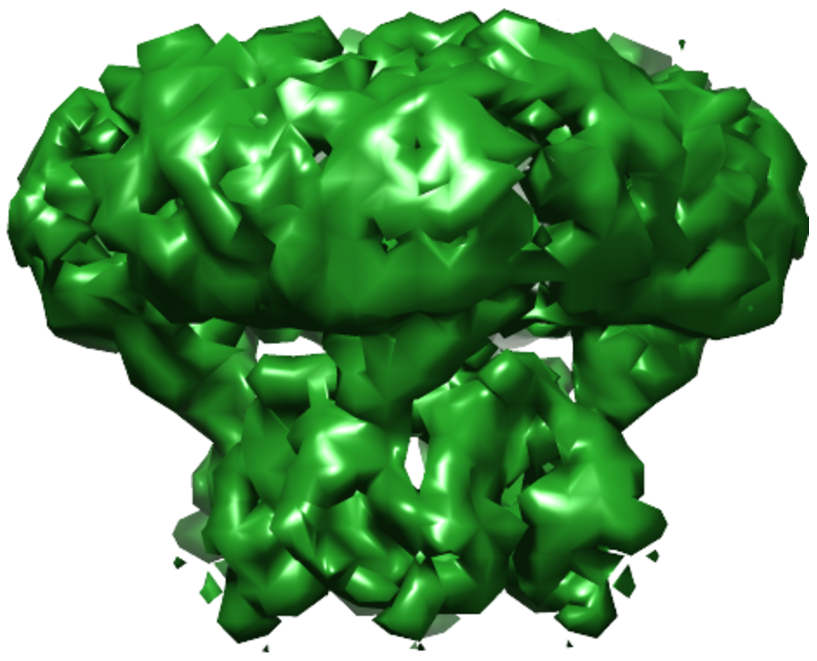}
	}
	\caption{\protect\subref{fig:10081_mine} Density map reconstructed from EMPIAR-$10081$~\cite{EMPIAR-10081} using Algorithm~\ref{alg:end_2_end}. \protect\subref{fig:10081_emdb8511} Reference density map~\cite{EMD_8511_paper}.}
\end{figure}

\subsection{GroEL protein ($D_{7}$ symmetry)}\label{sec:experiments_c7_groel}
As was already mentioned above, strictly speaking, the GroEL protein~\cite{GroEL_Ludtke} has $D_7$ (dihedral) symmetry, meaning that it has both a $7$-fold cyclic symmetry as well as a $2$-fold cyclic symmetry. We nevertheless applied Algorithm~\ref{alg:end_2_end} to it while taking into account only its $7$-fold cyclic symmetry. To this end, we used once again the class-averaging procedure in ASPIRE~\cite{aspire} to generate class averages from the raw particle images, where each image was averaged with its $K=100$ most similar images. A sample of class averages is displayed in Figure~\ref{fig:groel_classavg_nn100}. We then picked the top (highest-contrast) $m=1000$ class averages, estimated the set of corresponding rotation matrices using Algorithm~\ref{alg:end_2_end}, and reconstructed the density map. The resolution was found to be $5.37\angstrom$ (using the $0.5$ criterion). Three different views of the reconstructed density map are shown in Figure~\ref{fig:groel}.

\begin{figure}
	\centering
	\includegraphics[bb=56 99 378 227, clip=true,width=3.5in]{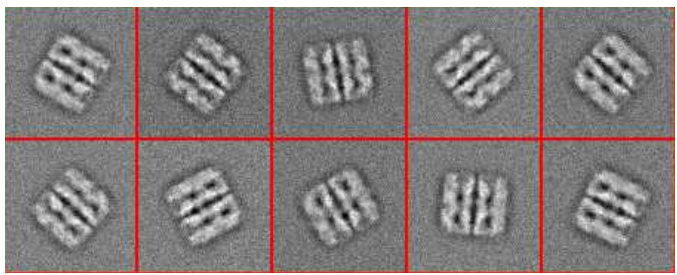}
	\caption{A sample of $89 \times 89$ class averages of the GroEL dataset using $K=100$ raw projection-images per class.}
	\label{fig:groel_classavg_nn100}
\end{figure}

\begin{figure}
	\centering
	\subfloat{
		\centering
		\includegraphics[width=1.3in]{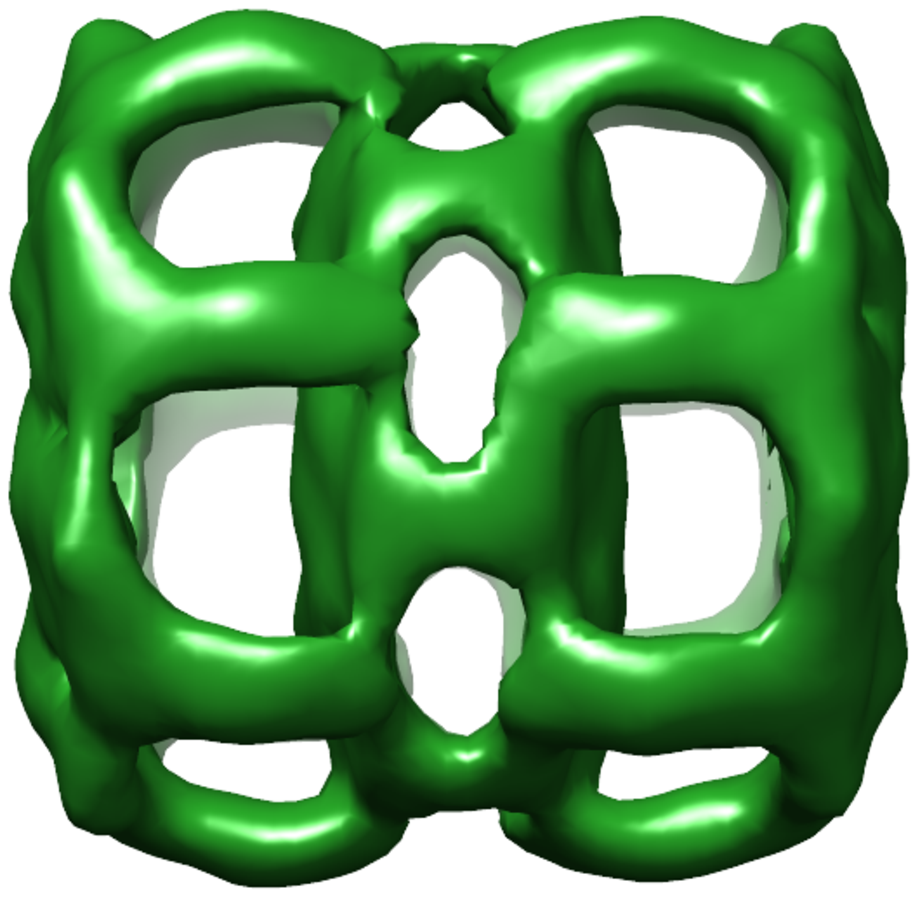}
	}
	\subfloat{
		\centering
		\includegraphics[width=1.3in]{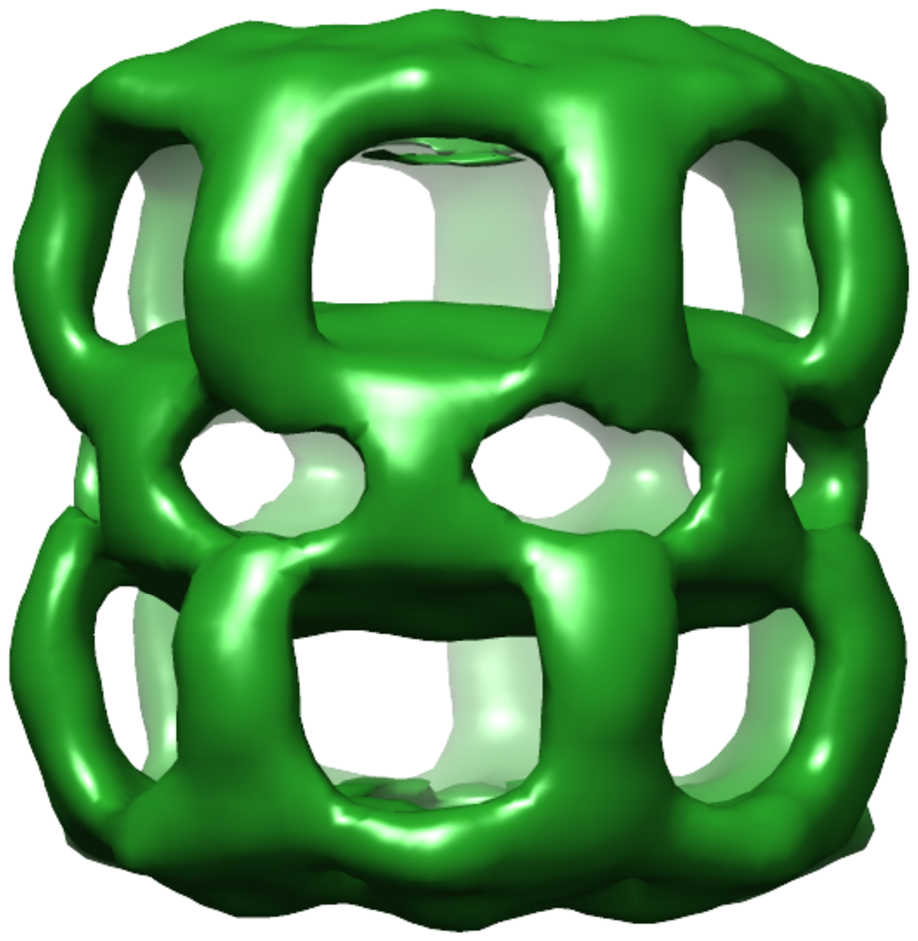}
	}
	\subfloat{
		\centering
		\includegraphics[width=1.3in]{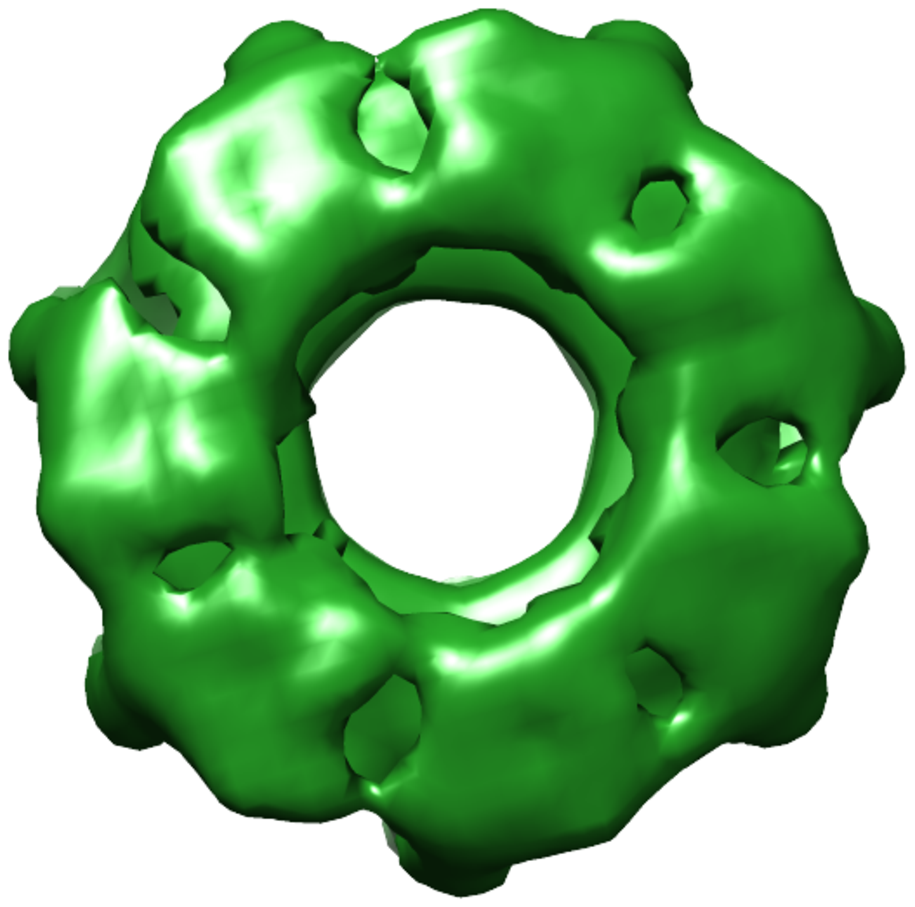}
	}
	\caption{Three different views of the reconstructed density map of the GroEL protein using Algorithm~\ref{alg:end_2_end}.}
	\label{fig:groel}
\end{figure}

\section{Discussion and future work}\label{sec:summary}
In this paper, we proposed a method for finding the orientations that correspond to a given set of projection-images of a cyclically-symmetric molecule. In addition, we described the inherent geometry that underlies such molecules, as well as the way this geometry is expressed in their projection-images. We further demonstrated the efficacy of our proposed method by providing some numerical results using simulated and experimental datasets.

A typical pipeline for reconstructing a three-dimensional volume from a dataset of raw particle images consists of first generating a low-resolution model of the molecule from a subset of the dataset, which is then refined to a high-resolution model using the entire dataset. The reason for breaking the reconstruction process into these two steps is that all high-resolution refinement algorithms are based on non-convex optimization schemes (such as the EM-algorithm or stochastic gradient descent) which must be initialized properly in order not to converge to a molecule which is inconsistent with the data. Thus, the first step of the reconstruction pipeline~(known as ab-initio reconstruction) should generate a reliable low-resolution model of the molecule. The algorithm presented in this paper addresses this step of the pipeline, and has been shown using experimental data to produce reliable low-resolution models for three different datasets. The resolutions achieved by our algorithm are significantly lower than that of the reference models since our algorithm was applied to only part of the data (in the form of a few hundreds or thousands of class averages), and not to the entire raw data set. Nevertheless, the obtained resolutions are consistent with the required resolutions at the ab-initio modeling step. In particular, high-resolution refinement algorithms typically low-pass filter the ab-initio model, and so any reconstruction whose resolution is better than 30~\AA~is typically sufficient.

Obviously, all existing software packages include some functionality for ab-initio modelling. However, all of them are based on some local optimization and have no mathematical guarantees. The algorithm presented in this paper is the first to be specifically designed to the geometry of the problem.

A natural future research is to extend the work to other symmetry groups. We are currently at the final stages of devising an algorithm for molecules with $D_{2}$ symmetry. As it turns out, the geometry of molecules with $D_{2}$ symmetry is completely different from that of molecules with $C_{n}$ symmetry. The reason is that molecules with $D_{2}$ symmetry have three perpendicular symmetry axes, and three corresponding generators of the symmetry group. The resulting algorithm is thus completely different than the one in the current paper. A preliminary analysis of the $D_{2}$ algorithm suggests that it may be extended to $D_{n}$ symmetry with $n>2$.

Once we derive an algorithm for molecule with $D_{n}$ symmetry for $n\ge 2$, there remain three symmetry groups to be handled -- $T$ (tetrahedral), $O$ (octahedral), and $I$ (icosahedral) symmetries. Those enjoy very high-order symmetry, that should be advantageous due to the large number of common lines between any two images and within each image (self common lines). These high-order symmetry groups are currently under investigation.

\begin{appendices}
	\label{appendix}
	\section{Relative viewing directions estimation for $C_3$ or $C_4$ symmetry}\label{sec:third_row_c2_c3_c4}
	In this section, we describe an alternative method for estimating the set of all relative viewing directions $\{ \vivj \mid i \le j, \ i,j=1,\ldots,m \}$ which is applicable only to molecules with either $C_{3}$ or $C_{4}$ symmetry.
	The following lemma, whose proof is given in Appendix~\ref{app:Proof of lemma_c3_c4_self_unique}, is central to the proposed method.
	\begin{lemma}\label{lemma:c3_c4_self_unique}
		For any $n \ge 3$ and for any $i,j \in [m]$ and $s_{ij} \in [n]$,
		\begin{subnumcases}{\vivj=}
		\frac{1}{n} \sum_{s=0}^{n-1}
		\left( R_{i}^T g_{n} R_{i} \right)^{s}
		\left( R_i^T g_{n}^{s_{ij}} R_j \right)
		\left( R_{j}^T g_{n} R_{j} \right)^{s}, \label{eq:linear_comb_vivj_c34_1} \\
		\frac{1}{n} \sum_{s=0}^{n-1}
		\left( R_{i}^T g_{n}^{n-1} R_{i} \right)^{s}
		\left( R_i^T g_{n}^{s_{ij}} R_j \right)
		\left( R_{j}^T g_{n}^{n-1} R_{j} \right)^{s}. \label{eq:linear_comb_vivj_c34_2}
		\end{subnumcases}
		Similarly, for any $n \ge 3$ and $\iInm$,
		\begin{subnumcases}{\vivi=}
		\frac{1}{n} \sum_{s=0}^{n-1}
		\left( R_{i}^T g_{n} R_{i} \right)^{s},\hphantom{+++++++++++++}
		\label{eq:linear_comb_vivi_c34_1} \\
		\frac{1}{n} \sum_{s=0}^{n-1}
		\left( R_{i}^T g_{n}^{n-1} R_{i} \right)^{s}. \label{eq:linear_comb_vivi_c34_2}
		\end{subnumcases}
	\end{lemma}
	
	As a result of Lemma~\ref{lemma:c3_c4_self_unique} we have the following:
	\begin{itemize}
		\item For $n=3$, since the only self relative orientations for any $\iInm$ are $R_{i}^T g_{n} R_{i}$ and $R_{i}^T g_{n}^{n-1} R_{i}$, and since $(R_{i}^T g_{n} R_{i})^T = R_{i}^T g_{n}^{n-1} R_{i}$ and vice versa, it follows that in order to recover $\vivj$ and $\vivi$ it suffices to determine either one of these two self relative orientation for each $\iInm$. To recover $\vivj$, it is also required to determine a single, arbitrary, relative orientation $R_i^T g_{n}^{s_{ij}} R_j$.
		\item For $n=4$, we shall later show that recovering $R_{i}^{T} g_{n}^{2} R_{i}$ may be easily avoided for any $\iInm$. Thus, since for any $\iInm$ the only remaining self relative orientations besides $R_{i}^{T} g_{n}^{2} R_{i}$ are $R_{i}^T g_{n} R_{i}$ and $R_{i}^T g_{n}^{n-1} R_{i}$, it follows that in addition to recovering a single, arbitrary, relative orientation $R_i^T g_{n}^{s_{ij}} R_j$, recovering either one of the remaining two self relative orientations $R_{i}^T g_{n} R_{i}$ or $R_{i}^T g_{n}^{n-1} R_{i}$ for each $\iInm$ is sufficient in order to recover any $\vivj$ and any $\vivi$.
	\end{itemize}
	Applying Lemma~\ref{lemma:c3_c4_self_unique} to molecules with $C_{n}$ symmetry with $n>4$ requires the estimation of more than just a single self relative orientation per image, and was found to be not robust in practice, and so the method of this section may be applied to either $C_{3}$ or $C_{4}$ symmetry. The advantage of this method is that it provides more accurate results in practice than the method of Section~\ref{sec:third_row_cn} and is also significantly faster.
	
\subsection{Estimating self relative orientations}\label{sec:self_deter}
	
	We next describe a robust procedure for determining, for both $n=3$ and $n=4$, and for every $\iInm$, an estimate $\Rii$ such that
	\begin{equation}\label{eq:self_relative_rot_J_amb}
	\Rii \in \left\{ \RigRi, \ J \RigRi J, \ \RignRi, \ J \RignRi J \right\}.
	\end{equation}
	By~\eqref{eq:euler_angles_self}, $\RigRi$ is parameterized by $\left( \alpha_{ii}^{(1)},
	\gamma_{ii}^{(1)}, -\alpha_{ii}^{(n-1)} - \pi \right)$ and $\RignRi$ is parameterized by $\left(
	\alpha_{ii}^{(n-1)},
	\gamma_{ii}^{(n-1)},
	-\alpha_{ii}^{(1)} - \pi \right)$. We next show that $\gamma_{ii}^{(1)} = \gamma_{ii}^{(n-1)}$ for any $\iInm$. As a result, since by~\eqref{eq:linear_comb_vivj_c34_1}--\eqref{eq:linear_comb_vivi_c34_2} we are oblivious as to which of the self relative orientations in~\eqref{eq:self_relative_rot_J_amb} $\Rii$ corresponds to, the angles $\alpha_{ii}^{(1)}$ and $\alpha_{ii}^{(n-1)}$ in the aforementioned parameterizations may be freely interchanged.
	By the projection slice theorem
	\begin{gather}
	\cos \gamma_{ii}^{(1)}
	=
	\dotprod{R_{i}^{(3)}}{g_{n} R_{i}^{(3)}}, \label{eq:gamma_one}\\
	\cos \gamma_{ii}^{(n-1)} = \dotprod{R_{i}^{(3)}}{g_{n}^{n-1} R_{i}^{(3)}}. \label{eq:gamma_n_minus_one}
	\end{gather}
	Thus, since $g_{n}^{T} = g_{n}^{n-1}$, it follows that $\cos \gamma_{ii}^{(1)} = \cos \gamma_{ii}^{(n-1)}$. As a result, since both of these angles are acute, it follows that indeed $\gamma_{ii}^{(1)} = \gamma_{ii}^{(n-1)}$ (which we subsequently denote by $\gamma_{ii}$).
	In order to recover the angles $\alpha_{ii}^{(1)}$ and $\alpha_{ii}^{(n-1)}$, for which the values along the lines they subtend in $\hatPi$ are conjugate equal (see~\eqref{eq:self_clm_part2}), let us define for any $\iInm$ the mapping $S_{i} \colon (0,\pi) \times (0,2\pi) \to \mathbb{R}$ by
	\begin{equation}\label{def:S_i}
	S_{i} \left( \phi,\theta \right)
	=
	\operatorname{Re}\int_{\xi}
	\hatPi \left( \xi \cos \phi, \xi \sin \phi \right)
	\hatPi \left( \xi \cos \theta,\xi \sin \theta \right)\,\mathrm{d}\xi,
	\end{equation}
	where each ray in $\hatPi$ is normalized to have its norm equal to one. By~\eqref{eq:self_clm_part2}, the two angles $\alpha_{ii}^{(1)}$ and $\alpha_{ii}^{(n-1)}$ subtend lines in $\hatPi$ whose Fourier transforms agree up to conjugation. As such, both $( \alpha_{ii}^{(1)}, \alpha_{ii}^{(n-1)})$ and $( \alpha_{ii}^{(n-1)}, \alpha_{ii}^{(1)})$ are solutions of the optimization problem
	\begin{equation}\label{opt:self_angles}
	\begin{aligned}
	& {\text{maximize}}
	& & S_{i} \left( \phi,\theta \right) \\
	& \text{subject to}
	& & \left \lvert \phi - \theta \right \rvert \neq \pi,0.
	\end{aligned}
	\end{equation}
	The constraint in~\eqref{opt:self_angles} is needed because any ray through the origin in $\hatPi$ is conjugate symmetric, and therefore, any two collinear lines would otherwise maximize~\eqref{opt:self_angles}. Moreover, for $n=4$ this constraint also guarantees that the collinear lines that constitute the second pair of self common lines do not maximize~\eqref{opt:self_angles}. For otherwise, it would lead to estimating $R_{i} g_{4}^2 R_{i}$ instead of the desired $R_{i} g_{4} R_{i}$ for~\eqref{eq:linear_comb_vivj_c34_1} and~\eqref{eq:linear_comb_vivi_c34_1}, or $R_{i} g_{4}^3 R_{i}$ for~\eqref{eq:linear_comb_vivj_c34_2} and~\eqref{eq:linear_comb_vivi_c34_2}.
	
	The main incentive to use self relative orientations is the ability to estimate them in a robust manner due to the following two properties:
	\begin{enumerate}
		\item \label{itm:property_first} The domain of each mapping $S_{i}$ defined in~\eqref{def:S_i} may in fact be restricted to a narrower range of angles. Specifically, we show in Lemma~\ref{lemma:angle_greater_pi2} below that for \mbox{$n=3$} it holds that $\lvert \alpha_{ii}^{(n-1)} - \alpha_{ii}^{(1)} \rvert \in [\pi/3,\pi)$, and for \mbox{$n=4$} it holds that $\lvert \alpha_{ii}^{(n-1)} - \alpha_{ii}^{(1)} \rvert \in [\pi/2,\pi)$. As such, when the input projection-images are noisy, constraining the optimization problem~\eqref{opt:self_angles} to these narrower ranges of angles increases the probability of detecting $\alpha_{ii}^{(1)}$ and $\alpha_{ii}^{(n-1)}$.
		\item \label{itm:property_second} We show in Lemma~\ref{lemma:main} below that, for both cases \mbox{$n=3$} and \mbox{$n=4$}, each of the angles $\gamma_{ii}$ may be computed directly from $\lvert \alpha_{ii}^{(n-1)}-\alpha_{ii}^{(1)}
		\rvert$. This is in sharp contrast to relative orientations in general, in which the common lines with a third arbitrary central plane are required~\cite{voting} in order to determine any such angle $\gamma_{ii}$. In particular, when the input images are noisy, the common lines with the third image might be misidentified, leading to a wrong estimation of $\gamma_{ii}$.
	\end{enumerate}
	%
	%
	\begin{lemma}\label{lemma:angle_greater_pi2}
		For any $\iInm$,
		\begin{subnumcases}{\alpha_{ii}^{(n-1)} - \alpha_{ii}^{(1)} \geq}
		\pi/3, & if $n=3$, \label{eq:angle_greater_pi3}\\
		\pi/2, & if $n=4$. \label{eq:angle_greater_pi2}
		\end{subnumcases}
	\end{lemma}
	
	\begin{lemma}\label{lemma:main}
		For any $\iInm$,
		\begin{subnumcases}{\cos \gamma_{ii}=}
		\frac{	\cos \left( \alpha_{ii}^{(n-1)} - \alpha_{ii}^{(1)} \right)} {1-	\cos \left( \alpha_{ii}^{(n-1)} - \alpha_{ii}^{(1)} \right) }, & if $n=3$, \label{eq:cos_gamma_practice__c3}\\
		\frac{1 + \cos \left(  \alpha_{ii}^{(n-1)} - \alpha_{ii}^{(1)} \right)}{1 - \cos \left(  \alpha_{ii}^{(n-1)} - \alpha_{ii}^{(1)}\right)}, & if $n=4$. \label{eq:cos_gamma_practice_}
		\end{subnumcases}
	\end{lemma}
	
	The proof of Lemma~\ref{lemma:angle_greater_pi2} is given in Appendix~\ref{app:Proof of Lemma angle_greater_pi2}, and the proof of Lemma~\ref{lemma:main} is given in Appendix~\ref{app:Proof of Lemma main}.
	Based on~\eqref{opt:self_angles}, on Lemma~\ref{lemma:angle_greater_pi2}, and on Lemma~\ref{lemma:main}, the procedure for determining $\left\{ \Rii \right\}_{i=1}^m$ for molecules with either $C_3$ symmetry {or $C_4$ symmetry} is summarized in Algorithm~\ref{alg:selfEstC34}.
	
	\begin{algorithm*}
		\caption{Recover self relative orientations for molecules with either $C_3$ symmetry or $C_4$ symmetry}\label{alg:selfEstC34}
		\begin{algorithmic}[1]
			\State {\bfseries Input:} Images $\hatPi, \ \iInm$ and symmetry order $n$ (either $n=3$ or $n=4$).
			\If{$n = 3$}
			\For{$\iInm$}
			\State{$\left( \alpha_{ii}^{\ast(1)}, \alpha_{ii}^{\ast(n-1)} \right) \gets \argmaxl_{ \lvert \phi - \theta \rvert \in [\pi/3,\pi)}
				S_{i} \left( \phi,\theta \right)$}
			\Comment{~\eqref{opt:self_angles},~\eqref{eq:angle_greater_pi3}}
			\State{$\gamma_{ii}^\ast
				\gets
				\arccos \left( \frac{	\cos \left( \alpha_{ii}^{\ast(n-1)} - \alpha_{ii}^{\ast(1)} \right)} {1- \cos \left( \alpha_{ii}^{\ast(n-1)} - \alpha_{ii}^{\ast(1)} \right) } \right)$}
			\Comment{~\eqref{eq:cos_gamma_practice__c3}}
			\State{$\Rii \gets R_z(\alpha_{ii}^{\ast(1)}) R_x(\gamma_{ii}^\ast) R_z(-\alpha_{ii}^{\ast(n-1)}-\pi)$}			
			\Comment{~\eqref{def:R_x_R_z},~\eqref{eq:euler_angles_self}}
			\EndFor
			\ElsIf{$n = 4$}
			\For{$\iInm$}
			\State{$\left( \alpha_{ii}^{\ast(1)}, \alpha_{ii}^{\ast(n-1)} \right) \gets \argmaxl_{ \lvert \phi - \theta \rvert \in [\pi/2,\pi)}
				S_{i} \left( \phi,\theta  \right)$}
			\Comment{~\eqref{opt:self_angles},~\eqref{eq:angle_greater_pi2}}
			\State{$\gamma_{ii}^\ast
				\gets
				\arccos \left( \frac{1 + \cos \left(  \alpha_{ii}^{\ast(n-1)} - \alpha_{ii}^{\ast(1)} \right)}{1 - \cos \left(  \alpha_{ii}^{\ast(n-1)} - \alpha_{ii}^{\ast(1)}\right)} \right)$}
			\Comment{~\eqref{eq:cos_gamma_practice_}}
			\State{$\Rii \gets R_z(\alpha_{ii}^{\ast(1)}) R_x(\gamma_{ii}^\ast) R_z(-\alpha_{ii}^{\ast(n-1)}-\pi)$}			
			\Comment{~\eqref{def:R_x_R_z},~\eqref{eq:euler_angles_self}}
			\EndFor
			\EndIf
			\State {\bfseries Output:} $\Rii, \ \iInm.$
			\Comment{$\Rii \in \left\{ R_{i}^{T} g_{n} R_{i}, R_{i}^{T} g_{n}^{n-1} R_{i}, J R_{i}^{T} g_{n} R_{i} J, J R_{i}^{T} g_{n}^{n-1} R_{i} J \right\}$}
		\end{algorithmic}
	\end{algorithm*}
	\subsection{Estimating relative orientations}\label{sec:rel_deter}
	In light of Lemma~\ref{lemma:c3_c4_self_unique}, we next describe how to determine for each of the cases \mbox{$n=3$} or \mbox{$n=4$}, and for every $\ijInm$, a single relative orientation $\RigsijRj$ where $s_{ij} \in [n]$ is arbitrary, and may be different for each $\ijInm$. To this end, for every two images $\hatPi$ and $\hatPj$, we first determine using normalized cross correlation a single common line between these images. We then find the acute angle between the underlying central planes using the voting scheme~\cite{voting}. Finally, using~\eqref{eq:euler_angles} we find an estimate $\Rij$ for a relative orientation of the central planes, which due to the handedness ambiguity corresponds to either $\RigsijRj$ or $J \RigsijRj J$ for some unknown $s_{ij} \in [n]$.
	\subsection{Local handedness synchronization}\label{sec:handedness_cn_local}
	At this stage, any two estimates $\Rii$ and $\Rjj$ (obtained by Algorithm~\ref{alg:selfEstC34}) satisfy
	\begin{equation}\label{eq:Rii_Rjj_satisfy}
	\Rii \in \left\{  \RigsiRi, \ J \RigsiRi J, \right\}, \quad
	\Rjj \in \left\{  \RjgsjRj, \ J \RjgsjRj J, \right\}, \end{equation}
	for some unknown $s_{i},s_{j} \in \left\{ 1,n-1 \right\}$. In light of~\eqref{eq:linear_comb_vivi_c34_1} and~\eqref{eq:linear_comb_vivi_c34_2}, we therefore set
	\begin{equation}\label{def:vii}
	\vii
	=
	\frac{1}{n} \sum_{s=0}^{n-1}
	\left( \Rii \right)^{s}, \quad i=1,\ldots,m,
	\end{equation}
	which guarantees that $\vii \in \{\vivi, J \vivi J\}$ for every $\iInm$.
	Similarly to~\eqref{eq:Rii_Rjj_satisfy}, any estimate $\Rij$ (obtained in Section~\ref{sec:rel_deter}) satisfies
	\begin{equation}\label{eq:Rij_satisfy}
	\Rij \in \left\{  \RigsijRj, \ J \RigsijRj J, \right\},
	\end{equation}
	for some unknown $s_{ij} \in [n]$. However, in order to find an estimate $\vij$ for $\vivj$ using either~\eqref{eq:linear_comb_vivj_c34_1} or~\eqref{eq:linear_comb_vivj_c34_2}, it is essential that
	\begin{inlinelist}
		\item $s_{i}=s_{j}$
		\item either all three estimates $\Rii$, $\Rjj$, and $\Rij$ have a spurious $J$, or none do at all.
	\end{inlinelist}
	In other words, for every $ \ijInm$, the task is to manipulate $\Rii$, $\Rjj$, and $\Rij$ so that they correspond to one of the sets
	\begin{equation}\label{eq:four_poss}
	\begin{Bmatrix}
	\RigRi \\
	\RjgRj \\
	\RigsijRj
	\end{Bmatrix}, \quad
	\begin{Bmatrix}
	R_{i}^T g_{n}^{n-1} R_{i}\\
	R_{j}^T g_{n}^{n-1} R_{j}\\
	\RigsijRj
	\end{Bmatrix}, \quad
	\begin{Bmatrix}
	J \RigRi J \\
	J \RjgRj J \\
	J \RigsijRj J
	\end{Bmatrix}, \quad
	\begin{Bmatrix}
	J R_{i}^T g_{n}^{n-1} R_{i} J   \\
	J R_{j}^T g_{n}^{n-1} R_{j} J   \\
	J\RigsijRj J
	\end{Bmatrix},
	\end{equation}
	followed by setting
	\begin{equation}\label{def:vij}
	\vij
	=
	\frac{1}{n} \sum_{s=0}^{n-1}
	( \Rii )^{s}
	( \Rij )
	( \Rjj )^{s},
	\end{equation}
	as per~\eqref{eq:linear_comb_vivj_c34_1} and~\eqref{eq:linear_comb_vivj_c34_2}. By so doing, it follows that $\vij = \vivj$ whenever one of the first two sets in~\eqref{eq:four_poss} is obtained, and $\vij = J \vivj J$ whenever one of the last two sets in~\eqref{eq:four_poss} is obtained. The task of obtaining for every $\ijInm$ either one of the four sets in~\eqref{eq:four_poss} is referred as the ``local handedness synchronization" of the estimates, and will be addressed next. Once this task is completed, we are guaranteed that $\vij \in \{\vivj, J \vivj J\}$ for every $i \leq j \in [m]$.
	
	A crucial observation is that both matrices $\vivj$ and $J \vivj J$ are rank-$1$. In addition, if an estimate has a spurious $J$, e.g., if $\Rii = J \RigsiRi J$, then since $J^2=I$ it follows that $J \Rii J =  \RigsiRi$. Also, if $\Rii = \RigRi$, then since $g^{T} = g_{n}^{n-1}$, we get that $\Riitrans = R_{i}^{T} g_{n}^{n-1} R_{i}$ and vice versa. Thus, for the case of $C_3$~(i.e., $n=3$), we examine for every $\ijInm$, which of the following expressions yields a rank-$1$ matrix, each obtained by $J$-conjugating a subset of $\{\Rii, \Rjj \}$, and choosing either $\Rii$~(expressions $1$-$4$ in~\eqref{eq:eight_poss_c3}) or $\Riitrans$~(expressions $5$-$8$ in~\eqref{eq:eight_poss_c3}).
	\begin{equation}\label{eq:eight_poss_c3}
	\begin{array}{ll}
	1. \; \Rij+\Rii\Rij\Rjj + \Riitrans\Rij\Rjjtrans & \; 5. \; \Rij+\Riitrans\Rij\Rjj + \Rii\Rij\Rjjtrans\\
	2. \; \Rij+\blueJ\Rii\blueJ\Rij\Rjj + \blueJ\Riitrans\blueJ\Rij\Rjjtrans & \; 6. \; \Rij+\blueJ\Riitrans\blueJ\Rij\Rjj + \blueJ\Rii\blueJ\Rij\Rjjtrans \\
	3. \; \Rij+\Rii\Rij\blueJ\Rjj\blueJ + \Riitrans\Rij\blueJ\Rjjtrans\blueJ& \; 7. \; \Rij+\Riitrans\Rij\blueJ\Rjj\blueJ + \Rii\Rij\blueJ\Rjjtrans\blueJ \\
	4. \; \Rij+\blueJ\Rii\blueJ\Rij\blueJ\Rjj\blueJ + \blueJ\Riitrans\blueJ\Rij\blueJ\Rjjtrans\blueJ & \; 8. \; \Rij+\blueJ\Riitrans\blueJ\Rij\blueJ\Rjj\blueJ + \blueJ\Rii\blueJ\Rij\blueJ\Rjjtrans\blueJ\\
	\end{array}
	\end{equation}
	For example, the first expression in~\eqref{eq:eight_poss_c3} would yield a rank-$1$ matrix in case $s_{i}=s_{j}$ and in addition either all three estimates have a spurious $J$ in them, in which case $\vij = J \vivj J$, or when none of these estimates have a spurious $J$ in them, in which case $\vij = \vivj$. The fifth expression in~\eqref{eq:eight_poss_c3} would yield a rank-$1$ matrix in the same cases as the first expression, only that $s_{i} \neq s_{j}$. As another example, consider the case where only $\Rii$ has a spurious $J$ in it. Then, the second expression in~\eqref{eq:eight_poss_c3} would yield a rank-$1$ matrix in case $s_{i}=s_{j}$ (otherwise, if $s_{i} \neq s_{j}$, then the sixth expression would prevail).
	
	Similarly, due to~\eqref{eq:Rii_Rjj_satisfy} and~\eqref{eq:Rij_satisfy}, in order to use~\eqref{eq:linear_comb_vivj_c34_1} and~\eqref{eq:linear_comb_vivj_c34_2} for the case of $C_4$ (i.e., $n=4$), we need to examine which of the following eight expressions
	\begin{equation}\label{eq:eight_poss_c4_orig}
	\frac{1}{4} \sum_{s=0}^{3}
	( J^{\mu_{i}} \Riitilde  J^{\mu_{i}} )^{s}
	\Rij
	( J^{\mu_{j}} \Rjj J^{\mu_{j}} )^{s}, \quad \Riitilde \in \{\Rii, \Riitrans \}, \quad \mu_{i},\mu_{j} \in \{0,1\}
	\end{equation}
	yields a rank-$1$ matrix. However, it can be shown (see Appendix~\ref{app:Justification of expressions C4}) that these eight expressions are in fact equivalent to
	\begin{equation}\label{eq:eight_poss}
	\begin{array}{ll}
	1. \; \Rij+\Rii\Rij\Rjj & \quad 5. \; \Rij+\Riitrans\Rij\Rjj \\
	2. \; \Rij+\blueJ\Rii\blueJ\Rij\Rjj & \quad 6. \; \Rij+\blueJ\Riitrans\blueJ\Rij\Rjj \\
	3. \; \Rij+\Rii\Rij\blueJ\Rjj\blueJ & \quad 7. \; \Rij+\Riitrans\Rij\blueJ\Rjj\blueJ \\
	4. \; \Rij+\blueJ\Rii\blueJ\Rij\blueJ\Rjj\blueJ & \quad 8. \; \Rij+\blueJ\Riitrans\blueJ\Rij\blueJ\Rjj\blueJ \\
	\end{array}
	\end{equation}
	We thus inspect which of the expressions in~\eqref{eq:eight_poss} results in a rank-$1$ matrix.
	Note that in practice, due to misidentification of (self) common lines, it might be that none of the above expressions yields a rank-$1$ matrix. Therefore, we choose the expression that is closest to be rank-$1$. Specifically, we first compute the three singular values $s_{1}^{(k)}, s_{2}^{(k)}, s_{3}^{(k)} \in \mathbb{R}$ of each of the expressions $k=1,\ldots,8$ of~\eqref{eq:eight_poss_c3} for $n=3$, or of~\eqref{eq:eight_poss} for $n=4$. Then, we choose the expression $k^\ast$ such that
	$$
	k^\ast \gets \argmin_{k=1,\ldots,8} \left\| \left( s_{1}^{(k)}, s_{2}^{(k)}, s_{3}^{(k)} \right)^{T} - \Big( 1,0,0 \Big)^T \right\|_{2},
	$$
	and decide accordingly whether or not to transpose $\Rii$, whether or not to $J$-conjugate $\Rii$, and whether or not $J$-conjugate $\Rjj$. Finally, we apply~\eqref{def:vij} to obtain $\vij$.
	
	The procedure for finding (for molecules with either $C_{3}$ symmetry or $C_{4}$ symmetry) all estimates $\vij$ and $\vii$ which, due to the inherent handedness ambiguity, satisfy $\vij \in \{ \vivj, J \vivj J \}$ and $\vii \in \{ \vivi, J \vivi J \}$, is summarized in Algorithm~\ref{alg:vijEstC34}. This algorithm may replace Algorithm~\ref{alg:vijEstCn} for molecules with either $C_{3}$ or $C_{4}$ symmetry.
	\begin{algorithm*}
		\caption{Estimate $\vij, \ i \leq j \in [m]$, for molecules with either $C_{3}$ symmetry or $C_{4}$ symmetry}\label{alg:vijEstC34}
		\begin{algorithmic}[1]
			\State {\bfseries Input:}
			\begin{inlinelist}
				\item images $\hatPi, \ i=1,\ldots m$
				\item cyclic symmetry order $n=3$ or $n=4$.
			\end{inlinelist}			
			\State{Compute the estimates $\Rii, \ \iInm$, using Algorithm~\ref{alg:selfEstC34}}
			\For{$\ijInm$}
			\State{$\alpha_{ij}^{\ast}, \alpha_{ji}^{\ast} \gets \argmaxl_{\phi, \theta \in [0,2\pi)}
				\operatorname{Re}
				\int_{\xi}
				\hatPi \left( \xi \cos \phi, \xi \sin \phi \right)
				\conjugatet{\hatPj \left( \xi \cos \theta,\xi \sin \theta \right)}\,\mathrm{d}\xi$}
			\Comment{$~\eqref{eq:clm_four_pairs_cls}$}
			\State{Estimate $\gamma_{ij}^{\ast} \in [0,\pi)$ using the voting scheme of~\cite{voting}}
			\State{$\Rij \gets
				R_z(\alpha_{ij}^{\ast})
				R_x(\gamma_{ij}^{\ast})
				R_z(-\alpha_{ji}^{\ast})$}
			\Comment{$~\eqref{eq:euler_angles}$}
			\EndFor
			\For{$\iInm$}
			\State{$\vii
				=
				\frac{1}{n} \sum_{s=0}^{n-1}
				\left( \Rii \right)^{s}$}
			\Comment{~\eqref{def:vii}}
			\EndFor
			\For{$\ijInm$}
			\State{ $\left( \Rii^{\ast}, \Rij^{\ast}, \Rjj^{\ast} \right) \gets \text{Local handedness synchronization of} \left( \Rii, \Rij, \Rjj \right)$}
			\State{$\vij
				=
				\frac{1}{n} \sum_{s=0}^{n-1}
				( \Rii^{\ast} )^{s}
				( \Rij^{\ast} )
				( \Rjj^{\ast} )^{s}
				$}
			\Comment{~\eqref{def:vij}}
			\EndFor
			\State {\bfseries Output:} $\vij, \ i\leq j \in [m]$.
			\Comment{$\vij \in \{\vivj, J \vivj J\}$}
		\end{algorithmic}
	\end{algorithm*}
	
\section{Proofs}

\subsection{Proof of Lemma~\ref{lemma:self_cls_coincide}}\label{app:proof of lemma_self_cls_coincide}
	Let $\iInm$ and let $s \in [n-1]$. By~\eqref{eq:q_ij_s},
	\begin{equation}\label{eq:self_coincide_1}
	\left( g_{n}^{s} R_{i} \right)^T q_{ii}^{(s)}
	=
	\left( g_{n}^{s} R_{i} \right)^T \frac{R_{i}^{(3)} \times g_{n}^{s}
		R_{i}^{(3)}}{\left\|R_{i}^{(3)} \times g_{n}^{s} R_{i}^{(3)}\right\|}.
	\end{equation}
	Since $\left( g_{n}^{s} \right)^T = g_{n}^{n-s}$,~\eqref{eq:self_coincide_1} is equal to
	\begin{equation}\label{eq:self_coincide_2}
	R_{i}^T g_{n}^{n-s} \frac{R_{i}^{(3)} \times g_{n}^{s}
		R_{i}^{(3)}}{\left\|R_{i}^{(3)} \times g_{n}^{s} R_{i}^{(3)}\right\|}
	=
	R_{i}^T \frac{ g_{n}^{n-s}  R_{i}^{(3)} \times
		R_{i}^{(3)}}{\left\|R_{i}^{(3)} \times g_{n}^{s} R_{i}^{(3)}\right\|}.
	\end{equation}
	Next, since the cross-product is an anti-symmetric operation,~\eqref{eq:self_coincide_2} is equal to
	\begin{equation}\label{eq:self_coincide_3}
	-R_{i}^T \frac{R_{i}^{(3)} \times g_{n}^{n-s}
		R_{i}^{(3)}}{\left\|R_{i}^{(3)} \times g_{n}^{s} R_{i}^{(3)}\right\|}
	=
	-R_{i}^T \frac{R_{i}^{(3)} \times g_{n}^{n-s}
		R_{i}^{(3)}}{\left\|R_{i}^{(3)} \times g_{n}^{n-s} R_{i}^{(3)}\right\|}
	=
	-R_{i}^T q_{ii}^{(n-s)},
	\end{equation}
	where the first equality in~\eqref{eq:self_coincide_3} is because the vector $2$-norm is rotation invariant, and therefore the denominator may be written as
	\begin{equation*}
	\left\|R_{i}^{(3)} \times g_{n}^{s} R_{i}^{(3)}\right\|
	=
	\left\|g_{n}^{s} R_{i}^{(3)} \times R_{i}^{(3)}\right\|
	=
	\left\|g_{n}^{n-s}\left( g_{n}^{s} R_{i}^{(3)} \times R_{i}^{(3)}\right)\right\|
	=
	\left\|R_{i}^{(3)} \times g_{n}^{n-s} R_{i}^{(3)}\right\|,
	\end{equation*}
	and the last equality in~\eqref{eq:self_coincide_3} is due to~\eqref{eq:q_ij_s}. From~\eqref{eq:self_coincide_1}--\eqref{eq:self_coincide_3} we get that
	$$
	\left( g_{n}^{s} R_{i} \right)^T q_{ii}^{(s)}
	=
	-R_{i}^T q_{ii}^{(n-s)}.
	$$
	Thus, from~\eqref{eq:alpha_ii_q_ii},
	\begin{equation*}
	\begin{aligned}
	\left(\cos \alpha_{gi}^{(s)},\sin \alpha_{gi}^{(s)},0\right)^T
	&=
	-\left(\cos \alpha_{ii}^{(n-s)},\sin \alpha_{ii}^{(n-s)},0\right)^T \\
	&=
	\left(\cos \left( \alpha_{ii}^{(n-s)} + \pi \right), \sin \left( \alpha_{ii}^{(n-s)} + \pi \right), 0\right)^T,
	\end{aligned}
	\end{equation*}
	from which it follows that
	$$
	\alpha_{gi}^{(s)} = \alpha_{ii}^{(n-s)} + \pi \mod{2 \pi}.
	$$

\qed

\subsection{Proof of Lemma~\ref{lemma:g_n_ks}}\label{app:Proof of Lemma g_n_ks}

	By~\eqref{def:g},
	\begin{equation}\label{eq:g_sum_step1}
	\frac{1}{n} \sum_{s=0}^{n-1} g_{n}^{ls}
	=
	\frac{1}{n} \sum_{s=0}^{n-1} \left(
	\begin{array}{rrr}
	\cos\frac{2 \pi ls}{n} & -\sin\frac{2 \pi ls}{n} & 0\\
	\sin\frac{2 \pi ls}{n} & \cos\frac{2 \pi ls}{n} & 0\\
	0 &   0 & 1
	\end{array}
	\right)
	=
	\frac{1}{n} \left(
	\begin{array}{rrr}
	\sum_{s=0}^{n-1} \cos\frac{2 \pi ls}{n} & - \sum_{s=0}^{n-1}\sin\frac{2 \pi ls}{n} & 0\\
	\sum_{s=0}^{n-1} \sin\frac{2 \pi ls}{n} & \sum_{s=0}^{n-1} \cos\frac{2 \pi ls}{n} & 0\\
	0 &   0 & n
	\end{array}
	\right).
	\end{equation}
	Next, by using the assumption that $(l \bmod n) \neq 0$, we get that
	$$
	\sum_{s=0}^{n-1} \cos\frac{2 \pi ls}{n} + \imath \sum_{s=0}^{n-1} \sin\frac{2 \pi ls}{n}
	=
	\sum_{s=0}^{n-1} e^{\imath \frac{2 \pi ls}{n}}
	=
	\frac{1-e^{\imath 2 \pi l}}{1-e^{\imath\frac{2 \pi l}{n}}}
	=
	0.
	$$
	As such,
	$$
	\sum_{s=0}^{n-1} \cos\frac{2 \pi ls}{n} = 0, \qquad \sum_{s=0}^{n-1} \sin\frac{2 \pi ls}{n} = 0,
	$$
	and therefore, the right hand side of~\eqref{eq:g_sum_step1} is equal to $\operatorname{diag}\left( 0,0,1 \right)$, as needed.

\qed

\subsection{Proof of Lemma~\ref{lemma:in_plane_single_param}}\label{app:Proof of lemma in_plane_single_param}

	We first prove that given any such $R$ and $\tilde{R}$, there exists a unique angle $\phi \in [0,2\pi)$ such that
	\begin{equation}\label{eq:inplane_first}
	R = R_z ( \phi ) \tilde{R}.
	\end{equation}
	To this end, let us denote the three rows of $R$ by $r_{1}^{T}, r_{2}^{T}, r_{3}^{T}$, and the three rows of $\tilde{R}$ by $\tilde{r}_{1}^{T}, \tilde{r}_{2}^{T}, \tilde{r}_{3}^{T}$. Since $r_{3}=\tilde{r}_{3}$, it follows that $\tilde{r}_{3} \perp r_{1}$ and $\tilde{r}_{3} \perp r_{2}$, as well as $r_{3} \perp \tilde{r}_{1}$ and $r_{3} \perp \tilde{r}_{2}$. Thus, by direct calculation
	\begin{equation}\label{eq:RRtilde_abcd}
	R \tilde{R}^{T} =
	\begin{pmatrix}
	a & b & 0 \\
	c & d & 0 \\
	0 & 0 & 1
	\end{pmatrix},
	\end{equation}
	for some $a,b,c,d \in \mathbb{R}$. Since $SO(3)$ is closed with respect to matrix multiplication, it follows that $R \tilde{R}^{T} \in SO(3)$ as well. As such, there exists a unique angle $\phi \in [0,2\pi)$ such that
	\begin{equation}\label{eq:R_Rtilde}
	R \tilde{R}^{T} =
	\begin{pmatrix}
	\cos \phi & -\sin \phi & 0 \\
	\sin \phi & \hphantom{+}\cos \phi & 0 \\
	0 & 0 & 1
	\end{pmatrix}.
	\end{equation}
	Finally, by right multiplying~\eqref{eq:R_Rtilde} by $\tilde{R}$ we get that (since $\tilde{R}^{T} \tilde{R} = I$),
	\begin{equation*}
	R
	=
	R \tilde{R}^{T} \tilde{R}
	=
	\begin{pmatrix}
	\cos \phi & -\sin \phi & 0 \\
	\sin \phi & \hphantom{+}\cos \phi & 0 \\
	0 & 0 & 1
	\end{pmatrix}
	\tilde{R}
	=
	R_z ( \phi ) \tilde{R}.
	\end{equation*}
	which proves~\eqref{eq:inplane_first}. We next prove~\eqref{eq:inplane}. Let $n \in \mathbb{N}$, and let $s \in [n]$ be the unique number such that $\frac{2 \pi s}{n} + \phi  \mod 2 \pi \in [0,2 \pi/n)$, and further define $\theta = \left(\frac{2 \pi s}{n} + \phi \right) \ \text{mod} \ 2\pi$. By construction, $\theta \in [0,2 \pi/n)$. In addition,
	\begin{align}
	g_{n}^{s} R
	=
	g_{n}^{s} R_z( \phi) \tilde{R}
	&=
	R_z \left( \frac{2 \pi s}{n} + \phi \right) \tilde{R}\label{eq:inplane_last_a}\\
	&=
	R_z \left( \left( \frac{2 \pi s}{n} + \phi \right) \ \text{mod} \ 2\pi \right) \tilde{R} \label{eq:inplane_last_b}\\
	&=
	R_z ( \theta ) \tilde{R}\label{eq:inplane_last_c},
	\end{align}
	where~\eqref{eq:inplane_last_a} follows from~\eqref{eq:inplane_first} and from the fact that $g_{n}^{s} = R_z \left( \frac{2 \pi s}{n}\right)$,~\eqref{eq:inplane_last_b} is because $R_z(\tau) = R_z(\tau \ \text{mod} \ 2\pi)$ for any $\tau \in \mathbb{R}$, and~\eqref{eq:inplane_last_c} uses the definition of $\theta$ above. Finally, the uniqueness of $\theta$ follows from the uniqueness of $\phi$, the uniqueness of $s$ and from~\eqref{eq:inplane_last_a}--\eqref{eq:inplane_last_c}.

\qed

\subsection{Proof of Lemma~\ref{lemma:c3_c4_self_unique}}\label{app:Proof of lemma_c3_c4_self_unique}
	We first prove~\eqref{eq:linear_comb_vivj_c34_1} and~\eqref{eq:linear_comb_vivj_c34_2}. Let $a \in \{1, n-1\}$ and let $k \in [m]$. Since $R_{k} R_{k}^T = I$, we get that $\left( R_{k}^T g_{n}^a R_{k} \right)^{s} = R_{k}^T g_{n}^{as} R_{k}$ for any $s \in [n]$. In addition, since $\gcd(a,n) = 1$, it follows that $a$ is a generator of $\mathbb{Z}_{n}$. As a result,
	\begin{equation}\label{eq:a_is_generator}
	\left\{ \left( R_{k}^T g_{n}^a R_{k} \right)^{s} \right\}_{s=0}^{n-1}
	=
	\left\{ R_{k}^T g_{n}^{as} R_{k} \right\}_{s=0}^{n-1}
	=
	\left\{ \left( R_{k}^T g_{n}^s R_{k} \right) \right\}_{s=0}^{n-1}.
	\end{equation}
	Thus, for any $i, j \in [m]$ and for any $s_{ij} \in [n]$,
	\begin{equation}\label{eq:linear_comb_self_all_stp1}
	\begin{aligned}
	\frac{1}{n} \sum_{s=0}^{n-1}
	\left( R_{i}^T g_{n}^a R_{i} \right)^{s}
	\left( R_i^T g_{n}^{s_{ij}} R_j \right)
	\left( R_{j}^T g_{n}^a R_{j} \right)^{s}
	&=
	\frac{1}{n} \sum_{s=0}^{n-1}
	\left( R_{i}^T g_{n}^s R_{i} \right)
	\left( R_i^T g_{n}^{s_{ij}} R_j \right)
	\left( R_{j}^T g_{n}^s R_{j} \right)\\
	&=
	\frac{1}{n}  \sum_{s=0}^{n-1}
	R_{i}^T g_{n}^{s} g_{n}^{s_{ij}} g_{n}^{s} R_{j}\\
	&=
	R_{i}^T g_{n}^{s_{ij}} 	\left( \frac{1}{n}  \sum_{s=0}^{n-1} g_{n}^{2s} \right) R_{j}.
	\end{aligned}
	\end{equation}
	Since for any $n \ge 3$ it holds that $(2 \bmod n) \neq 0$, applying Lemma~\ref{lemma:g_n_ks} using $l=2$ to the right hand side of~\eqref{eq:linear_comb_self_all_stp1} yields
	\begin{equation*}
	R_i^T g_{n}^{s_{ij}} \operatorname{diag}\left( 0,0,1 \right) R_j
	=
	R_i^T \operatorname{diag}\left( 0,0,1 \right) R_j
	=
	\vivj,
	\end{equation*}
	where the first equality used the fact that for any $s_{ij} \in [n]$ the third column of $g_{n}^{s_{ij}}$ is equal to $\left( 0,0,1 \right)^{T}$, and therefore $g_{n}^{s_{ij}} \operatorname{diag}\left( 0,0,1 \right) = \operatorname{diag} \left( 0,0,1 \right)$. This proves~\eqref{eq:linear_comb_vivj_c34_1} and~\eqref{eq:linear_comb_vivj_c34_2}. We next prove~\eqref{eq:linear_comb_vivi_c34_1} and~\eqref{eq:linear_comb_vivi_c34_2}. To this end, for both $a=1$ and $a=n-1$,
	\begin{align}
	\frac{1}{n} \sum_{s=0}^{n-1}
	\left( R_{i}^T g_{n}^a R_{i} \right)^{s}
	&=
	\frac{1}{n} \sum_{s=0}^{n-1}
	\left( R_{i}^T g_{n}^s R_{i} \right)
	=
	R_{i}^T
	\left( \frac{1}{n} \sum_{s=0}^{n-1} g_{n}^s \right)
	R_{i} \label{eq:vii_step_1}\\
	&=
	R_i^T \operatorname{diag}\left( 0,0,1 \right) R_i
	=
	\vivi, \label{eq:vii_step_2}
	\end{align}
	where~\eqref{eq:vii_step_1} follows from~\eqref{eq:a_is_generator}, and~\eqref{eq:vii_step_2} follows from Lemma~\ref{lemma:g_n_ks} using $l=1$ which indeed satisfies $(l \bmod n) \neq 0$ for any $n \ge 3$.

\qed
	
	\subsection{Proof of Lemma~\ref{lemma:angle_greater_pi2}}\label{app:Proof of Lemma angle_greater_pi2}
	We need first the following two lemmas.
	\begin{lemma}\label{lemma:before_thm_}
		For any $n \in \mathbb{N}$, $\iInm$, and $s \in [n-1]$,
		\begin{equation}\label{eq:lemma_qii_cn}
		\dotprod{q_{ii}^{(n-s)}}{q_{ii}^{(s)}}
		=
		\frac{\dotprod{R_{i}^{(3)}}{g_{n}^{2s} R_{i}^{(3)}} - \dotprod{R_{i}^{(3)}}{g_{n}^{s} R_{i}^{(3)}}^{2}}{1-\dotprod{R_{i}^{(3)}}{g_{n}^{s} R_{i}^{(3)}}^2},
		\end{equation}
		where $q_{ii}^{(n-s)}$ and $q_{ii}^{(s)}$ are defined in~\eqref{eq:q_ij_s}.
	\end{lemma}
	\begin{proof}
		Let $n \in \mathbb{N}$, $\iInm$, and $s \in [n-1]$. By~\eqref{eq:q_ij_s}
		\begin{equation}\label{eq:qii1_qii3_}
		\dotprod{q_{ii}^{(n-s)}}{q_{ii}^{(s)}}
		=
		\frac{\dotprod{R_{i}^{(3)} \times g_{n}^{n-s} R_{i}^{(3)}}{R_{i}^{(3)} \times g_{n}^{s} R_{i}^{(3)}}}{\left\|R_{i}^{(3)} \times g_{n}^{n-s} R_{i}^{(3)}\right\| \left\|R_{i}^{(3)} \times g_{n}^{s} R_{i}^{(3)}\right\|}
		=
		\frac{\dotprod{R_{i}^{(3)} \times g_{n}^{n-s} R_{i}^{(3)}}{R_{i}^{(3)} \times g_{n}^{s} R_{i}^{(3)}}}{\left\|R_{i}^{(3)} \times g_{n}^{s} R_{i}^{(3)}\right\|^2 },
		\end{equation}
		where the second equality uses the fact that the vector $2$-norm is rotation invariant, and therefore
		$$
		\left\|R_{i}^{(3)} \times g_{n}^{n-s} R_{i}^{(3)}\right\|
		=
		\left\| g_{n}^{s} \left( R_{i}^{(3)} \times g_{n}^{n-s} R_{i}^{(3)} \right) \right\|
		=
		\left\| g_{n}^{s} R_{i}^{(3)} \times R_{i}^{(3)}\right\|
		=
		\left\|R_{i}^{(3)} \times g_{n}^{s} R_{i}^{(3)}\right\|.
		$$
		The nominator in~\eqref{eq:qii1_qii3_} may be simplified by using Lagrange's identity $\dotprod{a \times b}{c \times d} = \dotprod{a}{c}\dotprod{b}{d} - \dotprod{a}{d}\dotprod{b}{c}$ which holds for any $a,b,c,d \in \mathbb{R}^3$, so that
		\begin{align}
		&\dotprod{R_{i}^{(3)} \times g_{n}^{n-s} R_{i}^{(3)}}{R_{i}^{(3)} \times g_{n}^{s} R_{i}^{(3)}}\\
		=&\dotprod{R_{i}^{(3)}}{R_{i}^{(3)}} \dotprod{g_{n}^{n-s} R_{i}^{(3)}}{g_{n}^{s} R_{i}^{(3)}} -\dotprod{R_{i}^{(3)}}{g_{n}^{s} R_{i}^{(3)}} \dotprod{g_{n}^{n-s} R_{i}^{(3)}}{R_{i}^{(3)}} \nonumber \\
		=&\dotprod{g_{n}^{n-s} R_{i}^{(3)}}{g_{n}^{s} R_{i}^{(3)}} - \dotprod{R_{i}^{(3)}}{g_{n}^{s} R_{i}^{(3)}} \dotprod{g_{n}^{n-s} R_{i}^{(3)}}{R_{i}^{(3)}} \label{eq:nom_1}  \\
		=&\dotprod{R_{i}^{(3)}}{g_{n}^{2s} R_{i}^{(3)}} - \dotprod{R_{i}^{(3)}}{g_{n}^{s} R_{i}^{(3)}}^{2},
		\label{eq:nom_2}
		\end{align}
		where~\eqref{eq:nom_1} uses the fact that $\dotprod{R_{i}^{(3)}}{R_{i}^{(3)}}=1$, and~\eqref{eq:nom_2} is because $\left(g_{n}^{n-s}\right)^{T} = g_{n}^{s}$. As for the denominator in~\eqref{eq:qii1_qii3_}, since the magnitude of the cross-product is given by the sine of the angle between its arguments, and since $\left\|R_{i}^{(3)} \right \| = \left\| g_{n}^{s} R_{i}^{(3)} \right\| = 1$, it follows that
		\begin{equation}\label{eq:denom_}
		\left\|R_{i}^{(3)} \times g_{n}^{s} R_{i}^{(3)}\right\|^{2}
		=
		1-\dotprod{R_{i}^{(3)}}{g_{n}^{s} R_{i}^{(3)}}^2.
		\end{equation}
		Plugging-in~\eqref{eq:nom_2} and~\eqref{eq:denom_} into~\eqref{eq:qii1_qii3_} yields~\eqref{eq:lemma_qii_cn} which completes the proof.
	\end{proof}
	
	\begin{lemma}\label{lemma:self_ang_diff_polar}
		For any $\iInm$, let $R_{i}^{(3)}=(\sin\theta_{i} \cos\phi_{i},\sin\theta_{i}\sin\phi_{i},\cos\theta_{i})^T$ be the representation of $R_{i}^{(3)}$ in spherical coordinates for some $\theta_{i} \in [0,\pi)$ and $\phi_{i} \in [0,2\pi)$. Then,
		\begin{subnumcases}{\cos \left( \alpha_{ii}^{(n-1)}-\alpha_{ii}^{(1)} \right)=}
		\frac{\cos^2\theta_{i}-\frac{1}{2}\sin^2\theta_{i}}{1+\cos^2\theta_{i}-\frac{1}{2}\sin^2\theta_{i}}, & if $n=3$, \label{eq:lemma_c3}\\
		\frac{\cos^2\theta_{i}-1}{\cos^2\theta_{i}+1}, & if $n=4$. \label{eq:lemma_c4}
		\end{subnumcases}
	\end{lemma}
	\begin{proof}
		Let $\iInm$. For any $n \in \mathbb{N}, \ n>1$,
		\begin{align}
		\cos \left( \alpha_{ii}^{(n-1)} - \alpha_{ii}^{(1)} \right)
		&=
		\dotprod{R_{i}^{T} q_{ii}^{(n-1)}}{R_{i}^{T} q_{ii}^{(1)}}
		=
		\dotprod{q_{ii}^{(n-1)}}{q_{ii}^{(1)}} \label{eq:step_1} \\
		&=
		\frac{\dotprod{R_{i}^{(3)}}{g_{n}^{2} R_{i}^{(3)}} - \dotprod{R_{i}^{(3)}}{g_{n} R_{i}^{(3)}}^{2}}{1-\dotprod{R_{i}^{(3)}}{g_{n} R_{i}^{(3)}}^2}, \label{eq:step_2}
		\end{align}
		where~\eqref{eq:step_1} is due to~\eqref{eq:alpha_ii_q_ii} and because $R_{i} R_{i}^{T} = I$, and~\eqref{eq:step_2} is the application of Lemma~\ref{lemma:before_thm_} with $s=1$. By~\eqref{def:g}, we get by a direct calculation that for any $s \in \mathbb{N}$,
		\begin{equation*}
		g_{n}^{s}  R_{i}^{(3)}=\left(  \sin\theta_{i}\cos\left( \phi_{i} + 2s\pi/n \right), \sin\theta_{i}\sin\left( \phi_{i} + 2s\pi/n \right), \cos\theta_{i} \right)^{T},
		\end{equation*}
		from which it follows that
		\begin{equation}\label{eq:gsRi_spher_coord_}
		\dotprod{R_{i}^{(3)}}{g_{n}^{s} R_{i}^{(3)}} = \cos^2\theta_{i} + \sin^2\theta_{i}\cos \left( 2s\pi/n \right), \quad \forall s \in \mathbb{N}.
		\end{equation}
		We first prove~\eqref{eq:lemma_c3} (i.e., the case $n=3$). To this end,~\eqref{eq:step_1}--\eqref{eq:step_2} become
		\begin{equation}\label{eq:tmpp_c3}
		\cos \left( \alpha_{ii}^{(2)} - \alpha_{ii}^{(1)} \right)
		=
		\frac{\dotprod{R_{i}^{(3)}}{g_{3}^{2} R_{i}^{(3)}} - \dotprod{R_{i}^{(3)}}{g_{3} R_{i}^{(3)}}^{2}}{1-\dotprod{R_{i}^{(3)}}{g_{3} R_{i}^{(3)}}^2}.
		\end{equation}
		Since $g_{3}^{T} = g_{3}^{2}$ it follows that
		$$
		\dotprod{R_{i}^{(3)}}{g_{3}^{2} R_{i}^{(3)}}
		=
		\dotprod{R_{i}^{(3)}}{g_{3}^{T} R_{i}^{(3)}}
		=
		\dotprod{g_{3} R_{i}^{(3)}}{ R_{i}^{(3)}}
		=
		\dotprod{ R_{i}^{(3)}}{ g_{3} R_{i}^{(3)}}.
		$$
		As such,~\eqref{eq:tmpp_c3} reduces to
		\begin{align}
		\frac{\dotprod{R_{i}^{(3)}}{g_{3} R_{i}^{(3)}} - \dotprod{R_{i}^{(3)}}{g_{3} R_{i}^{(3)}}^{2}}{1-\dotprod{R_{i}^{(3)}}{g_{3} R_{i}^{(3)}}^2}
		&=
		\frac{\dotprod{R_{i}^{(3)}}{g_{3} R_{i}^{(3)}}
			\left(
			1-\dotprod{R_{i}^{(3)}}{g_{3} R_{i}^{(3)}}
			\right)}{\left(
			1+\dotprod{R_{i}^{(3)}}{g_{3} R_{i}^{(3)}}
			\right)\left(
			1-\dotprod{R_{i}^{(3)}}{g_{3} R_{i}^{(3)}}
			\right)} \nonumber \\
		&=
		\frac{\dotprod{R_{i}^{(3)}}{g_{3} R_{i}^{(3)}}}{1+\dotprod{R_{i}^{(3)}}{g_{3} R_{i}^{(3)}}}. \label{lemma:last_step}
		\end{align}
		Next, applying~\eqref{eq:gsRi_spher_coord_} with $n=3$ and $s=1$ gives
		\begin{equation*}
		\dotprod{R_{i}^{(3)}}{g_{3} R_{i}^{(3)}}
		=
		\cos^2\theta_{i} + \sin^2\theta_{i} \cos \left( 2\pi/3 \right)
		=
		\cos^2\theta_{i} - \frac{1}{2}\sin^2\theta_{i},
		\end{equation*}
		and plugging this in~\eqref{lemma:last_step} yields~\eqref{eq:lemma_c3}.
		
		We next prove~\eqref{eq:lemma_c4}. Applying~\eqref{eq:gsRi_spher_coord_} with $n=4$ and $s=1$, and with $n=4$ and $s=2$ yields
		\begin{equation*}
		\dotprod{R_{i}^{(3)}}{g_{4} R_{i}^{(3)}} = \cos^2\theta_{i}, \quad
		\dotprod{R_{i}^{(3)}}{g_{4}^{2} R_{i}^{(3)}} = \cos^2\theta_{i} - \sin^2\theta_{i}.
		\end{equation*}
		Plugging this in~\eqref{eq:step_2} yields
		\begin{equation}\label{eq:tmp4_}
		\frac{\cos^2\theta_{i} - \sin^2\theta_{i} - \cos^4\theta_{i}}{1-\cos^4\theta_{i}}
		=
		\frac{\left( \cos^2\theta_{i}-1 \right) \left( 1-\cos^2\theta_{i} \right)}{\left( 1+\cos^2\theta_{i} \right)\left( 1-\cos^2\theta_{i} \right)}
		=
		\frac{\cos^2\theta_{i}-1}{\cos^2\theta_{i}+1},
		\end{equation}
		which proves~\eqref{eq:lemma_c4}.
	\end{proof}
	
We are now ready to prove Lemma~\ref{lemma:angle_greater_pi2}.
	\begin{proof}
		We begin by proving~\eqref{eq:angle_greater_pi3}. Consider the function $f \colon \mathbb{R} \mapsto \mathbb{R}$ where
		\begin{equation}\label{def:f_theta}
		f\left( \theta \right)
		= \frac{\cos^2\theta-\frac{1}{2}\sin^2\theta}{1+\cos^2\theta-\frac{1}{2}\sin^2\theta}, \quad \theta \in \mathbb{R}.
		\end{equation}
		Then, in light of~\eqref{eq:lemma_c3}, it suffices to show that $\arccos f\left( \theta \right) \geq \pi/3$ for any $\theta \in \mathbb{R}$. As it may readily be verified, $f$ is periodic with period $\pi$, and therefore $\arccos f$ is periodic with the same period. In addition, $f\left( \pi/2 + \theta \right) = f\left( \pi/2 - \theta \right)$ for any $\theta \in \mathbb{R}$, see Figure~\ref{fig:acosf}. As a result, it suffices to show that $\arccos f\left( \theta \right) \geq \pi/3$ for any $\theta \in [0,\pi/2]$. To this end, by the chain-rule,
		\begin{equation}
		\begin{aligned}
		\frac{d \left( \arccos f \right)}{d \theta}
		&=
		\frac{-1}{\sqrt{1-f^2}}\frac{df}{d \theta}
		=
		\frac{-1}{\sqrt{1-\left( \frac{\cos^2\theta-\frac{1}{2}\sin^2\theta}{1+\cos^2\theta-\frac{1}{2}\sin^2\theta}\right)^2}}
		\frac{-3\sin \theta \cos \theta}{\left( 1+\cos^2\theta-\frac{1}{2}\sin^2\theta \right)^2}\\
		&=
		\frac{1+\cos^2\theta-\frac{1}{2}\sin^2\theta}{\sqrt{3}\cos\theta}
		\frac{3\sin \theta \cos \theta}{\left( 1+\cos^2\theta-\frac{1}{2}\sin^2\theta \right)^2}\\
		&=
		\frac{\sqrt{3}\sin\theta}{1+\cos^2\theta-\frac{1}{2}\sin^2\theta} \ge 0,
		\end{aligned}
		\end{equation}
		for any $\theta \in [0,\pi/2]$. As such, $\arccos f$ is non-decreasing in $[0,\pi/2]$. Thus, since $\arccos f$ is continuous, it follows that for any $\theta \in [0,\pi/2]$,
		\begin{align*}
		\arccos f\left( \theta \right)
		&=
		\arccos \left( \frac{\cos^2\theta-\frac{1}{2}\sin^2\theta}{1+\cos^2\theta-\frac{1}{2}\sin^2\theta} \right)\\
		&\ge
		\arccos \left( \frac{\cos^2 0-\frac{1}{2}\sin^2 0}{1+\cos^2 0-\frac{1}{2}\sin^2 0} \right)\\
		&=
		\arccos \frac{1}{2}
		=
		\frac{\pi}{3},
		\end{align*}
		which proves~\eqref{eq:angle_greater_pi3}, see Figure~\ref{fig:acosf}. Similarly, in order to prove~\eqref{eq:angle_greater_pi2}, consider the function $g \colon \mathbb{R} \mapsto \mathbb{R}$ where
		\begin{equation}\label{def:g_theta}
		g\left( \theta \right) = \frac{\cos^2\theta-1}{\cos^2\theta+1}, \quad \theta \in \mathbb{R}.
		\end{equation}
		In light of~\eqref{eq:lemma_c4}, it suffices to show that $\arccos g\left( \theta \right) \geq \pi/2$ for any $\theta \in \mathbb{R}$. Since $g$ is periodic with period $\pi$, it follows that $\arccos g$ is periodic with the same period. In addition, $g\left( \pi/2 + \theta \right) = g\left( \pi/2 - \theta \right)$ for any $\theta \in \mathbb{R}$, see Figure~\ref{fig:acosg}. As a result, it suffices to show that $\arccos g\left( \theta \right) \geq \pi/2$ for any $\theta \in [0,\pi/2]$. To this end, by the chain-rule
		\begin{align*}
		\frac{d \left( \arccos g \right)}{d \theta}
		&=
		\frac{-1}{\sqrt{1-g^2}}\frac{dg}{d \theta}
		=
		\frac{-1}{\sqrt{1-\left( \frac{\cos^2\theta-1}{\cos^2\theta+1} \right)^2}}
		\frac{-4\sin \theta \cos \theta }{\left( \cos^2\theta+1 \right)^{2}}\\
		&=
		\frac{\cos^2 \theta + 1}{2 \cos \theta}
		\frac{4 \sin \theta \cos \theta}{\left( \cos^2 \theta + 1 \right)^2}
		=
		\frac{2 \sin \theta }{\cos^2 \theta + 1}.
		\end{align*}
		Thus, since $\frac{2 \sin \theta }{\cos^2 \theta + 1} \ge 0$ for any $\theta \in [0,\pi/2]$, we conclude that $\arccos g$ is non-decreasing in $[0,\pi/2]$, and since $\arccos g$ is continuous, it follows that for any $\theta \in [0,\pi/2]$,
		$$
		\arccos g\left( \theta \right)
		=
		\arccos \left( \frac{\cos^2\theta-1}{\cos^2\theta+1}\right)
		\ge
		\arccos \left( \frac{\cos^2 0 - 1}{\cos^2 0 + 1}\right)
		=
		\arccos 0
		=
		\frac{\pi}{2},
		$$
		which proves~\eqref{eq:angle_greater_pi2}, see Figure~\ref{fig:acosg}.
	\end{proof}
	
	\begin{figure}
		\centering
		\subfloat[The graph of $\arccos f$, see~\eqref{def:f_theta}.]{
			\includegraphics[width=2.5in]{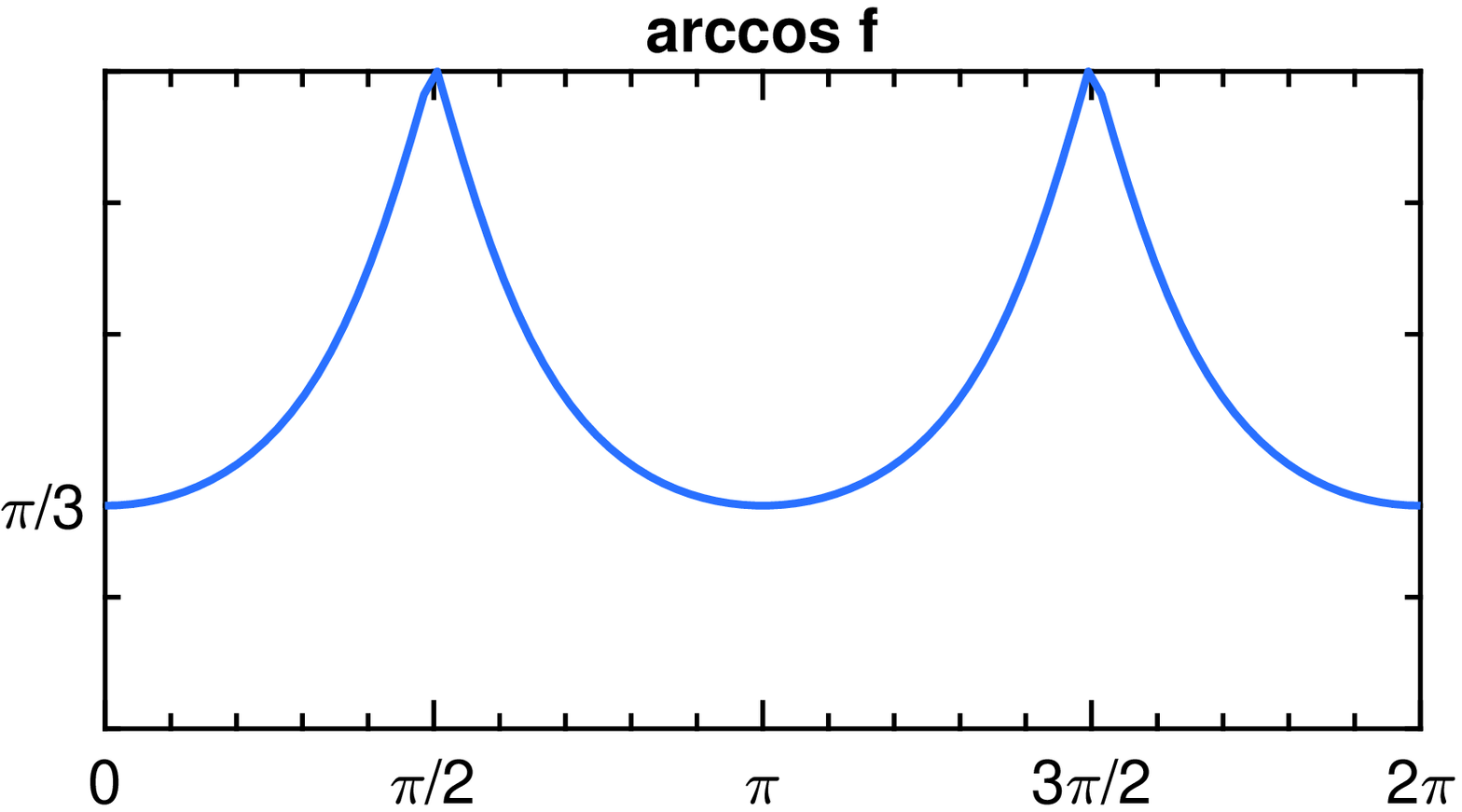}
			\label{fig:acosf}
		}
		\subfloat[The graph of $\arccos g$, see~\eqref{def:g_theta}.]{
			\centering
			\includegraphics[width=2.5in]{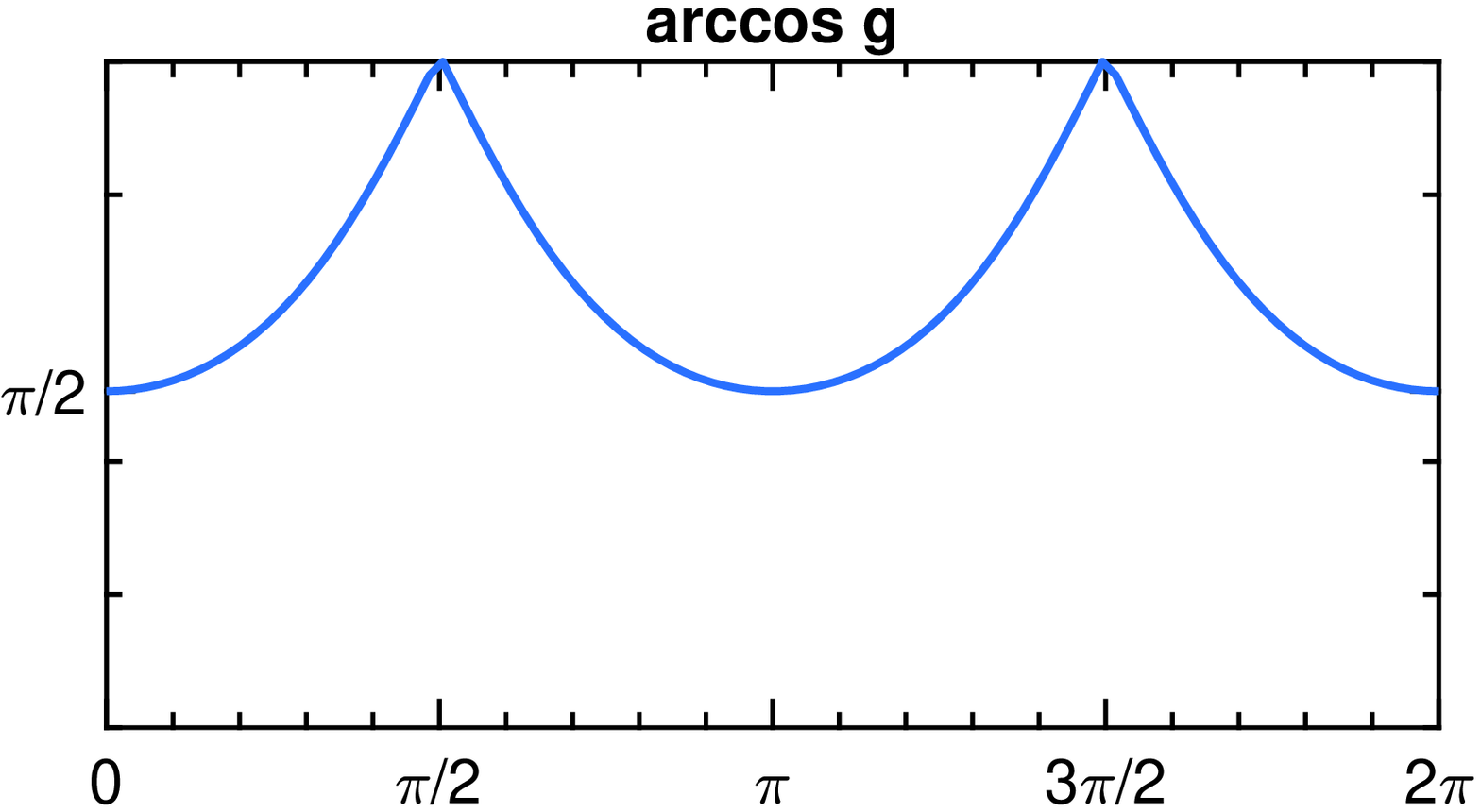}
			\label{fig:acosg}	
		}
		\caption{The graphs of the functions used in Lemma~\ref{lemma:angle_greater_pi2}.}
	\end{figure}

\subsection{Proof of Lemma~\ref{lemma:main}}\label{app:Proof of Lemma main}

	Let $\iInm$. We begin by proving~\eqref{eq:cos_gamma_practice__c3} (i.e., the case where $n=3$). It holds that (in fact for any $n \in \mathbb{N}$ and not just for $n=3$),
	\begin{align}
	\cos \left( \alpha_{ii}^{(n-1)} - \alpha_{ii}^{(1)} \right)
	&=
	\dotprod{R_{i}^{T} q_{ii}^{(n-1)}}{R_{i}^{T} q_{ii}^{(1)}}
	=
	\dotprod{q_{ii}^{(n-1)}}{q_{ii}^{(1)}} \label{eq:step_1_} \\
	&=
	\frac{\dotprod{R_{i}^{(3)}}{g_{n}^{2} R_{i}^{(3)}} - \dotprod{R_{i}^{(3)}}{g_{n} R_{i}^{(3)}}^{2}}{1-\dotprod{R_{i}^{(3)}}{g_{n} R_{i}^{(3)}}^2}, \label{eq:step_2_}
	\end{align}
	where the first equality in~\eqref{eq:step_1_} is due to~\eqref{eq:alpha_ii_q_ii}, the second equality in~\eqref{eq:step_1_} is because $R_{i} R_{i}^{T} = I$, and~\eqref{eq:step_2_} is the result of applying Lemma~\ref{lemma:before_thm_} with $s=1$. For $n=3$,~\eqref{eq:step_1_}--\eqref{eq:step_2_} become
	\begin{equation}\label{eq:tmp_c3}
	\cos \left( \alpha_{ii}^{(2)} - \alpha_{ii}^{(1)} \right)
	=
	\frac{\dotprod{R_{i}^{(3)}}{g_{3}^{2} R_{i}^{(3)}} - \dotprod{R_{i}^{(3)}}{g_{3} R_{i}^{(3)}}^{2}}{1-\dotprod{R_{i}^{(3)}}{g_{3} R_{i}^{(3)}}^2}.
	\end{equation}
	Next, by the projection slice theorem, it follows that
	\begin{equation}
	\begin{gathered}
	\left \langle R_{i}^{(3)},g_{3} R_{i}^{(3)} \right \rangle
	=
	\cos \gamma_{ii}^{(1)} = \cos \gamma_{ii}, \\
	\left \langle R_{i}^{(3)},g_{3}^{2} R_{i}^{(3)} \right \rangle
	=
	\cos \gamma_{ii}^{(2)} = \cos \gamma_{ii},
	\end{gathered}
	\end{equation}
	where we have used the fact that $\gamma_{ii}^{(1)} = \gamma_{ii}^{(2)}$ (see text after~\eqref{eq:gamma_n_minus_one}). Thus,~\eqref{eq:tmp_c3} may be written as
	\begin{equation}\label{eq:tmpp_c3_}
	\cos \left( \alpha_{ii}^{(2)} - \alpha_{ii}^{(1)} \right)
	=
	\frac{\cos \gamma_{ii} - \cos^{2} \gamma_{ii}}{1-\cos^{2} \gamma_{ii}}
	=
	\frac{\cos \gamma_{ii}}{1+\cos \gamma_{ii}},
	\end{equation}
	and solving~\eqref{eq:tmpp_c3_} for $\cos \gamma_{ii}$ yields~\eqref{eq:cos_gamma_practice__c3}.
	
	We next prove~\eqref{eq:cos_gamma_practice_} (i.e., the case where $n=4$). Let $\iInm$, and let $R_{i}^{(3)}=(\sin\theta_{i}\cos\phi_{i},\sin\theta_{i}\sin\phi_{i},\cos\theta_{i})^T$ be the representation of $R_{i}^{(3)}$ in spherical coordinates for some $\theta_{i} \in [0,\pi)$ and $\phi_{i} \in [0,2\pi)$. On the one hand, by the projection slice theorem,
	\begin{equation}\label{eq:left_c4_pre}
	\cos \gamma_{ii}
	=
	\dotprod{R_{i}^{(3)}}{g_{4} R_{i}^{(3)}}.
	\end{equation}
	Next, since $g_{4} = R_z(\pi/2)$, we get by a direct calculation that
	$$
	g_{4} R_{i}^{(3)}=\left( \sin\theta_{i}\cos\left( \phi_{i} + \pi/2 \right), \sin\theta_{i}\sin\left( \phi_{i} + \pi/2 \right), \cos\theta_{i} \right)^{T},
	$$
	and therefore by a direct calculation it follows that
	\begin{equation}\label{eq:left_c4}
	\dotprod{R_{i}^{(3)}}{g_{4} R_{i}^{(3)}}
	=
	\cos^2\theta_{i}.
	\end{equation}
	Thus, from~\eqref{eq:left_c4_pre} and~\eqref{eq:left_c4} we get that
	\begin{equation}\label{eq:left_c4_post}
	\cos \gamma_{ii}
	=
	\cos^2\theta_{i}.
	\end{equation}
	On the other hand, by~\eqref{eq:lemma_c4} in Lemma~\ref{lemma:self_ang_diff_polar}
	\begin{equation}\label{eq:right}
	\frac{1 + \cos \left(  \alpha_{ii}^{(n-1)} - \alpha_{ii}^{(1)} \right)}{1 - \cos \left(  \alpha_{ii}^{(n-1)} - \alpha_{ii}^{(1)}\right)}
	=
	\frac{1 + \cos \left(  \alpha_{ii}^{(3)} - \alpha_{ii}^{(1)} \right)}{1 - \cos \left(  \alpha_{ii}^{(3)} - \alpha_{ii}^{(1)}\right)}
	=
	\frac{1+\frac{\cos^2\theta_{i}-1}{\cos^2\theta_{i}+1}}{1-\frac{\cos^2\theta_{i}-1}{\cos^2\theta_{i}+1}}
	=
	\frac{2\cos^2\theta_{i}}{2}
	=
	\cos^2\theta_{i}.
	\end{equation}
	Equation~\eqref{eq:cos_gamma_practice_} now follows from~\eqref{eq:left_c4_post} and~\eqref{eq:right}.

\qed

\section{Justification for the expressions in~\eqref{eq:eight_poss}}\label{app:Justification of expressions C4}
We prove that the expressions listed in~\eqref{eq:eight_poss} are equivalent to the expressions in~\eqref{eq:eight_poss_c4_orig}. To this end, notice that since $g_{4}^{4} = I$, it follows that for any $s_{ij} \in \{0,1,2,3\}$,
\begin{equation}\label{eq:justification_c4}	
\begin{aligned}
&\left( R_{i}^T g_{4} R_{i} \right)^0
\left( R_{i}^{T} g_{4}^{s_{ij}} R_{j} \right)
\left( R_{j}^T g_{4} R_{j} \right)^0
=
\left( R_{i}^T g_{4}^{3} R_{i} \right)^0
\left( R_{i}^{T} g_{4}^{s_{ij}} R_{j} \right)
\left( R_{j}^T g_{4}^{3} R_{j} \right)^0
=
R_{i}^{T} g_{4}^{s_{ij}} R_{j}	
\\
&\left( R_{i}^T g_{4} R_{i} \right)^1
\left( R_{i}^{T} g_{4}^{s_{ij}} R_{j} \right)
\left( R_{j}^T g_{4} R_{j} \right)^1
=
\left( R_{i}^T g_{4}^{3} R_{i} \right)^1
\left( R_{i}^{T} g_{4}^{s_{ij}} R_{j} \right)
\left( R_{j}^T g_{4}^{3} R_{j} \right)^1
=
R_{i}^{T} g_{4}^{s_{ij}+2} R_{j}
\\
&\left( R_{i}^T g_{4} R_{i} \right)^2
\left( R_{i}^{T} g_{4}^{s_{ij}} R_{j} \right)
\left( R_{j}^T g_{4} R_{j} \right)^2
=
\left( R_{i}^T g_{4}^{3} R_{i} \right)^2
\left( R_{i}^{T} g_{4}^{s_{ij}} R_{j} \right)
\left( R_{j}^T g_{4}^{3} R_{j} \right)^2
=
R_{i}^{T} g_{4}^{s_{ij}} R_{j}
\\
&\left( R_{i}^T g_{4} R_{i} \right)^3
\left( R_{i}^{T} g_{4}^{s_{ij}} R_{j} \right)
\left( R_{j}^T g_{4} R_{j} \right)^3
=
\left( R_{i}^T g_{4}^{3} R_{i} \right)^3
\left( R_{i}^{T} g_{4}^{s_{ij}} R_{j} \right)
\left( R_{j}^T g_{4}^{3} R_{j} \right)^3
=
R_{i}^{T} g_{4}^{s_{ij}+2} R_{j}.
\end{aligned}
\end{equation}
As such, if for example, $\Rii = R_{i}^T g_{4} R_{i}$ and $\Rjj = R_{j}^T g_{4} R_{j}$, or for example, $\Rii = R_{i}^T g_{4}^3 R_{i}$ and $\Rjj = R_{j}^T g_{4}^3 R_{j}$, then setting $\Riitilde = \Rii$, and $\mu_{i}=\mu_{j}=0$ in~\eqref{eq:eight_poss_c4_orig}, we get using~\eqref{eq:justification_c4} that
$$
\frac{1}{4} \sum_{s=0}^{3}
( J^{\mu_{i}} \Riitilde  J^{\mu_{i}} )^{s}
\Rij
( J^{\mu_{j}} \Rjj J^{\mu_{j}} )^{s}
=
\frac{1}{2} \left( R_{i}^{T} g_{4}^{s_{ij}} R_{j} + R_{i}^{T} g_{4}^{s_{ij}+2} R_{j}\right)
=
\frac{1}{2} \left( \Rij+\Rii\Rij\Rjj \right),
$$
which is (twice) the first expression in~\eqref{eq:eight_poss}. Alternatively, if for example $\Rii = R_{i}^T g_{4} R_{i}$ and $\Rjj = R_{j}^T g_{4}^{3} R_{j}$, or for example $\Rii = R_{i}^T g_{4}^{3} R_{i}$ and $\Rjj = R_{j}^T g_{4} R_{j}$, then by~\eqref{eq:justification_c4}, setting $\Riitilde = \Riitrans$ and $\mu_{i}=\mu_{j}=0$ in~\eqref{eq:eight_poss_c4_orig} yields (twice) the fifth expression in~\eqref{eq:eight_poss}. Other cases are similar and correspond to cases where either $\Rii$ has a spurious $J$, or $\Rjj$ has a spurious $J$, or both have a spurious $J$, in which case setting in~\eqref{eq:eight_poss_c4_orig} (respectively) either $\mu_{i}=1$, or $\mu_{j}=1$, or $\mu_{i}=\mu_{j}=1$  would yield each of the remaining expressions listed in~\eqref{eq:eight_poss}.

\end{appendices}

\section*{Acknowledgments}
This research was supported by the European Research Council
(ERC) under the European Union’s Horizon 2020 research and innovation programme (grant
agreement 723991 - CRYOMATH) and by Award Number R01GM090200 from the NIGMS.

\bibliographystyle{plain}
\bibliography{Cn}

\end{document}